\documentclass[aps,prx,amsmath,amssymb,twocolumn,10pt,superscriptaddress,a4paper]{revtex4-2}

\usepackage[T1]{fontenc}
\usepackage{lmodern}
\usepackage{bbm,physics}
\usepackage{graphicx}
\usepackage{amsmath, amssymb, mathtools}
\usepackage{amsthm}
\usepackage{physics}
\usepackage{xcolor}
\usepackage{braket}
\usepackage[ruled]{algorithm2e}

\usepackage{tikz}
\usepackage{pgfplots}\pgfplotsset{compat=1.8}
\usetikzlibrary{positioning, shapes.multipart, arrows.meta,decorations.pathreplacing, shapes, calc,fadings,3d}

\usetikzlibrary{math}
\usepackage{ragged2e}
\justifying
\usepackage{apptools}

\usepackage{hyperref}
\usepackage{cleveref}
\hypersetup{citecolor=blue}
\AtAppendix{\counterwithin{theorem}{section}}

\newtheorem{theorem}{Theorem}
\newtheorem{lemma}[theorem]{Lemma}
\newtheorem{proposition}[theorem]{Proposition}

\newtheorem{remark}[theorem]{Remark}

\theoremstyle{definition}
\newtheorem{definition}[theorem]{Definition}
\newtheorem{example}[theorem]{Example}
\newtheorem*{theorem*}{Theorem}
\newtheorem*{corollary*}{Corollary}

\newcommand{\R}{\mathbbm{R}}
\newcommand{\N}{\mathbbm{N}}
\newcommand{\IC}{\mathbbm{C}}
\newcommand{\IN}{\mathbbm{N}}
\newcommand{\IR}{\mathbbm{R}}
\newcommand{\FF}{\mathfrak{F}}
\newcommand{\sfa}{\mathsf{a}}

\newcommand{\cO}{\mathcal{O}}

\renewcommand{\i}{{\mathsf{i}}}
\renewcommand{\j}{{\mathbf{j}}}
\renewcommand{\k}{{\mathbf{k}}}

\newcommand{\chr}{\mathbf 1}

\newcommand{\Id}{\mathbbm{1}}

\newcommand{\Jf}{\mathfrak{J}}

\newcommand{\supp}{\operatorname{supp}}
\newcommand{\dist}{\operatorname{dist}}

\newcommand{\Aj}[2]{A^{#2}_{#1}}

\newcommand{\chiRx}[2]{\chr_{\{\abs{\cdot - #2}\ge#1\}}}

\newcommand{\Cobt}[1]{\tilde C^{#1}_{\mathrm{ob}}}

{\tiny }

\newcommand{\cs}[1]{\braket{#1}}

\newcommand{\IP}{\mathbb{P}}
\newcommand{\IE}{\mathbb{E}}
\renewcommand{\d}{\,\mathrm{d}}
\renewcommand{\j}{{\mathbf{j}}}
\renewcommand{\k}{{\mathbf{k}}}
\newcommand{\q}{\mathbf{q}}

\newcommand{\chip}{\chi^\perp}
\newcommand{\chit}{\widetilde \chi}
\newcommand{\chitp}{ \chit^\perp}
\newcommand{\gd}{\gamma_d}

\renewcommand{\:}{\colon}

\newcommand{\dG}{\mathsf{d}\Gamma}

\newcommand{\kb}{\mathbf{k}}

\newcommand{\Dd}{p}
\newcommand{\dDd}[3]{D_{#1} \Dd_{#2}^{#3}}

\newcommand{\nn}[1]{\norm{#1}}
\renewcommand{\d}{\mathsf d}
\DeclareMathOperator*{\argmax}{arg\,max}
\DeclareMathOperator*{\argmin}{arg\,min}

\crefname{equation}{Eq.}{Eqs.}
\Crefname{equation}{Eq.}{Eqs.}

\begin{document}
\title{Learning Coulomb Potentials and Beyond with Free Fermions in Continuous Space}

\author{\begingroup
	\hypersetup{urlcolor=black}
	\href{https://orcid.org/0000-0003-4796-7633}{Andreas Bluhm
		\endgroup}
}
\email{andreas.bluhm@univ-grenoble-alpes.fr}
\affiliation{Univ.\ Grenoble Alpes, CNRS, Grenoble INP, LIG, 38000 Grenoble, France}

\author{\begingroup
	\hypersetup{urlcolor=black}
	\href{https://orcid.org/0000-0001-6459-8046}{Marius Lemm
		\endgroup}
}
\email{marius.lemm@uni-tuebingen.de}
\affiliation{Department of Mathematics, University of T\"{u}bingen, 72076 T\"{u}bingen, Germany}

\author{\begingroup
	\hypersetup{urlcolor=black}
	\href{https://orcid.org/0000-0003-0471-2745}{Tim Möbus
		\endgroup}
}
\email{moebustim@gmail.com}
\affiliation{Department of Mathematics, University of T\"{u}bingen, 72076 T\"{u}bingen, Germany}
\affiliation{Department of Applied Mathematics and Theoretical Physics, University of Cambridge, Cambridge CB3 0WA, United Kingdom}

\author{\begingroup
	\hypersetup{urlcolor=black}
	\href{https://orcid.org/0000-0002-0800-2436}{Oliver Siebert
		\endgroup}
}
\email{osiebert@ucdavis.edu}
\affiliation{Department of Mathematics, University of California, Davis, CA 95616-8633, USA}

\begin{abstract}
	\setlength{\parindent}{0pt}
	The first-principles formulation of quantum mechanics relevant for quantum chemistry and trapped quantum gases involves particles in the continuous space $\mathbb R^d$. We present a unified framework and modular algorithm for learning external potentials $V$ with free-fermion models in the continuum. Compared to the lattice-based approaches, the continuum presents new mathematical challenges: the state space is infinite-dimensional and the Hamiltonian contains the Laplacian, which is unbounded in the continuum and produces an unbounded speed of information propagation. We address these through novel optimization methods and information-propagation bounds in combination with a priori regularity assumptions on the external potential. The resulting algorithm provides a unified and robust approach to learn parametric interactions (e.g., Coulomb potentials or periodic potentials) and general smooth functions. Our results lay the foundation for a scalable and generalizable toolkit to learn   Hamiltonians in continuous space. 
\end{abstract}

\maketitle

\section{Introduction}\label{sec:intro}
Hamiltonian learning starts from an experimentally observable quantum time evolution and aims to characterize the underlying Hamiltonian operator. This is a central  \textit{inverse problem} of quantum mechanics and involves inferring the parameters defining the Hamiltonian in a chosen representation, from experimental data.
Hamiltonian learning plays a central role, as it is relevant both for improved theoretical understanding and for the practical calibration of quantum devices.  

Hamiltonian learning has been extensively studied in quantum spin systems \cite{Bakshi2024,StilckFranca.2024,Huang.2023,Zubida.2021,Bairey2019,Caro2024} and lattice fermions \cite{hayden}, where the local Hilbert space dimension is finite and the Hamiltonian is a bounded operator. For example \cite{Huang.2023,Bakshi2024} achieved the information-theoretic optimal Heisenberg scaling for lattice qubit systems and \cite{StilckFranca.2024},  achieved the standard quantum limit for lattice qubit systems under weak assumptions on the experimental setup. Extensions to bosonic lattice models were performed in \cite{mobus2023dissipation,TongBosons2024,moebus2025heisenberg}.

The first-principle setting of quantum theory is that of particles  living in \textit{continuous space} $\mathbb R^d$. In particular, quantum chemistry usually deals with particles unconfined to a lattice. Simulating such systems has recently gained attention \cite{Metz.2024}.
However, a theory of Hamiltonian learning for particles living in  $\mathbb R^d$ is lacking. The reason is that working on $\mathbb R^d$ creates new fundamental  challenges: First, in contrast to quantum spin systems, our systems are  described by infinite-dimensional state spaces and unbounded interactions. Second, in contrast to bosonic lattice gases (also known as ``CV'' systems), even the kinetic energy of a single particle, which is represented by the Laplacian $-\Delta$, is now unbounded.

In this work, we develop the first mathematical framework, as well as a concrete algorithm, for Hamiltonian learning of external potentials in continuous space. Working in the continuum allows for a natural representation of quantum states and dynamics, free from the discretization artifacts sometimes introduced by ad-hoc lattice models. For example, continuous space preserves rotational and full translational symmetries and supports all plane wave solutions, providing a universal framework to interpret experimental data. Furthermore, learning in continuous space avoids artificial complications such as  fermion doubling and lattice-induced boundary effects, making it particularly well-suited for modeling systems where spatial resolution and smooth potential landscapes are essential. 
Concretely, we consider a system of non-interacting fermions each described by the Hamiltonian
\begin{align}\label{eq:onebody-hamiltonian}
	h = -\Delta + V \, ,
\end{align}
with unknown external potential $V:\mathbb R^d\to \mathbb R$ to be determined. Our learning algorithm can handle both parameter learning when $V$ is taken from a specific class of functions (e.g.,\ Coulomb potentials) or learning general functions $V$ of sufficient regularity.

For example, we may consider a system of $K$ ions of unknown electric charges $\lambda_1,\ldots,\lambda_K\geq 0$ and unknown positions $y_1,\ldots,y_K\in \mathbb R^3$. This creates the unknown external potential
\begin{equation*}
	V(x) = \sum_{k=1}^K \frac{\lambda_k}{\|x - y_k\|}.
\end{equation*}
We then study the following learning task: Assume we can prepare an initial state of non-interacting fermions subjected to $-\Delta+V$, perform time evolution and measure. Using this information, we aim to estimate  the parameters $\lambda_1,\ldots,\lambda_K\geq 0$ and positions $y_1,\ldots,y_K\in \mathbb R^3$. The motion of the ions can be neglected by the time-dependent Born-Oppenheimer approximation; see, e.g., \cite{hagedorn1980time}.
The assumption that the fermions are non-interacting is only approximate when they are electrons, as these exhibit Coulomb repulsion amongst each other. The corrections are small, e.g., for the case of strongly charged ions.   Alternatively, our algorithm can be modified to only involve a single pair of distant fermions each time, so that their interaction is heavily suppressed, but this requires repeating the experiment more often.

 The central analytical challenge of our work is the regularity constraints imposed by the Laplacian.
Learning the ion positions in the Coulomb potential introduces another novel challenge we need to address:  nonlinear parameter dependencies. Addressing these difficulties requires nonlinear numerical algorithms as well as a careful perturbation analysis of the Coulomb potential.


Beyond the Coulomb case, we also consider learning broader classes of potentials, including smooth potentials and Fourier potentials, again by studying the time evolution of suitably prepared non-interacting fermions. All of these models can be addressed within a unified data acquisition framework using our modularly designed algorithm. Here, parallelization exponentially reduces the sample complexity of the algorithm and allows for an efficient approximation of the considered potential. Parallelization is enabled by locality estimates, which are more challenging in the continuum and which we obtain through the IMS localization formula \cite[Theorem 3.2]{cycon1987schrodinger} and information-propagation bounds for continuum systems \cite{hinrichs2023lieb}. Also a suitable choice of physical states allows to avoid introducing any reference system imposed by the superselection rule as is common for fermionic learning protocols \cite{hayden}. Finally, the  main reason why we choose to work with free (non-interacting) fermions is precisely that we aim for parallelization, which requires strong control on information propagation (i.e., Lieb-Robinson bounds). In the continuous space $\mathbb R^d$, Lieb-Robinson bounds are currently underdeveloped. Using \cite{gebert2020lieb,hinrichs2023lieb}, it  would be possible to treat smeared-out two-body interactions, but this would still lead to bounds that grow with system size, because the bounds  treat the time evolution of individual fermionic creation and annihilation operators, but need to be applied to products of those. We emphasize that possible future improvements on information propagation bounds in $\mathbb R^d$ can be incorporated into our modular algorithm to cover also interacting fermions uniformly in the system size.

In detail, we first discuss the preliminaries and the protocol for acquiring the necessary data in \Cref{sec:prelim}. Based on this data, we present a numerical algorithm to learn single and few-Coulomb potentials in \Cref{sec:Coulomb}. Afterwards, the same data is used to learn any Lipschitz function or potentials that can be well-approximated by a set of linearly independent functions, like trigonometric functions, as shown in \Cref{sec:general-learning}. Finally, detailed proofs and possible technical extensions are outlined in \Cref{appx-sec:Coulomb}–\ref{appx-sec:technical-details}.

\section{Preliminaries and protocol}\label{sec:prelim}
As mentioned in the previous section, our Hamiltonian learning algorithm addresses the Schr\"{o}dinger operator acting on the infinite-dimensional space $L^2(\mathbb{R}^d)$,
\begin{align*}
	h = -\Delta + V
\end{align*}
for various classes of external potentials $V:\mathbb R^d\to\mathbb R$. Our objective is to approximate the unknown external potential $V$ with precision $\varepsilon$ and success probability $(1-\delta)$, by means of the time evolution of the fermionic many-body Hamiltonian in second quantization (see \Cref{appx-sec:technical-details} for the mathematical details), i.e.,
\begin{align}
	\label{eq:Hamiltonian physically}
	H = \int_{\mathbb{R}^d} \left(\nabla a_x^\dagger \nabla a_x + V(x) a_x^\dagger a_x\right)\, \mathrm{d}x\,.
\end{align}
To achieve this, we assume that it is possible to prepare locally defined initial states $\psi_0$ and to measure the number of fermions in predefined regions of the evolved state $e^{-it H} \psi_0$ after short time intervals. By this procedure, we obtain estimators for local averages of $V$ in small boxes (see Eq.~(\ref{eq:data-def})). From this data, we are able to learn various potentials, such as the Coulomb potential in three dimensions. For that, we translate the local averages into positions and charges via Newton's shell theorem.

With respect to the Coulomb potential, we start with a single Coulomb center and then provide an algorithm to extend to multi-Coulomb potentials. The basic idea is to first estimate the position of the particle, and then consider four appropriately chosen points in its neighborhood to approximate the charge number and position up to a given precision $\varepsilon$ with probability $1 - \delta$. Here, the lack of any information-propagation bounds due to the singular behavior of the Coulomb potential always induces a quadratic dependence of the sample complexity on the number of Coulomb centers. Owing to the free-fermion assumption of the model, we therefore generate the data sequentially to avoid errors in the free-fermion approximation. For the multi-Coulomb potential, the charge numbers can be estimated by a diagonally dominant linear system after a rough estimation of the positions. Then, the mutual influence on the output data is estimated and iteratively improved. In this way, the single-Coulomb case can be applied to the corrected and normalized data to find the positions and charge numbers up to a precision $\varepsilon$ and with probability $1 - \delta$. Here, we would like to emphasize that our direct treatment of the nonlinearity allows us to achieve the precision $\varepsilon$ independently of the discretization of the continuous space.

If we restrict our attention to the charge numbers, the reconstruction problem reduces to inverting the overlap matrix between the Coulomb potentials and the prepared states. This motivates the study of a broad class of linear systems. In particular, our modular algorithm is capable of learning not only Coulomb interactions, but also any smooth function or Fourier function. To achieve a good approximation, many states in different regions of space are required. To obtain an efficient algorithm, we parallelize our learning algorithm using novel information-propagation bounds \cite{hinrichs2023lieb} in combination with regularity assumptions on the input states. These regularity assumptions are also related to an interesting observation --- namely, Heisenberg’s uncertainty relation affecting the sample complexity.

Next, we provide more details on the state preparation, the measurement procedure, and the finite-difference scheme required for our protocol.

In more detail, we aim to obtain a good estimator of the potential on the box $[0,L]^d$, typically for $d=3$, but also in higher dimensions, with respect to the $L^\infty$-norm, where $L > 0$ is fixed. The measurement is carried out on a partition of the box into $m^d$ smaller boxes, each of side length $\ell = L/m > 0$ (see \Cref{fig:initial-states} for an illustration).

Following \Cref{protocol}, we allow the preparation of an even number of fermions, each described by a density function supported on one of the boxes (see \Cref{fig:initial-states})
\begin{equation*}
	B_{\j} \coloneqq \ell([0,1]^d + \j) \,,
\end{equation*}
with indices $\j \in \{0, \ldots, m-1\}^d\subset \mathbb{N}_0^d$. The requirement to prepare an even number of fermions ensures consistency with the fermionic superselection rule (see \cite{Vidal2021} for an overview). In particular, we group the set of boxes into triples, within which we prepare two fermions --- each in one of the three boxes. This avoids the use of a register system of size $m^d$. The desired expectation values are then estimated over three rounds (see \Cref{fig:initial-states}) by measuring the number of particles in each box after short time intervals.

\textbf{State preparation:} For the initial state we create fermions in the boxes $B_{\j}$, whose wave functions will be chosen as the rescaling of a fixed profile function $f$. We assume that $f$ is continuous, supported inside $[0,1]^d$ and $L^2$-normalized. Then
\begin{align}
	\label{eq:f_j defn}
	f_\j(x) \coloneqq \ell^{-d/2} f(\ell^{-1}x-\j),
\end{align}
is the corresponding rescaled and $L^2$-normalized profile inside the box $B_\j$. This choice is crucial for the local averages mentioned before. By taking local measurements, the derivative of the evolved probability distribution at $t=0$ is given by the local average
\begin{equation}\label{eq:data-def}
	\omega_{\j} := \int_{\R^d} \abs{ f_\j(x)}^2 V(x) \d x
\end{equation}
up to expectations of the Laplacian in the state $f_j$ (see \Cref{appx-subsec:finite-diff} for details). Estimators for these data points are the starting point for characterizing various models (see \Cref{sec:Coulomb} and \ref{sec:general-learning}). In order to construct suitable initial states, we group the boxes into triples (see~\Cref{fig:initial-states}). To this end, we construct a collection of disjoint subsets of $\{0:m-1\}^d\coloneqq\{0,...,m-1\}^d$ of cardinality three --- disjoint triples --- denoted by $\Jf$, where we assume that the three boxes form a connected set. For each element of a triple, we fix a label $\alpha \in \{0,1,2\}$. For brevity, when we write $j \in \Jf$, we refer to the triple whose zeroth box is $B_j$ --- labeled by $j$. This allows us to identify $\Jf \subset \{0:m-1\}^d$. Using this labeling of the triples, we denote the zeroth, first, and second box, as well as the corresponding initial functions, by
\begin{align*}
	B^\alpha_\j\qquad\text{and}\qquad f^\alpha_{\j}\,.
\end{align*}
We then prepare a set of three initial states
\begin{align}\label{eq:psi0_defn}
	\psi^{\alpha\beta}_0 = 2^{-\abs{\Jf} / 2}  \prod_{\j \in \Jf} (\Id + a^*(f^\alpha_{\j})  a^*(f^\beta_{\j})) \Omega,
\end{align}
where $(\alpha, \beta) \in\{(0,1), (0,2), (1,2)\}$.
\begin{figure}[ht!]
	\includegraphics[scale=0.8]{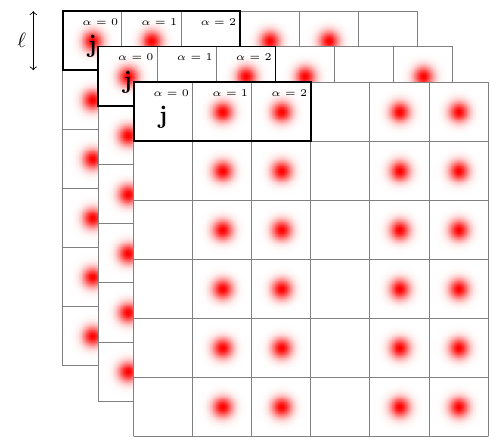}
	\caption{Depiction of the initial states for the three different pairs $(\alpha,\beta) = (0,1),(0,2),(1,2)$ on a grid of 6 $\times$ 6 boxes. The boxes are divided into triples (see first bold box, indexed by the first box at $\j =(0,0)$), and for each triple we consider the superposition of the vacuum with the state containing particles with rescaled profile $f$ in the boxes $\alpha$ and $\beta$ (shown in red).}
	\label{fig:initial-states}
\end{figure}
Before the measurement process, we apply an additional controlled displacement (q.v.~\cite{hayden}),
\begin{align*}
	W^{\alpha\beta}_\j = e^{\i \frac{\pi}{4} \bigl(a^*(f_\j^\alpha) a^*(f_\j^\beta) + a(f_\j^\beta) a(f_\j^\alpha)\bigr)},
\end{align*}
whose transpose acts on the vacuum as follows (cf. \Cref {th:controlled-displacement} for the short detailed computation):
\begin{equation}\label{eq:controlled-displacement}
	(W^{\alpha\beta}_\j)^* \ket{\Omega} = 2^{-1/2}\bigl(\Id  - \i  a^*(f^\alpha_{\j}) a^*(f^\beta_{\j})\bigr) \ket{\Omega}\,.
\end{equation}

\textbf{Measurement:} Depending on a priori knowledge of the potential --- Coulomb potential or smooth potentials --- we perform the measurements either sequentially or simultaneously for all triples $\mathbf{j}$ with fixed $\alpha, \beta$. Specifically, we test whether there are no fermions inside the two boxes $\alpha$ and $\beta$ of each triplet. This corresponds to measuring with respect to the projectors
\begin{equation}\label{eq:pvm}
	P^{\alpha\beta}_\mathbf{j} = \ketbra{\Omega}{\Omega}_{\FF(L^2(\widetilde B^{\alpha\beta}_\j))} \otimes \Id_{\mathrm{Rest}} \,,
\end{equation}
where we choose $\widetilde B^{\alpha\beta}_\j$ in the Coulomb case to be the union of the boxes $B^\alpha_\mathbf{j}$, $B^\beta_\mathbf{j}$, and all neighboring boxes. In all other cases we set $\widetilde B^{\alpha\beta} := B^\alpha_\mathbf{j} \cup B^\beta_\mathbf{j}$.
Further details on the rigorous construction of these projectors as well as other mathematical details about the second quantization formalism in $\IR^d$ are provided in \Cref{appx-sec:technical-details}.

With all the mentioned steps, we construct the following protocol:\\
\begin{algorithm}[H]
	\caption{Learning Fermions}\label{protocol}
	\KwIn{$T\in\N$, $B_{\j}$ for all $\j\in\{0:m-1\}^d$, $\Jf\subset\{0:m-1\}^d$}
	\KwOut{$\widehat{\omega_{\j}}$ for all $\j\in\{0:m-1\}^d$}
	\BlankLine
	\For{$(\alpha,\beta)\in\{(0,1), (1,2), (0,2)\}$}{
		\For{$k\in\{1:T\}$}{
			Prepare states $\psi_0^{\alpha\beta}$;\\
			Run $e^{-itH}$ on $\psi_0^{\alpha\beta}$ up to time $t=\cO(\varepsilon)$;\\
			\For{$\j\in\Jf$}{
				Perform measurements $P_{\j}^{\alpha\beta}$\\
				$\rightarrow$ get outcomes $Y_\j^{\alpha\beta,(k)}\in\{0,1\}$;
			}
		}
	}
	\Return{$\widehat{\omega}_{\j}^{\alpha}$ defined by $$\widehat{\omega}_{\j}^{\alpha}=\sum_{(\beta,\gamma)}\sigma_{\beta\gamma}^\alpha\sum_{k=1}^T\frac{Y_\j^{\beta\gamma,(k)}-1}{2tT}-  \cs{f_\j^{\alpha},(-\Delta) f_\j^{\alpha}}$$
		with $0\leq \beta < \gamma \leq 2$ and $\sigma_{\beta\gamma}^\alpha=1-2\cdot1_{\alpha\in\{\beta,\gamma\}}$ }
\end{algorithm}
which give a provably good estimator as shown in the following result:
\begin{theorem}\label{thm:sampling}
	Let $V \: \R^d \longrightarrow \R$ be  $4(d+1)$ times differentiable with bounded derivatives, and consider a grid consisting of boxes of side length $\ell=L/m$. Given a set of triples $\Jf$, \Cref{protocol} constructs estimators $\widehat{\omega}_{\j}^{\alpha}$ for the local averages
	\begin{align*}
		\omega_{\j}^{\alpha}=\cs{f_\j^\alpha, V f_\j^\alpha}, \quad \j \in \Jf,\,\alpha\in\{0,1,2\}
	\end{align*}
	with precision $\varepsilon$ and probability of success $(1-\delta)\in(0,1)$ after short time steps $t=\cO(\mathrm{poly}(\ell)\varepsilon)$. 
	The overall total evolution time of the algorithm is
	\begin{equation*}
		T_{\operatorname{c}}=\cO\left(\mathrm{poly}(\ell^{-1})  \varepsilon^{-3} \ln(\frac{3\abs{\Jf}}{\delta}) \right) \,.
	\end{equation*}
	If $\abs{\Jf}$ = 1, i.e., if only a single triple is learned, then the regularity assumption on V can be removed. It suffices to assume that V is relatively bounded with respect to $-\Delta$, which includes the Coulomb potential in $d = 3$.
\end{theorem}
\begin{remark}[Notation]
	The term $\widehat{\omega}_{\j}$ is used if there is a triple labeled by $\j'$ including $\j$ such that an estimator is constructed by the above algorithm and theorem.
\end{remark}
The proof and a discussion of the related methods are given in \Cref{sec:methods}. At this point, we would like to highlight the interplay between $\ell$ and $\varepsilon$ --- a phenomenon characteristic of the continuous-space learning case:
\begin{remark}[Precision vs.~pollution]
	If the potential is highly varying over the domain $[0, L]^d$, even classical interpolation requires a large number of points to achieve an accurate approximation. In our setting, this difficulty is further compounded by the uncertainty principle of quantum mechanics. Accurately probing such varying potentials necessitates constructing many particles within a confined region. Consequently, the fermionic wave packets increasingly resemble Dirac delta distributions, directly influencing the total evolution time as reflected in our bound.
\end{remark}

\section{Learning the Coulomb potential}\label{sec:Coulomb}
The presented learning algorithm is designed in a modular fashion building on the presented protocol in \Cref{protocol}, enabling the exploration of various quantum systems through a unified measurement procedure. As outlined in the introduction, our primary physical systems of interest are governed by the Coulomb potential. Prior to delving into the algorithms themselves, we highlight two foundational features essential for ensuring the stability of the learning process:
\begin{itemize}
	\item the \textit{harmonicity} of the Coulomb potential away from the singularities and
	\item \textit{Newton's shell theorem}.
\end{itemize}
From a physical point of view, and in order to make the above tools work, we assume that the profile function $f$, which defines our prepared states, is radially symmetric. This means there exists a real-valued function $g$ supported on $[0,1]$ such that
\begin{equation*}
	f(x)=g\!\left(\bigl\|x-\tfrac{1}{2}(1,\ldots,1)^T\bigr\|\right)
\end{equation*}
for all $x\in[0,1]^d$. A natural example is given by the so-called bump functions, i.e.,
\begin{equation*}
	g(y)=
	\begin{cases}
		e^{\tfrac{1}{y^2-\frac{1}{4}}} & \text{for } |y|\leq \frac{1}{2}, \\[6pt]
		0                              & \text{for } |y|>\frac{1}{2}
	\end{cases}\,.
\end{equation*}
With this, we define the states as before by
\begin{equation*}
	f_{\j}(x)=\ell^{-d/2}\,g\!\left(\bigl\|\ell^{-1}x-\tfrac{1}{2}(1,\ldots,1)^T-\j\bigr\|\right)
\end{equation*}
for all $x\in[0,L]^d$. Note that the midpoint of such a radially symmetric function is given by
\begin{equation}\label{eq:symmetry-points}
	p_{\j}\coloneqq \ell\Bigl(\j+\tfrac{1}{2}(1,\ldots,1)^T\Bigr)\,,
\end{equation}
with $(\cdot)^T$ denoting the transpose of the vectors. For the evaluation of Newton’s shell theorem, the symmetry is important.

As a first step, we introduce a robust estimator for both the charge number and the spatial location of a single Coulomb potential.

\subsection{Single Coulomb potential}
The single Coulomb potential is given by
\begin{equation*}
	V(x)=\frac{\lambda}{\nn{x-y}}
\end{equation*}
for a charge number $\lambda\in\R_{\geq0}$, location $y\in[0,L]^{3}$ and $x\in\R^3$. As stated in \Cref{thm:sampling}, we have access to estimators on the local averages $\omega_{\j}$, which satisfy
\begin{equation*}
	\omega_{\j}=\cs{f_{\j},V f_{\j}}=\int_{\R^3}|f_{\j}(x)|^2\frac{\lambda}{\nn{x-y}}dx=\frac{\lambda}{\nn{p_{\j}-y}}
\end{equation*}
if $y\notin B_{\j}$, due to the well-known \textit{Newton's shell theorem} (see \cite[Thm.~9.7]{LiebLoss.2001} for details). Note that $p_\j$ is the middle point of the box $B_{\j}$, i.e., $p_{\j}=\ell(\j+\tfrac{1}{2}(1,\ldots,1))$. This identity allows us to provide an estimator for the charge number and location:
\begin{theorem}\label{thm:single-Coulomb}
	Let $\varepsilon>0$ be the given precision and let $V$ be a single Coulomb potential
	\begin{equation*}
		V:\R^3\to\R, \quad V(x)=\frac{\lambda}{\|x-y\|}
	\end{equation*}
	for unknown $\lambda$ with $|\lambda|\in[\Lambda_*,\Lambda^*]\subset\R_+$ and $y\in[0,L]^3$. Then, using \Cref{protocol}, we obtain estimators $\widehat{\lambda}$ and $\widehat{y}$ so that, with success probability at least $(1-\delta)\in(0,1)$,
	\begin{equation*}
		|\widehat{\lambda}-\lambda| \leq \varepsilon,\qquad\|\widehat{y}-y\|\leq \varepsilon,
	\end{equation*}
	requiring a total evolution time of
	\begin{equation*}
		T_{\operatorname{c}}=\cO\Bigl(\varepsilon^{-3}\ln\biggl(\frac{1}{\delta}\biggr)\Bigr)\,.
	\end{equation*}
\end{theorem}

\begin{proof}[Proof-sketch]
	Define $\ell=L/m$ with $m=8$. Due to the identity of Newton's shell theorem (see \cite[Thm.~9.7]{LiebLoss.2001}), the proof reduces to solving a system of non-linear equations with four degrees of freedom:
	\begin{equation*}
		\omega_{\j}=\frac{\lambda}{\nn{p_{\j}-y}}\qquad\text{or}\qquad\frac{\lambda^2}{\omega_{\j}^2}=z_\j
	\end{equation*}
	for $\j\in\{0:7\}^3$ and $z_\j=\nn{p_{\j}-y}^2$. Next, we define
	\begin{equation}\label{eq:single-Coulomb-redefinition}
		\begin{aligned}
			p_{\mathbf{i}\j}    & =\frac{1}{2}(p_{\mathbf{i}}-p_{\j})                    \\
			c_{\mathbf{i}\j}    & =(\|p_\mathbf{i}\|^2-\|p_\j\|^2)                       \\
			\eta_{\mathbf{i}\j} & =(\frac{1}{\omega_\mathbf{i}^2}-\frac{1}{\omega_\j^2})
		\end{aligned}
	\end{equation}
	which results in
	\begin{equation*}
		c_{\mathbf{i}\j}=\eta_{\mathbf{i}\j}\lambda^2+\cs{p_{\mathbf{i}\j},y}\,.
	\end{equation*}
	To solve the above system of linear equations, we require maximally eight points with index $\j$ labeled by $j\in\{0:7\}$ for simplicity. Then, these points and the corresponding estimated local averages reduce to the linear system
	\begin{equation*}
		c\coloneqq
		\begin{pmatrix}
			c_{12} \\
			c_{34} \\
			c_{56} \\
			c_{78}
		\end{pmatrix}
		=
		\begin{pmatrix}
			\eta_{12} &  & p_{12}^T & \\
			\eta_{34} &  & p_{34}^T & \\
			\eta_{56} &  & p_{56}^T & \\
			\eta_{78} &  & p_{78}^T &
		\end{pmatrix}
		\begin{pmatrix}
			\lambda^2 \\
			y_1       \\
			y_2       \\
			y_3
		\end{pmatrix}
		\eqqcolon A\widetilde{y}\,.
	\end{equation*}
	Here it is important to mention that the matrix and the vector $c$ are fully characterized. As before, note that $(\cdot)^T$ denotes the transpose of the vectors. The algorithm then looks as follows: First the area $[0,L]^3$ is divided into $8^3=512$ boxes and partitioned (except for two boxes) in triples as explained in \Cref{fig:initial-states}. Then, \Cref{protocol} is executed for any triple in  $\Jf$ separately. Via this sequential scheme, we have learned all weighted averages $\omega_{\j}$ for the partition of $512$ boxes. From there, we calculate the variables defined in \Cref{eq:single-Coulomb-redefinition} and appropriately choose four vectors $p_{\mathbf{i}\j}$ so that the linear system is well-conditioned, i.e., errors propagate linearly (see \Cref{appx-sec:Coulomb} for details).
\end{proof}

\begin{figure}
	\centering
	\includegraphics[width=0.95\linewidth]{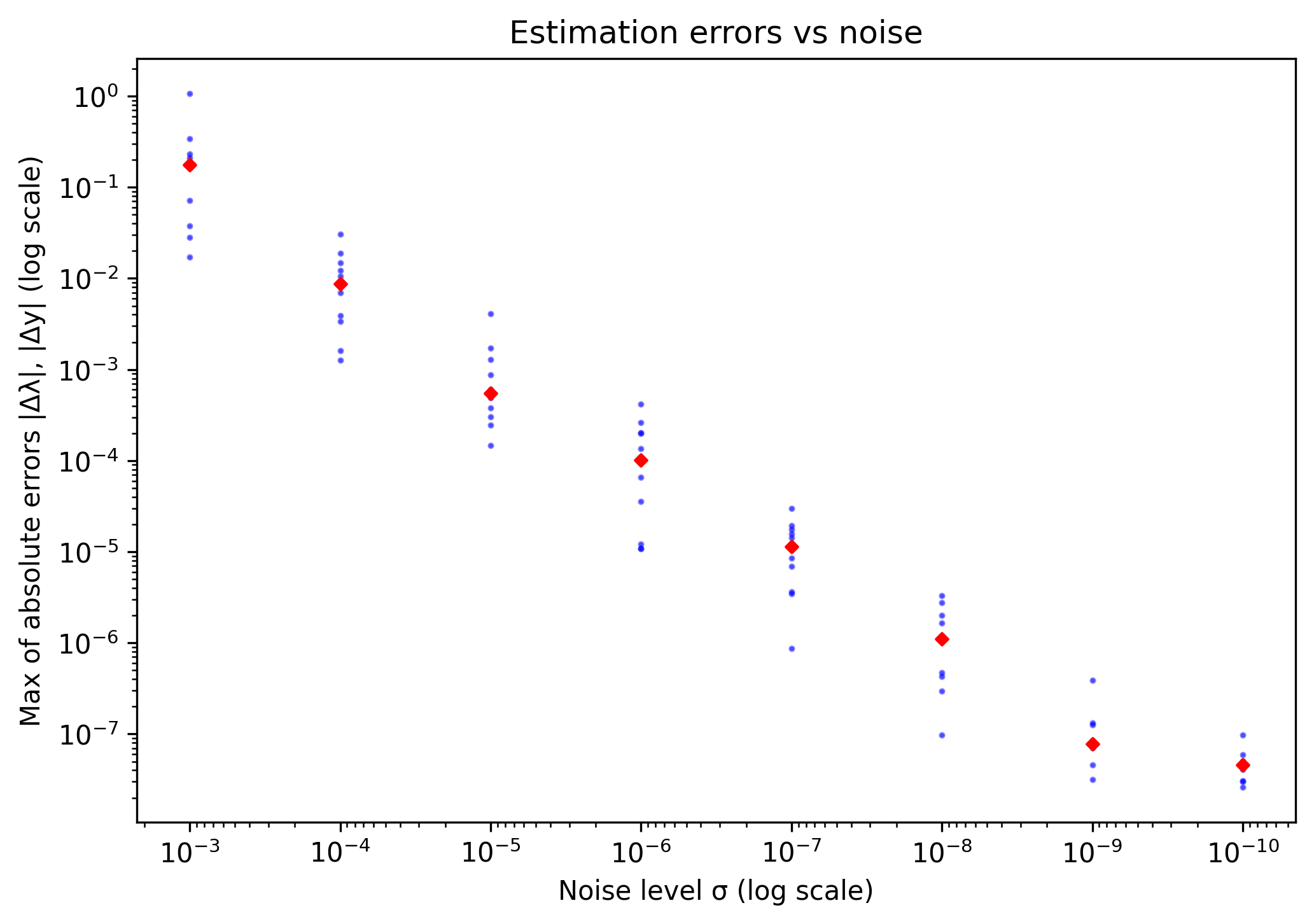}
	\caption{\justifying
		Absolute errors of the proposed single Coulomb post-processing algorithm for decreasing noise levels, for local averages generated by direct classical quadrature of a random Coulomb potential and perturbed by Gaussian noise with decreasing standard deviation $\sigma$ with grid $m=8$ (see \Cref{appx-sec:numerics}, \href{https://github.com/MitTimM/learning_coulomb_potentials_in_continuous_space}{GitRepo}, and \cite{DataZenodo.2026}).
	}
	\label{fig:single-coulomb-error-dependence}
\end{figure}

\begin{remark}
	Beyond the given Coulomb potential in $3$ dimensions, one can extend the above result to any radially symmetric potential. This is achieved by a simple optimization method and finding the approximate intersection of $d+1$ spheres for dimension $d$ (see \Cref{appx-rmk:extension-single-mode}).

	The partition of $[0,L]^3$ into boxes can be refined in various ways; here it was simply chosen to be $8^3$ in order to simplify the notation and the argument.
\end{remark}

\subsection{Multi-Coulomb potential}
In the next section, we extend the single Coulomb algorithm to a few Coulomb centers.
\begin{theorem}\label{thm:multi-Coulomb}
	Let $L$, $\Lambda_*$, $\Lambda^*, y_*>0$, $K\in\N$ be constants defining a multi-Coulomb potential $V$ by
	\begin{equation*}
		V:\R^3\to\R, \quad V(x)=\sum_{k=1}^K\frac{\lambda_k}{\|x-y_k\|}
	\end{equation*}
	for unknown $K\in\N$, $\lambda_k\in[\Lambda_*,\Lambda^*]$ and $y_k\in[0,L]^3$ satisfying $\|y_k-y_{k'}\|\geq y_*$ for all $k\neq k'\in\{1:K\}$. For a fixed grid size
	\begin{equation*}
		\ell \leq \frac{L}{m}\quad\text{with}\quad m=\Bigl\lceil L\mathrm{poly}(K, y_*^{-1},y_*,\Lambda_*^{-1},\Lambda^*)\Bigr\rceil\,,
	\end{equation*}
	we achieve estimators $\widehat{\lambda_i}$ and $\widehat{y_i}$ such that, with success probability at least $(1-\delta)\in(0,1)$,
	\begin{equation*}
		\max_{i\in\{1:K\}}\{|\widehat{\lambda}_i-\lambda_i|, |\widehat{y}_i-y_i|\}\leq \varepsilon,
	\end{equation*}
	requiring a total evolution time of
	\begin{equation*}
		T_{\operatorname{c}}=\cO\Bigl(\mathrm{poly}(K, y_*^{-1},y_*,\Lambda_*^{-1},\Lambda^*)\varepsilon^{-3}\ln\biggl(\frac{1}{\delta}\biggr)\Bigr)\,.
	\end{equation*}
\end{theorem}
\begin{proof}[Proof sketch]
	The detailed proof can be found in \Cref{appx-thm:multi-Coulomb}. In the first step, we apply \Cref{thm:sampling} sequentially to obtain estimators for the local averages $\widehat{\omega_{\j}}$ for all $\j \in \{0:m-1\}^3$, requiring a sample complexity of
	\begin{equation*}
		T_c = \mathcal{O}\!\left(\mathrm{poly}(\ell^{-1})\varepsilon_\omega^{-3}\ln\!\frac{1}{\delta}\right)\,.
	\end{equation*}
	Note that $\ell$ is constant in $\varepsilon_\omega$ and $\varepsilon_\omega$ depends linearly on $\varepsilon$ (see \Cref{appx-eq:varepsilon-multi-coulomb}). Next, we obtain coarse estimators ${p}_{\hat{\j}_k}$ for the locations of the Coulomb centers with precision
	\begin{equation}\label{eq:first-approx-multi-coulomb-location}
		\|{p}_{\hat{\j}_k} - y_k\| \leq \sqrt[3]{\ell}
	\end{equation}
	and probability $(1 - \delta)$ for all $k \in \{1:K\}$.

	Under the assumption on $\ell$, we construct a linear, diagonally dominant system for the charge numbers $\lambda_k$ (compare to \cite[Sec.~6]{Franca.2026}), which allows us to obtain estimators $\widehat{\lambda}_k$ satisfying
	\begin{equation}\label{eq:first-approx-multi-coulomb-charge}
		|\widehat{\lambda}_{k} - \lambda_k| \leq 2K\Lambda^*y_*^2\frac{\sqrt[3]{\ell}}{\ell^{3/12}}+y_*\varepsilon_\omega
	\end{equation}
	with probability $(1 - \delta)$ for all $k \in \{1:K\}$ (see \Cref{appx-eq:lambda-diag-dom}). Such systems possess the important properties of stability and efficient solvability (see \cite{fawzilecture.2022numerical}).

	Finally, we use the approximated charge numbers to compute estimators for scaled, isolated local averages
	\begin{equation*}
		\omega_\j^{(k')} \coloneqq\sqrt[3]{\ell}\int_{\mathbb{R}^3} |f_\j(x)|^2 \frac{1}{\|x - y_{k'}\|} \, \mathrm{d}x
	\end{equation*}
	in the neighborhood of $y_k$. This enables us to iteratively refine the upper bound given in \Cref{eq:first-approx-multi-coulomb-location} through repeated applications of \Cref{thm:single-Coulomb} and for isolated local averages and for \Cref{eq:first-approx-multi-coulomb-charge} by solving diagonally dominated linear systems stepwise improved by the better estimators on the Coulomb centers.
\end{proof}

\begin{figure}
	\centering
	\includegraphics[width=0.95\linewidth]{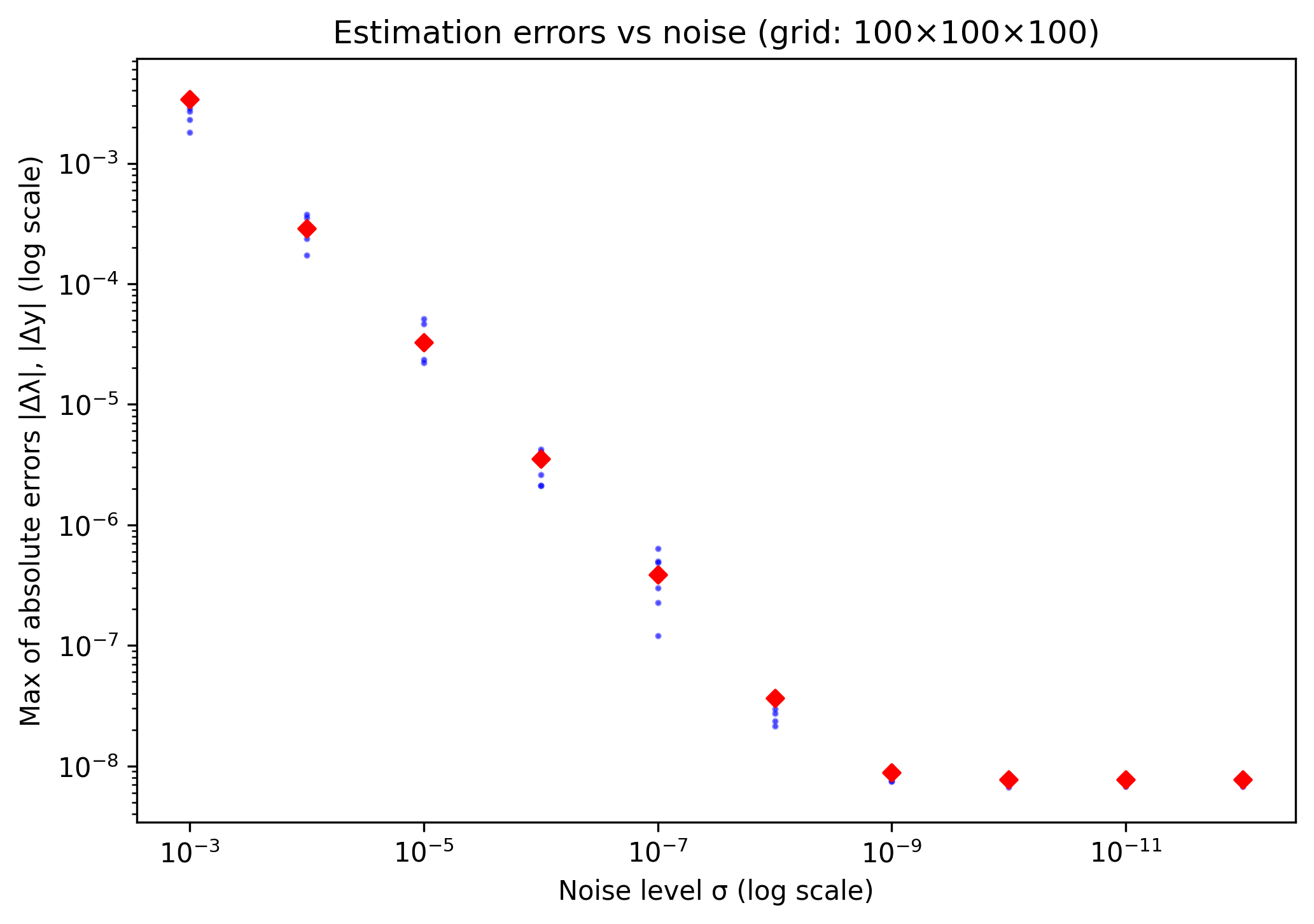}
	\caption{\justifying
		Absolute errors of the proposed $10$ multi-Coulomb post-processing algorithm for decreasing noise levels, for local averages generated by direct classical quadrature of a random Coulomb potential and perturbed by Gaussian noise with decreasing standard deviation $\sigma$ with grid $m=100$. The plateau is a consequence of the numerical error made by using quadrature instead of exact integration (see \Cref{appx-sec:numerics}, \href{https://github.com/MitTimM/learning_coulomb_potentials_in_continuous_space}{GitRepo}, and \cite{DataZenodo.2026}).
	}
	\label{fig:multi-coulomb-error-dependence}
\end{figure}

\begin{figure}
	\centering
	\includegraphics[width=0.95\linewidth]{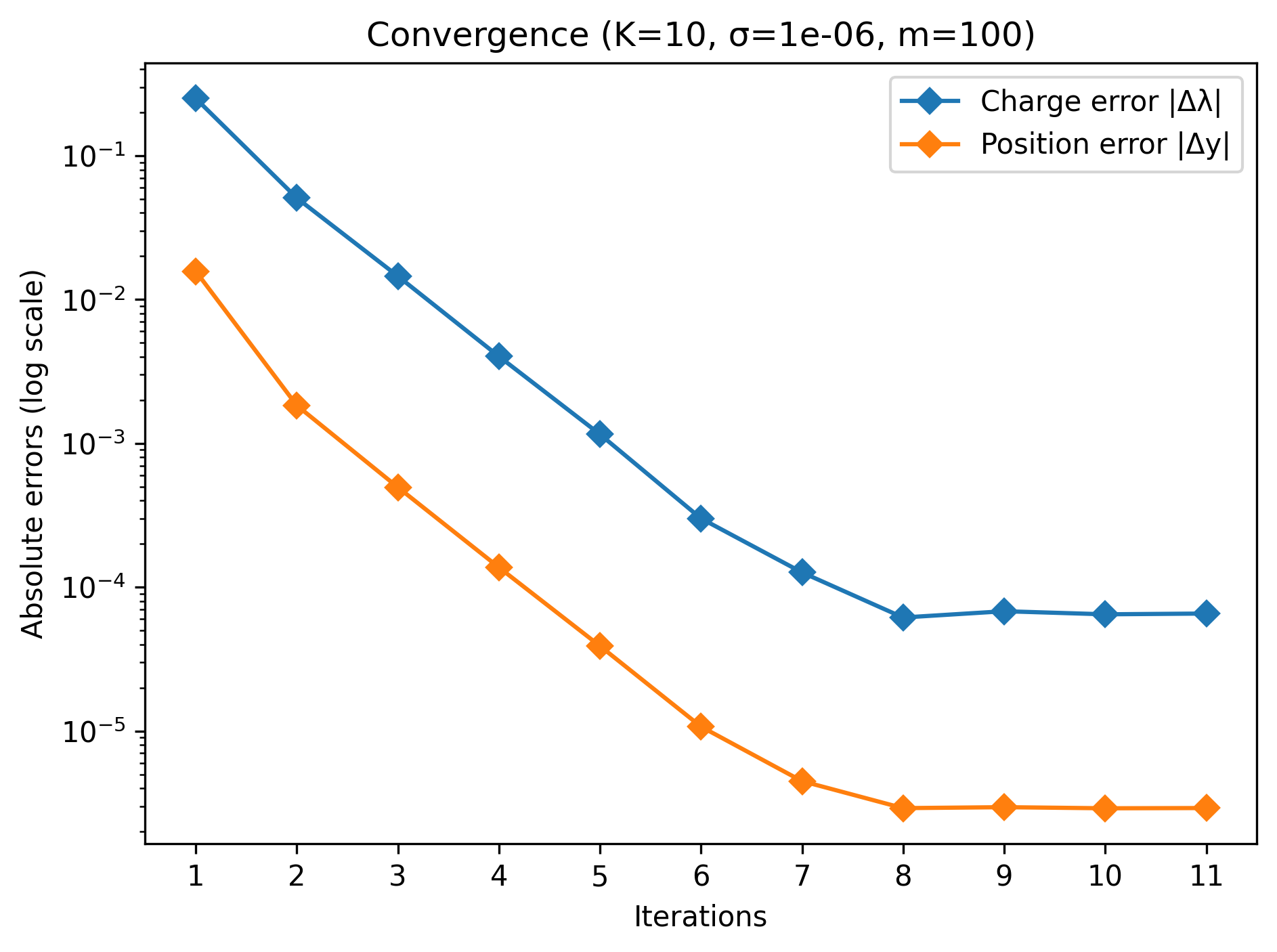}
	\caption{\justifying
		Absolute errors of the proposed multi Coulomb post-processing algorithm over the iteration presented in the proof. The algorithm is applied to local averages generated by direct classical quadrature of a random $10$ multi-Coulomb potential and perturbed by Gaussian noise with decreasing standard deviation $\sigma$ with grid $m=100$ (see \Cref{appx-sec:numerics}, \href{https://github.com/MitTimM/learning_coulomb_potentials_in_continuous_space}{GitRepo}, and \cite{DataZenodo.2026}).
	}
	\label{fig:multi-coulomb-convergence}
\end{figure}

\begin{remark}
	Note that the above method can be extended to rotationally invariant potential not satisfying Newton's shell theorem and only requires a certain decay rate of the individual potentials away from their centers, as well as stability of the maxima. If a more general potential in $d$ dimensions satisfies these properties, a similar method can be employed (see \Cref{appx-rmk:extension-single-mode} and (\ref{appx-rmk:extension-multi-mode})).
\end{remark}
As seen in the proof sketch above, determining the charge numbers reduces to solving a linear system. Next, we will investigate this case in greater detail.

\section{Learning general potentials}\label{sec:general-learning}
In this section, we extend our analysis beyond a specific model and aim for good approximations of certain function classes.

First, assume that the external potential $V$ is given by a Lipschitz continuous function with Lipschitz constant $C_V$. Using \Cref{thm:sampling} on the set of triples $\Jf$, \Cref{protocol} provides estimators $\widehat{\omega}_\j$ for the local averages $\langle f_\j, V f_\j \rangle$ with errors at most $\varepsilon$ and success probability at least $1-\delta$ for all $\j\in\{0:m-1\}^d$. Here, we implicitly assume that every box is included in a triple of $\Jf$.

In the next step, we consider an approximation of the potential given by step functions:
\begin{align*}
	\widetilde{V} := \sum_{\j} \widehat{\omega}_\j \, \chi_{B_\j}\,.
\end{align*}
For any $y \in B_\j$, we have
\begin{equation*}
	\| V(y) - V|_{B_\j} \|_{\infty,B_{\j}} \leq C_V \sqrt{d}\,\ell,
\end{equation*}
where $\|\cdot\|_{\infty,B_\j}$ denotes the supremum norm restricted to the set $B_\j$. By choosing $\widehat{\omega}_{\j}$ to be the estimated local average related to the box $B_{\j}$, we achieve
\begin{equation}\label{eq:V step function approx}
	\begin{aligned}
		\| V - \widetilde{V} \|_{\infty,[0,L]^d} & \leq C_V \sqrt{d}\,\ell+\sup_{\j}\sup_{y\in B_{\j}} \bigl| V(y) - \widehat{\omega}_\j \bigr| \\
		                                         & \leq C_V \sqrt{d}\,\ell + \varepsilon\,,
	\end{aligned}
\end{equation}
with probability $1-\delta$.
%

Determining the coefficients of the step functions can be viewed as solving a linear problem, which is naturally generalized to arbitrary bases. In fact, assume that the potential $V:\R^d \to \R$ admits the following sum representation up to an error of order $\varepsilon_V \geq 0$,
\begin{equation}\label{eq:V-rep}
	\Bigl\|V - \sum_{k=1}^K \lambda_k e_k\Bigr\|_{\infty,[0,L]^d} \leq \varepsilon_V,
\end{equation}
where $\{e_k\}_{k=1}^{K}$ are normalized, linearly independent vectors in $L^2(\R^d,\mathbb{C})$. We also allow for the case $\varepsilon_V=0$, which means that $V$ is represented exactly by this system. Notice that \Cref{eq:V-rep} implies
\begin{equation}\label{eq:omegaj sum estimate}
	\biggl|\omega_\j - \Bigl\langle f_{\mathbf{j}}, \sum_{k=1}^K \lambda_k e_k f_{\mathbf{j}} \Bigr\rangle \biggr| \leq \varepsilon_V \,.
\end{equation}
Assume we know the values $\omega_{\j}$ for $\j\in\{0:m-1\}^d$. For the sake of notation, we write $\omega_i$ for $i\in\{1:m^d\}$ for a given order $\{\j_i\}_{i=1}^{m^d}$ on the set $\{0:m-1\}^d$ such that $\omega_i=\omega_{\j_i}$. The same notation is applied to $f_\j$. Then we can write
\begin{equation*}
	\begin{aligned}
		\begin{pmatrix}
			\omega_{1} \\
			\vdots     \\
			\omega_{m^d}
		\end{pmatrix}
		=
		\begin{pmatrix}
			\ddots                & \vdots                                                       & \reflectbox{$\ddots$} \\
			                      & \left(\langle f_{k}, e_i f_{k}\rangle\right)_{i,k=1}^{m^d,K}                         \\
			\reflectbox{$\ddots$} & \vdots                                                       & \ddots
		\end{pmatrix}
		\begin{pmatrix}
			\lambda_1 \\
			\vdots    \\
			\lambda_K
		\end{pmatrix}
		+ \mathcal{O}(\varepsilon_V)\,,
	\end{aligned}
\end{equation*}
or in short-hand notation
\begin{align}\label{eq:matrix-rep}
	\boldsymbol{\omega} = \mathbb{M} \boldsymbol{\lambda} + \cO(\varepsilon_V)\,.
\end{align}
In the following we will use the bold notation to denote the vectors and $\mathbb{M}$ for the $m^d \times K$ matrix. \Cref{protocol} provides accurate estimators of the local averages
\begin{equation*}
	\widehat{\boldsymbol{\omega}} = (\widehat{\omega}_{1}, \ldots, \widehat{\omega}_{m^d})\,,
\end{equation*}
which can then be used to compute estimators for the coefficients via
\begin{equation*}
	\widehat{\boldsymbol{\lambda}} = \mathbb{M}^{+} \widehat{\boldsymbol{\omega}}\,,
\end{equation*}
where $\mathbb{M}^{+}$ denotes the pseudoinverse (or Moore–Penrose inverse) of $\mathbb{M}$. Note that if $\mathbb M \mathbf x = \mathbf y$, then $\mathbf x = \mathbb M^+ \mathbf y$ minimizes $\nn{\mathbb M \mathbf x - \mathbf y}$. In particular,
$\mathbb{M}^{+} = \mathbb{M}^{-1}$ in case of a square matrix ($m^d = K$) which is invertible.
From \Cref{eq:matrix-rep} it follows that
\begin{equation}\label{eq:linear-error-bound}
	\boldsymbol{\lambda}-\widehat{\boldsymbol{\lambda}} =  \mathbb{M}^{+} (\boldsymbol{\omega}-\widehat{\boldsymbol{\omega}}) + \mathbb{M}^{+} \, \mathcal{O}(\varepsilon_V)\,,
\end{equation}
which yields an error estimate for $\boldsymbol{\lambda}$ and, consequently, for $V$. This observation leads to the following result:
\begin{theorem}\label{thm:general-linear-learning}
	Given that the potential is $4(d+1)$ times differentiable with bounded derivatives and admits an approximation expansion in a system of linear independent, normalized vectors $e_1,\ldots,e_{K}\in L^2(\R^d,\R)$ up to an error $\varepsilon_V$ (Eq.~(\ref{eq:V-rep})).
	Moreover, assume that $m = \cO(\mathrm{poly}(K))$ and the norm of the pseudoinverse scales like $\|\mathbb{M}^{+}\|=\cO(1)$.
	Then, \Cref{protocol} constructs estimators $\widehat{\lambda}_{i}$ with precision $\varepsilon$, i.e.,
	\begin{equation}\label{eq:linear-estimation}
		\max_{i\in\{1:K\}}|\widehat{\lambda}_{i}-\lambda_{i}|<\|\mathbb{M}^{+}\|(\varepsilon+\varepsilon_V)
	\end{equation}
	and success probability $(1-\delta)\in(0,1)$. The overall total evolution time of the algorithm is
	\begin{equation*}
		T_{\operatorname{c}}=\cO(\varepsilon^{-3}\ln(\frac{K}{\delta}))\,.
	\end{equation*}
\end{theorem}
\begin{proof}
	The proof follows the strategy outlined above. First, the overlap matrix is invertible by assumption. This directly implies the result by applying the triangle inequality to \Cref{eq:linear-error-bound}.
\end{proof}
\begin{remark}
	In the above proof, we used the assumption that $\|\mathbb{M}^{+}\|$ is constant. Under the additional relative error assumption that $1\leq\|\mathbb{M}\|$, the above result can also be formulated in terms of a constant condition number
	\begin{equation}\label{eq:condition-number}
		\kappa(\mathbb{M}) =\nn{\mathbb{M}}\,\|\mathbb{M}^{+}\|=\cO(1)\,.
	\end{equation}
	In detail, \Cref{eq:linear-estimation} would change to
	\begin{equation}\label{eq:linear-estimation-conditon}
		\max_{i\in\{1:K\}}|\widehat{\lambda}_{i}-\lambda_{i}|<\kappa(\mathbb{M})(\varepsilon+\varepsilon_V)\,.
	\end{equation}
	The condition number has the advantage that it also provides stability with respect to relatively small perturbations of the matrix itself.
\end{remark}

Next, we briefly recap the case of a step function approximation:
\begin{example}[Step Functions]
	As described at the beginning of \Cref{sec:general-learning}, we define
	\begin{equation*}
		e_{\j}=\chi_{B_\j}\,,
	\end{equation*}
	which directly implies
	\begin{equation*}
		\langle f_{\j_1}, e_{\j_2} f_{\j_1}\rangle=\delta_{\j_1\j_2}\,.
	\end{equation*}
	Obviously, the matrix is invertible with $\|\mathbb{M}^{-1}\|=1$ and condition number $\kappa(\mathbb{M})=1$.
\end{example}
Even if the matrix is not perfectly diagonal as in the example above, but most of the weight lies on the diagonal, the matrix remains well-conditioned under sufficiently small perturbations. For that, we introduce diagonally dominant matrices.
\begin{definition}\label{def:diagonally-dominante-matrix}
	A square matrix $\mathbb{M} = (a_{ij}) \in \R^{K\times K}$ is called diagonally dominant if for all rows $i\in\{1:K\}$, the following condition holds:
	\begin{equation*}
		|a_{ii}| \ge \sum_{j \neq i} |a_{ij}|\,.
	\end{equation*}
	The matrix is called strictly diagonally dominant if the inequality is strict for all rows.
\end{definition}
\begin{example}[Diagonally Dominant Matrices]
	Assume that the overlap matrix $\mathbb{M}$ and its transpose $\mathbb{M}^{T}$ are strictly diagonally dominant, satisfying
	\begin{equation*}
		|a_{ii}| - \max\Bigl\{\sum_{j \neq i} |a_{ij}|, \sum_{j \neq i} |a_{ji}|\Bigr\} > \varepsilon_{\mathbb{M}}^{-1} \quad \text{for all } i\,.
	\end{equation*}
	Then, by the Levy–Desplanques theorem \cite[Chap.~6, Cor.~5.6.16/17]{Horn1985}, the matrix is invertible and Varah's bound shows \cite[Cor.~2]{Varah.1975}
	\begin{equation}\label{eq:diag-dom-inverse-norm}
		\|\mathbb{M}^{-1}\|\!\leq\!\frac{1}{\min_{i}\!\left(\!|a_{ii}|\!-\!\min\{\sum_{j \neq i}|a_{ij}|, |a_{ji}|\}\!\right)} \leq \varepsilon_{\mathbb{M}}.
	\end{equation}
	A simple example would be a diagonal matrix subject to small off-diagonal perturbations. Such systems remain stable and satisfy \Cref{thm:general-linear-learning}.
\end{example}

\begin{example}
	In this example, we assume that our $V$ can be decomposed into a linear combination of trigonometric functions with a fixed set of frequencies. More precisely, the base elements are given by 
    $$e_{i}(x) =  e^{ 2\pi \i  \kb_i  \cdot  x},\qquad i \in\{1:K\}$$
	for some given $\mathbf k_1, \ldots, \mathbf k_K \in \IR^d$. We want to assume that the nodes $\kb_i$ are well separated for large $K$. To this end, let $q := \min_{i \not= i'} \min_{\mathbf r \in \mathbb Z^d} \nn{\mathbf k_{i} - \mathbf k_{i'} + \mathbf r}_\infty$ be the minimal separation distance of $\kb_i$ on the $d$-dimensional torus $\mathbb T^d = (\IR/\mathbb Z)^d$. Under the assumption that
	\begin{align}
		\label{eq:q assumption}
		q = \Omega(1 / \text{poly}(K)).
	\end{align}
	we can show that the norm of the pseudoinverse of the transition matrix, $\nn{M^+}$,  is bounded  and therefore \Cref{thm:general-linear-learning} is applicable.

	To show that, we compute the elements of the transition matrix as
	\begin{align*}
		\cs{f_\j, e_i f_\j} & = \ell^{-d} \int_{\R^d} \abs{f(\ell^{-1} x - \j)}^2 e^{\i  \kb_i \cdot  x} \d x \\  &= \int_{\R^d} \abs{f(x)}^2 e^{\i  \ell  \kb_i \cdot  (x + \j)} \d x .
	\end{align*}
	Therefore, $\mathbb M$ factorizes as
	\begin{align*}
		\mathbb{M} =  \mathbb D  \mathbb V,
	\end{align*}
	where $\mathbb D$ is a diagonal matrix with entries $\mathbb D_{ii} := \int \abs{f(x)}^2 e^{\i  \ell  \kb_i \cdot x} \d x$, $i = 1, \ldots, K$, and $\mathbb V$ is a multivariate Vandermonde matrix with complex values on the unit circle,  given by $\mathbb V_{i \j} =  e^{\i  \ell  \kb_i \cdot  \j}$. Thus,
	\begin{align*}
		\nn{\mathbb{M}^{+}} \leq \nn{\mathbb D^{-1}} \nn{\mathbb V^{+}}.
	\end{align*}
	In order to find an upper bound for $\nn{\mathbb D^{-1}}$, we can use $\abs{e^{\i  \ell  \kb_i \cdot x} -1} \leq \ell \abs{\kb_i} \abs{x}$, which yields $\abs{\mathbb D_{ii} - 1} \leq \ell \abs{\kb_i} \int \abs{f(x)}^2 \abs{x} \d x$, which is strictly smaller than $1$ for $\ell$ small enough. Then by means of a Neumann series we see that
	\begin{align*}
		\nn{\mathbb D^{-1}} \leq \sum_{n=0}^\infty \nn{1-\mathbb D}^n \leq \frac{1}{1-C\ell},
	\end{align*}
	where $C>0$ depends on the $\kb_i$ and $f$, but not on $\ell$.

	In \cite[Theorem 3.5]{kunis2022multivariate} it is shown that, if
	\begin{align}
		\label{eq:q m assumption}
		q (m-1) \geq \pi^{-1}(8 \log d + 14),
	\end{align}
	then we have an upper bound
	\begin{align*}
		\nn{\mathbb
			V^{+}} = 1 / \sigma_{\text{min}}(\mathbb V) \leq C_d (m-1)^{-d/2},
	\end{align*}
	where $\sigma_{\text{min}}(\mathbb V)$ denotes the smallest singular value of $\mathbb{V}$ and $C_d$ is an explicit constant, only depending on $d$.

	Together, this means that we achieve a uniform bound on $\nn{\mathbb M^+}$ independent of small $\ell$ (or equivalently large $m$). Note that Eq.~(\ref{eq:q m assumption}) is satisfied by the assumption \eqref{eq:q assumption} and if we choose $m = \Omega(\text{poly(K)})$.
\end{example}

\section{Methods for Error Estimates and Proof of Theorem \ref{thm:sampling}}
\label{sec:methods}

In this section, we give the details how to derive error bounds for the estimators of the local averages $\omega_\j$ obtained in \Cref{protocol}. In particular, we give the proof of \Cref{thm:sampling}, up to the more technical estimates which can be found in \Cref{appx-subsec:finite-diff}. The continuum-specific challenges, unbounded propagation speed and infinite local Hilbert-space dimension, have to be handled with novel methods, particularly continuum Lieb-Robinson bounds. Once these issues have been addressed, standard tools in Hamiltonian learning, such as finite differences and Hoeffding's inequality, can be employed.

To obtain a good estimator for the second term of the derivative at $t=0$, i.e., $\omega_{\mathbf{j}}$ (see Eq.~(\ref{eq:data-def})), we estimate the expectations of the evolved states $\psi^{\alpha\beta}_t=e^{itH}\psi^{\alpha\beta}_0$, i.e.
\begin{align}\label{measure}
	\Dd^{\alpha\beta}_\j(t) \coloneqq \langle W^{\alpha\beta}_\mathbf{j} \psi^{\alpha\beta}_t, \,P^{\alpha\beta}_\mathbf{j} W^{\alpha\beta}_\mathbf{j} \psi^{\alpha\beta}_t \rangle\,,
\end{align}
using the corresponding mean estimator (\ref{eq:mean estimator}). The required sample complexity $T$ (see Eq.~(\ref{eq:T_est})) follows directly from Hoeffding's inequality (\Cref{lem:hoeffy}). Note that our scheme is based on i.i.d.~Bernoulli random variables, which is why alternative estimators such as the median-of-means (see \cite{Devroye2016} for a detailed analysis of sub-gaussian mean estimators) achieve similar performance.

To characterize the Hamiltonian, we want to compute the local averages $\omega_\j = \langle f_\mathbf{j}, V f_\mathbf{j} \rangle$ (Eq.~(\ref{eq:data-def})) from measurements of our time-evolved initial states. In fact, the values $\omega_\j$ can be computed explicitly from the first derivative of $\Dd^{\alpha\beta}_\j(t)$ at $t=0$, see \Cref{th:first_order} for the precise statement. Notice that the first derivative is given by a commutator with $H$. As the application of the Schrödinger Operator $h = - \Delta + V$ to some $f_\j$ is well-defined and still localized in the box $\j$, this does not destroy the localization of the particles to each box and therefore entails a straightforward computation.

Via a Taylor expansion in $t$, the first derivatives can then be approximated by finite-difference quotients, with an error controlled by the second derivative:
\begin{align*}
	\langle & f^\alpha_\j, V f^\alpha_\j \rangle + 2\ell^{-2}\cs{ f, (-\Delta) f }                                                                                               \\
	        & = \sum_{(\beta,\gamma)}\sigma_\alpha(\beta, \gamma) \biggl( \frac{\Dd_\j^{\beta\gamma}(t) - \Dd_\j^{\beta\gamma}(0) }{t}                                           \\
	        & \qquad\qquad\qquad\qquad\quad+ t\cs{\psi^{\beta\gamma}_{t_{\beta\gamma}} , [H,[H,\widetilde{P}^{\beta\gamma}_\j]] \psi^{\beta\gamma}_{t_{\beta\gamma}}} \biggr)\,,
\end{align*}
where $f$ is the fixed profile function in \eqref{eq:f_j defn}, the sum is over $(\beta,\gamma)\in\{(0,1),(0,2),(1,2)\}$, $\widetilde{P}^{\beta\gamma}_\j\coloneqq(W^{\beta\gamma}_\j)^* P^{\beta\gamma}_\j W^{\beta\gamma}_\j$, $\sigma_\alpha(\beta, \gamma) \in \{\pm 1\}$ is defined in \Cref{th:final_fVf_D_estimate}, and in the Lagrange remainder term, we have $t_{\beta\gamma} \in (0,t)$.

We would like to argue that the expectation of the second commutator is bounded, as the prefactor $t$ makes it small. In contrast to the first order term, computing and bounding this term for multiple triples $\abs{\Jf} > 1$ is significantly more challenging, as it is not evaluated at $t=0$ but at some intermediate times $t_{\beta\gamma} > 0$. The Schrödinger time evolution is not strictly local, so that particles originally confined to boxes in the initial state $\psi_0^{\beta\gamma}$, immediately leak into all other boxes as soon as $t_{\beta\gamma} > 0$. However, this leakage is small, as long as not too much time has passed and the strength is decaying with growing distance to the respective box. This is exactly the statement of Lieb-Robinson bounds \cite{LiebRobinson.1972}, and in this paper, we will use a recent version of Lieb-Robinson bounds for the continuum \cite{hinrichs2023lieb} with an almost linear light cone and arbitrarily strong power-law decay in a slightly modified form, cf.~\Cref{th:modified lrb}. The exact bound of the second derivative can be found in \Cref{th:second order}.

Combining the first order result \Cref{th:first_order} with the second order estimate \Cref{th:second order} then gives
\begin{align}
	\nonumber
	\langle f^\alpha_\mathbf{j}, V f^\alpha_\mathbf{j} \rangle
	 & = \sum_{(\beta,\gamma)}\sigma_\alpha(\beta, \gamma) \frac{ \Dd_\j^{\beta\gamma}(t) - \Dd_\j^{\beta\gamma}(0) }{t} \\
	 & \quad - 2\ell^{-2}\langle f, (-\Delta) f \rangle +  \cO(\ell^{-\gd} \abs{t} ), \label{eq:direct V}
\end{align}
with $\gd := 4d + 10$,
see also \Cref{th:final_fVf_D_estimate} for the detailed theorem. Notice that the error scales in inverse powers of the scale of the state preparation and measurement $\ell$. This is inevitable as smaller scales of the system entail higher Lieb-Robinson velocities and therefore smaller decay of the error. In the case of one box, we do not need Lieb-Robinson bounds and the error only scales like $\cO(\ell^{-4} \abs{t} )$, cf.~\Cref{th:second order} (a).

\begin{proof}[Proof of \texorpdfstring{\Cref{thm:sampling}}{???}]
	Let the random variables $Y_\j^{\alpha\beta,(1)}, \dots, Y_\j^{\alpha\beta,(T)}$, denote the $T$ measurement outcomes of the two boxes $\alpha,\beta \in \{0,1,2\}$ inside the triple $\j \in \Jf$ after time $t$. We aim to estimate the expectations $\mathbb{E}[Y_\j^{\alpha\beta,(k)}] = \Dd^{\alpha\beta}_\j(t)$ via the arithmetic mean estimators
	\begin{align}
		\label{eq:mean estimator}
		\widehat Y^{\alpha\beta}_\j := \frac{1}{T} \sum_{k=1}^T Y_\j^{\alpha\beta,(k)} .
	\end{align}
	The probability that the estimators are not equal to its expectation values up to an error of order $\varepsilon > 0$ can be bounded with a union bound and Hoeffding's inequality \cite{hoeffding1963probability}. More precisely,
	\begin{align}
		\IP & \left( \exists \j \in \Jf, ~(\alpha,\beta) \,|\, \abs{\widehat Y^{\alpha\beta}_\j - \IE Y^{\alpha\beta,(k)}_\j } \geq \varepsilon \right)
		\nonumber                                                                                                                    \\ &\leq  3 \abs{\Jf} \max_{\j \in \Jf, ~(\alpha,\beta)}\IP\left( \abs{\widehat Y^{\alpha\beta}_\j - \IE Y^{\alpha\beta,(1)}_\j } \geq \varepsilon \right) \nonumber \\
		    & \leq 6 \abs{\Jf}  \exp \left( - 2 \varepsilon^2 T \right). \label{eq:failure prob}
	\end{align}
	Requiring the success probability to be at least $1-\delta$ means that \Cref{eq:failure prob} must be smaller than $\delta$. Hence, in order to obtain the $\Dd_\j^{\alpha\beta}$ with an error of at most $\varepsilon$, the number of samples needs to be chosen as
	\begin{align}
		\label{eq:T_est first}
		T \geq \frac{1}{2 \varepsilon^2} \log\left( \frac{6 \abs{\Jf} }{\delta} \right).
	\end{align}
	Now we can compute estimators for the local averages $\omega_\j$ by \Cref{eq:direct V}. In order to get a total maximum error of order $\cO(\varepsilon)$, we have to choose $\abs{t} = \cO( \ell^{\gd} \varepsilon)$ and thus require an error $\cO(\ell^{\gd} \varepsilon^2)$ for the terms $\Dd_\j^{\alpha\beta}(t)$. Substituting $\varepsilon$ with $\ell^{\gd} \varepsilon^2$ in \Cref{eq:T_est first} yields a total sample complexity of
	\begin{align}
		\label{eq:T_est}
		T = \cO\left(  \ell^{-2\gd} \varepsilon^{-4} \log\left( \frac{\abs{\Jf} }{\delta} \right) \right),
	\end{align}
	with  corresponding total evolution time $T_c =t T = \cO\left(  \ell^{-\gd} \varepsilon^{-3} \log\left( \frac{\abs{\Jf} }{\delta} \right) \right)$.
\end{proof}

\section{Conclusions} We have developed a unified and modular framework for Hamiltonian learning in the first-principles position space of quantum particles, $\mathbb R^d$. We consider free fermions and show that local short-time measurements can be used to recover both Coulomb potentials and other broad classes of sufficiently regular external potentials. By combining tailored reconstruction procedures with continuum localization techniques and information-propagation bounds, our approach addresses the main difficulties caused by the infinite-dimensional setting and the unbounded Laplacian. These results provide a foundation for scalable Hamiltonian learning in  continuum models and suggest a natural path toward learning more general interacting systems as continuum locality bounds improve.

\section*{Data Availability}
The numerical code, input data, and analysis scripts that support the findings of this study are publicly available in the Zenodo repository "Learning Coulomb Potentials and Beyond with Fermions in Continuous Space" \cite{DataZenodo.2026} and via the associated \href{https://github.com/MitTimM/learning_coulomb_potentials_in_continuous_space}{GitHub repository} linked from that Zenodo record. The Zenodo archive includes simulation code, example input files, and postprocessing scripts; the software is released under the Creative Commons Zero (CC0 1.0 Universal) waiver.

\section*{Acknowledgments}
This work was funded by the ANR project PraQPV, grant number ANR-24-CE47-3023 (AB); the ERC Starting Grant “MathQuantProp” No.\ 101163620 (ML)\footnote{Views and opinions expressed are however those of the authors only and do not necessarily reflect those of the European Union or the European Research Council Executive Agency. Neither the European Union nor the granting authority can be held responsible for them}; the Deutsche Forschungsgemeinschaft (DFG, German Research Foundation) - Project-ID 470903074 - TRR 352 (ML, TM), QuantERA II Programme that has received funding from the EU’s H2020 research and innovation programme under the GA No 101017733 (TM)

\newpage

\widetext

\bibliography{lit}

\appendix

\section{Technical Proofs --- Learning Coulomb Potentials}\label{appx-sec:Coulomb}
We first focus on the technical results required to determine the position and charge number of a single Coulomb potential (see \Cref{sec:Coulomb}).  We then provide a complete proof of \Cref{thm:single-Coulomb}. For this, we repeat the short proofs from \Cref{sec:Coulomb} and include additional steps to present a complete proof.
\begin{theorem}\label{appx-thm:single-Coulomb}
	Let $\varepsilon>0$ with $\varepsilon<\Lambda_*\frac{\ell}{4L^2}$ be the given precision and let $V$ be a single Coulomb potential
	\begin{equation*}
		V:\R^3\to\R, \quad V(x)=\frac{\lambda}{\|x-y\|}
	\end{equation*}
	for some unknown $\lambda$ with $|\lambda|\in[\Lambda_*,\Lambda^*]\subset(0,\infty)$ and $y\in[0,L]^3$. Then, using \Cref{protocol}, we obtain estimators $\widehat{\lambda}$ and $\widehat{y}$ such that, with success probability at least $(1-\delta)\in(0,1)$,
	\begin{equation*}
		|\widehat{\lambda}-\lambda| \leq \varepsilon,\qquad\|\widehat{y}-y\|\leq \varepsilon,
	\end{equation*}
	requiring a total evolution time of
	\begin{equation*}
		T_{\operatorname{c}}=\cO\Bigl(\varepsilon^{-3}\ln\biggl(\frac{1}{\delta}\biggr)\Bigr)\,.
	\end{equation*}
\end{theorem}
\begin{proof}
	Define $\ell=L/m$ with $m=8$. Due to the identity of Newton's shell theorem (see \cite[Thm.~9.7]{LiebLoss.2001}), the proof reduces to solving a system of non-linear equations with four degrees of freedom:
	\begin{equation*}
		\omega_{\j}=\frac{\lambda}{\nn{p_{\j}-y}}\qquad\text{or}\qquad\frac{\lambda^2}{\omega_{\j}^2}=z_\j
	\end{equation*}
	for $\j\in\{0:3\}^3$ and $z_\j=\nn{p_{\j}-y}^2$. Here it is important to mention that in the following, $p_\j$ is chosen in such a way that the assumptions of Newton’s theorem are satisfied (see \cite[Thm.~9.7]{LiebLoss.2001}), that is $y\notin B_{\j}$. By definition, the sign of $\omega_j$ directly defines the sign of $\lambda$, w.l.o.g., $\lambda>\Lambda_*>0$.  Next, we subtract $z_{\mathbf{j}}-z_{\mathbf{i}}$ and rewrite the norms via the underlying scalar product. For simplicity, we define
	\begin{equation}\label{appx-eq:single-Coulomb-redefinition}
		\begin{aligned}
			p_{\mathbf{i}\j}    & =\frac{1}{2}(p_{\mathbf{i}}-p_{\j})                       \\
			v_{\mathbf{i}\j}    & =(\|p_\mathbf{i}\|^2-\|p_\j\|^2)                          \\
			\eta_{\mathbf{i}\j} & =(\frac{1}{\omega_\mathbf{i}^2}-\frac{1}{\omega_\j^2})\,,
		\end{aligned}
	\end{equation}
	which results in
	\begin{equation*}
		v_{\mathbf{i}\j}=\eta_{\mathbf{i}\j}\lambda^2+\cs{p_{\mathbf{i}\j},y}\,.
	\end{equation*}
	For simplicity, we label at most eight necessary points by $\j\in\{0:7\}$. This means that the label $\mathbf{i}$ of a grid point is used interchangeably with the numbering $\{1:8\}$. Then, these points and the corresponding estimated local averages reduce to the linear system
	\begin{equation}\label{appx-eq:matrix-LGS}
		\mathbf{v}\coloneqq
		\begin{pmatrix}
			v_{12} \\
			v_{34} \\
			v_{56} \\
			v_{78}
		\end{pmatrix}
		=
		\begin{pmatrix}
			\eta_{12} &  & p_{12}^T & \\
			\eta_{34} &  & p_{34}^T & \\
			\eta_{56} &  & p_{56}^T & \\
			\eta_{78} &  & p_{78}^T &
		\end{pmatrix}
		\begin{pmatrix}
			\lambda^2 \\
			y_1       \\
			y_2       \\
			y_3
		\end{pmatrix}
		\eqqcolon A\widetilde{\mathbf{y}}\,.
	\end{equation}
	Here it is important to mention that the matrix $\mathbf{A}$ and the vector $\mathbf{v}$ are characterized up to a small error. As before, note that $(\cdot)^T$ denotes the transpose of the vectors. To ensure a well-conditioned linear problem above, the algorithm looks as follows: First the area $[0,L]^3$ is partitioned in $8^3=2^9=512$ boxes and partitioned (with the exception of two boxes, which are combined with another neighboring, already-considered box to form a triple) into triples as explained in \Cref{fig:initial-states} and \ref{fig:sequential_scheme}. Then, \Cref{protocol} is executed but only with $\Jf$ being one triple (see \Cref{fig:sequential_scheme}).
	\begin{figure}[ht!]
		\includegraphics[scale=0.8]{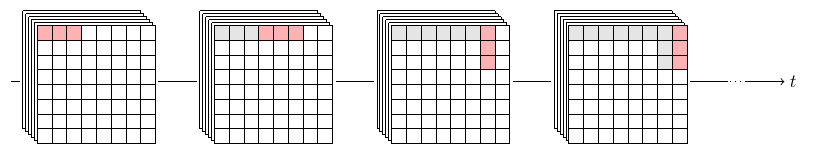}
		\caption{Demonstration of the sequential scheme, in which only one triplet is fixed, then \Cref{protocol} is executed, and the process continues with the next triplet. This approach is only efficient when the number of boxes is small.}
		\label{fig:sequential_scheme}
	\end{figure}
	This is then repeated for every triple, meaning that the sample complexity is $171$ times the total evolution time achieved in \Cref{thm:sampling}:
	\begin{equation*}
		T_c=\cO\Bigl(\varepsilon^{-3}\ln\biggl(\frac{1}{\delta}\biggr)\Bigr)\,.
	\end{equation*}
	Via this sequential scheme, we found estimators $\widehat{\omega}_{\j}$ up to precision $\varepsilon$ and probability $(1-\delta)$ for the weighted averages $\omega_{\j}$ uniformly over the partition of $512$ boxes. It is important to note that $m=8$ is fixed for the entire proof and is independent of the precision $\varepsilon$.

	From there, we calculate the variables defined in \Cref{appx-eq:single-Coulomb-redefinition} and choose four vectors $p_{\mathbf{i}\j}$ so that the linear system of equations is well-conditioned, i.e.~errors propagate linearly. First, we estimate the error in which the center of the Coulomb lies approximately by the simple optimization:
	\begin{equation}\label{appx-eq:maximization-y}
		\widehat{\j}_{y}\coloneqq\argmax_{\j\in\{0:7\}^3}(\widehat{\omega}_{\j})
	\end{equation}
	which takes a random element as single output if several maxima exist. Then, the assumption $\varepsilon<\Lambda_*\frac{\ell}{4L^2}$ allows us to infer that the Coulomb center lies in the neighborhood
	\begin{equation}\label{eq:closure}
		\partial B_{\hat{\j}_y}\coloneqq\bigcup_{\dist(\j, \hat{\j}_{y})\leq1}B_{\j}
	\end{equation}
	of $\widehat{\j}_{y}$. This can be seen by contraposition. If $y\notin\partial B_{\hat{\j}_y}$, we can choose $\mathbf{j_y}\in\{1:m\}^3$ such that $p_{\mathbf{j_y}}\in B_{\hat{\j}_y}\partial\coloneqq \partial B_{\hat{\j}_y}\backslash B_{\hat{\j}_y}$ and $\mathbf{j}_y$ is the label of $B_{\mathbf{j}_y}$, in which the path $[0,1]\ni s\mapsto sy+(1-s)p_{\hat{\mathbf{j}}_y}$ first enters (see \Cref{fig:contraposition_single_C}). 
    \begin{figure}[ht!]
		\includegraphics[scale=0.8]{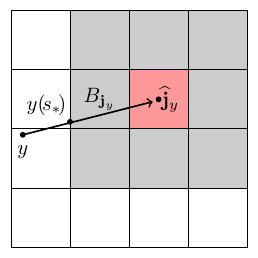}
		\caption{Visualization of path from $y$ to $p_{\hat{\j}_y}$ and first intersection $y(s_*)$ with $B_{\j_y}$.}
		\label{fig:contraposition_single_C}
	\end{figure}
    Since $y\notin\partial B_{\hat{\j}_y}$, we can apply Newton's shell theorem \cite[Thm.~9.7]{LiebLoss.2001}, which shows
    \begin{align}
		\omega_{\hat{\j}_y}-\omega_{\j_y} & =\int_{\R^3}|f_{\hat{\j}_y}(x)|^2\frac{\lambda}{\nn{x-y}}dx-\int_{\R^3}|f_{\j_y}(x)|^2\frac{\lambda}{\nn{x-y}}dx\nonumber\\
        & = \frac{\lambda}{\|p_{\hat{\j}_y}-y\|}-\frac{\lambda}{\|p_{\j_y}-y\|}\nonumber\\
        & = \lambda\frac{\|p_{{\j}_y}-y\|-\|p_{\hat{\j}_y}-y\|}{\|p_{\hat{\j}_y}-y\|\,\|p_{{\j}_y}-y\|}\nonumber
	\end{align}
    Next, we note that the maximal distance from the first intersection of $[0,1]\ni s\mapsto y(s)=sy+(1-s)p_{\hat{\mathbf{j}}_y}$ with $\partial B_{\j_y}$, i.e.~$y(s_*)$ with $s_*=\max\{s\in[0,1]\,|\,\min_{x\in \partial B_{\hat{\j}_y}}\|y(s)-x\|=0\}$ (see \Cref{fig:contraposition_single_C}), is upper bounded by
    \begin{equation*}
        \|y(s_*)-p_{\j_y}\|\leq \frac{\sqrt{3}}{2}\ell\,.
    \end{equation*}
    Then, the triangle inequality shows
    \begin{equation*}
        \|p_{\j_y}-y\|\leq \|y-y(s_*)\|+\|y(s_*)-p_{\j_y}\|\leq \|y-y(s_*)\|+\frac{\sqrt{3}}{2}\ell
    \end{equation*}
    and 
    \begin{equation*}
        \|p_{\hat{\j}_y}-y\|=\|y-y(s_*)\|+\|y(s_*)-p_{\hat{\j}_y}\|\geq\|y-y(s_*)\|+\frac{\sqrt{3}}{2}\ell+\frac{\ell}{2}\,.
    \end{equation*}
    The last inequality holds true, since from $y(s_*)$ there is a distance of at least $\frac{\sqrt{3}}{2}\ell$ to the boundary of $B_{\hat{\j}_y}$ and from there at least $\frac{\ell}{2}$ to $p_{\hat{\j}_y}$.
    Together, these inequalities show
    \begin{align*}
		\omega_{\hat{\j}_y}-\omega_{\j_y}
        \leq\lambda\frac{-\ell}{2\|p_{\hat{\j}_y}-y\|\,\|p_{{\j}_y}-y\|}\leq\lambda\frac{-\ell}{2L^2}<-2\varepsilon
	\end{align*}
	This contradicts the maximization in \Cref{appx-eq:maximization-y} and shows that $y\in\partial B_{\hat{\j}_y}$.
	\begin{figure}[ht!]
		\includegraphics[scale=0.8]{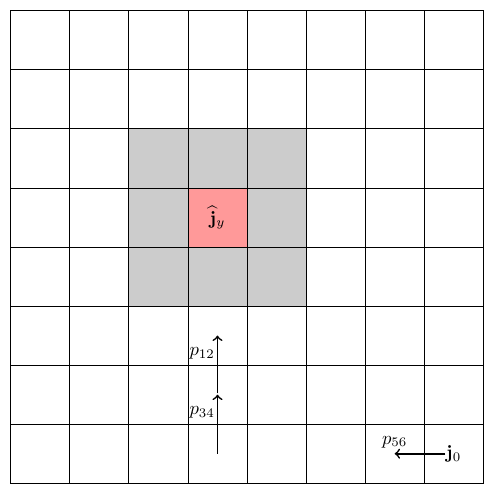}
		\caption{Visualization in two dimensions of a first estimator $p_{\hat{\j}_y}$ of $y$ indexed by $\widehat{\j}_y$, its neighborhood $\partial B_{\hat{\j}_y}$ in gray and the construction of the defining vectors $p_{12}, ..., p_{78}$ of the matrix $\mathbf{A}$ in \Cref{appx-eq:matrix-LGS}.}
		\label{fig:optimization_single_C}
	\end{figure}
	Next, we define the points $p_{\i\j}$ in such a way that the matrix in \Cref{appx-eq:matrix-LGS} is explicitly invertible. First, define the point of greatest distance to $\partial B_{\hat{\mathbf{j}}_y}$ by
	\begin{equation}\label{appx-eq:multi-mode-start}
		\j_0 \coloneqq \argmax_{\j\in\{0:7\}}\|\hat{\j}_y-\j\|
	\end{equation}
	and the direction compared to the canonical basis vectors by
	\begin{equation}\label{eq-appx:overlap}
		\begin{aligned}
			e_1 & =\mathrm{sgn}(\cs{(1,0,0),p_{\j_0}-p_{\hat{\j}_y}})(1,0,0)\,, \\
			e_2 & =\mathrm{sgn}(\cs{(0,1,0),p_{\j_0}-p_{\hat{\j}_y}})(0,1,0)\,, \\
			e_3 & =\mathrm{sgn}(\cs{(0,0,1),p_{\j_0}-p_{\hat{\j}_y}})(0,0,1)\,.
		\end{aligned}
	\end{equation}
	These allow us to pick the points $p_{\j_0}\!\!=\!p_8\!=\!p_6$, $p_{\j_0+e_1}\!\!=p_7$, $p_{\j_0+e_2}\!\!=p_5$, $p_{\hat{\j}_y-2e_3}\!\!=p_3$, $p_{\hat{\j}_y-3e_3}\!\!=\!p_4\!=\!p_1$ and $p_{\hat{\j}_y-4e_3}\!\!=p_2$ (the label $\mathbf{j}\in \{ 1:m\} ^3$ and $j\in \{ 1:8\}$  is simply an alternative notation), which define the direction vector $p_{\mathbf{i}\j}$ of the matrix $\mathrm{A}$ (see \Cref{fig:optimization_single_C}):
	\begin{equation}\label{appx-ex:matrix-p-vec}
		\begin{aligned}
			p_{(\hat{\j}_y-3e_3)(\hat{\j}_y-4e_3)} & =\frac{\ell}{2}e_3=p_{12}    \\
			p_{(\hat{\j}_y-2e_3)(\hat{\j}_y-3e_3)} & =\frac{\ell}{2}e_3=p_{34}    \\
			p_{(\j_0+e_2)\j_0}                     & =\frac{\ell}{2}e_2=p_{56}    \\
			p_{(\j_0+e_1)\j_0}                     & =\frac{\ell}{2}e_1=p_{78}\,.
		\end{aligned}
	\end{equation}
	Moreover, the values $\eta_{(\hat{\j}_y-2e_3)(\hat{\j}_y-3e_3)}\!\!=\eta_{34}$ and $\eta_{(\hat{\j}_y-3e_3)(\hat{\j}_y-4e_3)}\!\!=\eta_{12}$ satisfy
	\begin{equation}\label{appx-eq:lower-bound-single-C}
		\begin{aligned}
			\eta_{(\hat{\j}_y-2e_3)(\hat{\j}_y-3e_3)}-\eta_{(\hat{\j}_y-3e_3)(\hat{\j}_y-4e_3)} & = \frac{1}{\omega_{(\hat{\j}_y-2e_3)}^2}-\frac{2}{\omega_{(\hat{\j}_y-3e_3)}^2}+\frac{1}{\omega_{(\hat{\j}_y-4e_3)}^2}\\
            & = \frac{1}{\lambda^2}\Bigl(\|p_{\hat{\j}_y-2e_3}-y\|^2-2\|p_{\hat{\j}_y-3e_3}-y\|^2+\|p_{\hat{\j}_y-4e_3}-y\|^2\Bigr)\\
            & = \frac{1}{\lambda^2}\Bigl(\|(p_{\hat{\j}_y}-y)-\frac{2\ell}{2}e_3\|^2-2\|(p_{\hat{\j}_y}-y)-\frac{3\ell}{2}e_3\|^2+\|(p_{\hat{\j}_y}-y)-\frac{4\ell}{2}e_3\|^2\Bigr)\\
            & = \frac{1}{\lambda^2}\biggl(\Bigl(\frac{2\ell}{2}\Bigr)^2-2\Bigl(\frac{3\ell}{2}\Bigr)^2+\Bigl(\frac{4\ell}{2}\Bigr)^2+(-2\ell+6\ell-4\ell)\braket{p_{\hat{\j}_y}-y\,,\,e_3}\biggr)\\
			&=\frac{\ell^2}{2\lambda^2}
		\end{aligned}
	\end{equation}
	by applying Newton's shell theorem \cite[Thm.~9.7]{LiebLoss.2001} in the second equality. Clearly, this choice is motivated by the invertibility of the matrix $\mathbf{A}$. Due to the construction above and the analysis of the sign of $\lambda$ at the very beginning of the proof, we can calculate
	\begin{equation}\label{appx-eq:solve-LGS-single-mode}
		\begin{aligned}
			|\lambda| & = \sqrt{\Bigl|\frac{v_{34}-v_{12}}{\eta_{34}-\eta_{12}}\Bigr|} \\
			y_1       & =\frac{2}{\ell}\Bigl(v_{34}-\eta_{34}\lambda^2\Bigr)                                     \\
			y_2       & =\frac{2}{\ell}\Bigl(v_{56}-\eta_{56}\lambda^2\Bigr)                                     \\
			y_3       & =\frac{2}{\ell}\Bigl(v_{78}-\eta_{78}\lambda^2\Bigr)\,.
		\end{aligned}
	\end{equation}
	Due to the construction above --- in particular the lower bound in \Cref{appx-eq:lower-bound-single-C} and the numerical stability of floating‑point operations --- we obtain a numerically stable solution; that is, initial errors propagate only linearly through the algorithm. This completes the proof of the theorem.
\end{proof}
\begin{remark}
	The proof not only establishes the statement but also provides an explicit construction for a possible algorithm to determine the charge number and position of the Coulomb potential. However, several steps in the proof, while necessary for readability, can be numerically optimized. For example, the relationship between the number of boxes in the partition and the condition number of the matrix $\mathrm{A}$ is not optimized at the moment.
\end{remark}

\begin{remark}\label{appx-rmk:extension-single-mode}
	Based on the result above, it is evident that the local averages decrease as the distance $r$ increases. This observation allows one to generalize \Cref{appx-thm:single-Coulomb} to arbitrary rotationally invariant potentials --- without applying Newton's shell theorem. To do so, one optimizes the following difference over $r$ (e.g., via gradient descent):
	\begin{equation*}
		\int_{\R^d} |f_{\j}(x)|^2 V(rv) - \widehat{\omega}_{\j}
	\end{equation*}
	for three distinct points $\mathbf{j}$ that do not lie on a straight line. This yields three spheres, which are intersected in the subsequent step. The intersection approximates the center of the potential well. Consequently, the above method can be extended to any rotationally invariant algorithm in arbitrary dimensions.
\end{remark}

\subsection{Multi-Coulomb Learning}
Next, we present the complete proof of \Cref{thm:multi-Coulomb}, which is based on the idea of first finding a weak approximation of the Coulomb centers and charge numbers. This is achieved using a fundamental result from the perturbation theory of Coulomb centers, along with methods that constructs a diagonally dominant matrix approximating the charge numbers. In the next step, we reduce the problem separately to isolated single-Coulomb problems, in which we apply \Cref{thm:single-Coulomb} to improve the result by a factor of $\frac{1}{3}$. This method is then bootstrapped to obtain estimates for the positions and charge numbers in
\begin{equation}\label{appx-eq:multi-Coulomb}
	V(x) = \sum_{k=1}^{K} \frac{\lambda_k}{\nn{x - y_k}}\,.
\end{equation}
For the sake of clarity, we restate \Cref{thm:multi-Coulomb} along with the convention (see \Cref{eq:closure}) that
\begin{equation}\label{appx-eq:convention-boundary}
	\begin{aligned}
		\partial_{k}B_{\j} & = \bigcup_{\dist(\j,\mathbf{i})\leq k} B_{\mathbf{i}} \\
		B_{\j}\partial_{k} & = \partial_{k}B_{\j} \setminus \partial_{k-1}B_{\j}
	\end{aligned}
\end{equation}
for $k \in \{0 : m-1\}$, with $\partial_{-1}B_{\j} = \emptyset$. Moreover, we assume the minimal distance between two Coulomb centers to be
\begin{equation}\label{eq:min-center-dist}
	\min_{k.k'\in\{1:K\}}|y_k-y_{k'}|=y_*\,.
\end{equation}

\begin{theorem}\label{appx-thm:multi-Coulomb}
	Let $\varepsilon>0$ be the given precision, $L$, $\Lambda_*$, $\Lambda^*>0$, $K\in\N$ and $V$ be a multi-Coulomb potential
	\begin{equation*}
		V:\R^3\to\R, \quad V(x)=\sum_{k=1}^K\frac{\lambda_k}{\|x-y_k\|}
	\end{equation*}
	for unknown $K\in\N$, $\lambda_k\in[\Lambda_*,\Lambda^*]$ and $y_k\in[0,L]^3$ satisfying $\|y_k-y_{k'}\|\geq y_*$ for all $k\neq k'\in\{1:K\}$. For a fixed grid size $\ell \leq \frac{L}{m}$ with
	\begin{equation*}
		m=\Biggl\lceil L\max\left\{\Bigl(\frac{2K\Lambda^*}{\Lambda_*y_*}\Bigr)^{3}\!\!,\frac{12K\Lambda^*}{\Lambda_*y_*}, \bigl(64K\bigr)^{12}\!\!, \Bigl(\frac{90K}{y_*}\Bigr)^{12}\!\!,\Bigl(\frac{26}{y_*}\Bigr)^3\!\!,(144K(y_*\!+\!1)^2)^{12},\Bigl(\frac{4K}{y_*^2}\Bigr)^{12}\!\!, \Bigl(\frac{87\Lambda^* K}{y_*^2}\Bigr)^3\!\!, \bigl(44K^2(\Lambda^*\!+\!1)y_*\bigr)^{12}\right\}\Biggr\rceil\,,
	\end{equation*}
    we achieve estimators $\widehat{\lambda_i}$ and $\widehat{y_i}$ such that, with success probability at least $(1-\delta)\in(0,1)$,
	\begin{equation*}
		\max_{i\in\{1:K\}}\{|\widehat{\lambda}_i-\lambda_i|, |\widehat{y}_i-y_i|\}\leq \varepsilon,
	\end{equation*}
	requiring a total evolution time of
	\begin{equation*}
		T_{\operatorname{c}}=\cO\Bigl(L\,\mathrm{poly}(K, y_*, y_*^{-1}, \Lambda^*, \Lambda_*^{-1})\varepsilon^{-3}\ln\biggl(\frac{1}{\delta}\biggr)\Bigr)\,.
	\end{equation*}
\end{theorem}

Before diving into the proof, we emphasize that although we cannot estimate the values $V(p_{\j})$ at specific points and only observe local averages $\omega_{\j}$ for $\j \in \{0:m-1\}^3$, we nevertheless construct an algorithm that achieves precision $\varepsilon$, independent of the grid size $\ell$. Moreover, we view our result as a proof of principle, meaning that the constants are not optimized and leave room for numerical refinement.

\begin{proof}
	In a first step, we apply the data acquisition \Cref{protocol} with arbitrary grid size $\ell$. Then, \Cref{thm:sampling} provides estimators $\hat{\omega}_{\j}$ with precision $\varepsilon_\omega$ and probability $(1-\delta)$ for all $\j \in \{0:m-1\}^{3}$ requiring a total evolution time
	\begin{equation*}
		T_{\operatorname{c}}=\cO\left(\mathrm{poly}(\ell^{-1})  \varepsilon_\omega^{-3} \ln(\frac{m}{\delta}) \right)
	\end{equation*}
	and measurements after short time steps $t=\cO(\mathrm{poly}(\ell)\varepsilon)$.

	Second, we apply \Cref{prop:multi-Coulomb-stability-max} to achieve estimators ${p}_{\hat{\j}_k}$ for the Coulomb centers $\{y_k\}$. To find estimators ${p}_{\hat{\j}_k}$ satisfying the bound
	\begin{equation}\label{appx-eq:estimators-rough}
		\|{y_k-{p}_{\hat{\j}_k}}\|\leq\sqrt[3]{\ell}
	\end{equation}
	for all $k\in\{1:K\}$, we choose 
    \begin{equation}\label{appx-eq:proof-maximization-choice1}
		\varepsilon_\omega<c=\frac{K\Lambda^*}{2y_*}\qquad
	\end{equation}
    and require the assumption stated in \Cref{prop:multi-Coulomb-stability-max} (in \Cref{proof-eq:all_assumptions_ell}, we list all assumptions and connect those to $m$):
    \begin{equation}\label{appx-eq:proof-grid-step-size}
        \ell\leq \min\Bigl\{\frac{1}{64},\Bigl(\frac{y_*}{8}\Bigr)^3\!\!, \Lambda_*^3\Bigl(\max\Bigl\{\frac{2K\Lambda^*}{y_*}-\frac{2c}{2},0\Bigr\}\Bigr)^{-3}\!\!, \Bigl(\frac{2c}{288K\Lambda^*}\Bigr)^3, \frac{\Lambda_*}{2\sqrt{3}} \Bigl(\frac{2K\Lambda^*}{y_*}+2c\Bigr)^{-1}\Bigr\}
	\end{equation}
    Then, \Cref{appx-eq:maximization-multi} yields estimators ${p}_{\hat{\j}_k}$ satisfying \Cref{appx-eq:estimators-rough}. The estimators are robust with respect to small errors $\varepsilon_\omega$ induced by the estimation of $\widehat{\omega}_{\j}$ for all $\j\in\{0:m-1\}^3$ because of $\varepsilon_\omega<c$. By construction the index of the estimators ${\hat{\j}_k}$ satisfy
	\begin{equation}\label{appx-eq:proof-maximization-multi}
		\|\widehat{\omega}_{\hat{\j}_k} - \widehat{\omega}_{\j'}\| \geq 2c \qquad \text{for all} \qquad \j' \in \partial_{s_1} B_{\j_k}
	\end{equation}
    with $s_1 \coloneqq \lceil \ell^{-2/3} \rceil=\lceil\frac{\sqrt[3]{\ell}}{\ell}\rceil$. Note that \Cref{prop:multi-Coulomb-stability-max} shows that the distance between two estimators is at least $y_*-2\ell s_1>0$. This shows
	\begin{equation*}
		\|\omega_{\hat{\j}_k} - \omega_{\j'}\| \geq \bigl|\|\omega_{\j_k}-\widehat{\omega}_{\hat{\j}_k}\|-\|\widehat{\omega}_{\hat{\j}_k} - \omega_{\j'}\|\bigr|\geq 2c-\varepsilon_\omega\geq c
	\end{equation*}
    for all $\j' \in \partial_{s_1} B_{\hat{\j}_k}$. With the upper bound (see \Cref{proof-eq:all_assumptions_ell})
    \begin{equation}\label{appx-eq:ell-assum-1}
        \ell\leq \min\Bigl\{\frac{1}{64},\Bigl(\frac{y_*}{8}\Bigr)^3\!\!, \Bigl(\frac{2\Lambda_*y_*}{3K\Lambda^*}\Bigr)^{3}\!\!, \Bigl(\frac{1}{288y_*}\Bigr)^3, \frac{\Lambda_*}{2\sqrt{3}}\frac{y_*}{3K\Lambda^*}\Bigr\}
	\end{equation}
	which combines \Cref{appx-eq:proof-grid-step-size} and \Cref{appx-eq:proof-maximization-choice1}, we obtain a set of centers $\{{p}_{\hat{\j}_k}\}_{k=1}^K$ satisfying
	\begin{equation}\label{appx-eq:input-precision-lgs}
		\|{p}_{\hat{\j}_k}-y_k\|\leq \sqrt[3]{\ell}\eqqcolon \eta_y
	\end{equation}
	for all $k\in\{1:K\}$. Here, we introduce the label $\eta_y$ for the iteration process later in the proof. Note that \Cref{prop:multi-Coulomb-stability-max} also proves uniqueness in the sense that the number of Coulomb centers and the number of maximizers coincide.

	Third, we compute the charge numbers from the local averages and the approximated Coulomb centers by solving an approximation of the following linear system for $K$ points $j'_{k}$ satisfying $p_{\j'_{k}}\in B_{\hat{\j}_k}\partial_{s_3}$ with $s_3=s_1+\lceil\ell^{-11/12}\rceil+1$ for all $k\in\{1:K\}$ (see \Cref{fig:multi-coulomb-regions}): For $k'\in\{1:K\}$
    \begin{equation*}
		\omega_{\j'_{k'}}=\begin{pmatrix}
			\int_{\R^3}|f_{\j'_{k'}}(x)|^2\frac{1}{\nn{x-y_{1}}}dx & \cdots & \int_{\R^3}|f_{\j'_{k'}}(x)|^2\frac{1}{\nn{x-y_{k}}}dx & \cdots & \int_{\R^3}|f_{\j'_{k'}}(x)|^2\frac{1}{\nn{x-y_{K}}}dx
		\end{pmatrix}
		\begin{pmatrix}
			\lambda_1 \\
			\vdots    \\
			\lambda_k \\
			\vdots    \\
			\lambda_K
		\end{pmatrix}.
	\end{equation*}
	\begin{figure}[ht!]
		\includegraphics[scale=0.7]{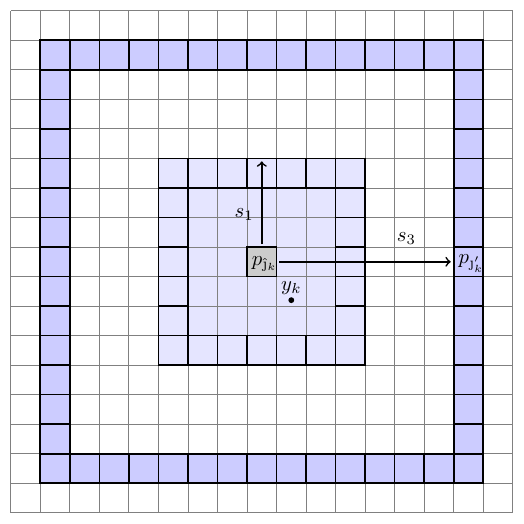}
		\caption{Schematic illustration of the choices of regions via $s_1$ and $s_3$ and associated point ${p}_{\hat{\j}_k}$, $y_k$, and $p_{\j'}$ used in the proof of \Cref{appx-thm:multi-Coulomb}}
		\label{fig:multi-coulomb-regions}
	\end{figure}
    The choice of the points $p_{j'_{k'}}$ balances two competing sources of error: the error introduced by approximating the true positions, which grows as $p_{j'_{k'}}$ approaches the singularity at $y_{k'}$, and the error arising from being sufficiently close to the singularities to obtain a diagonally dominant linear system. In the following, we use the notation
	\begin{equation}\label{appx-eq:diag-dom-matrix}
		\mathbb{M}\coloneqq\Bigl(\int_{\R^3}|f_{\j'_{k'}}(x)|^2\frac{1}{\nn{x-y_k}}dx\Bigr)_{k',k=1}^K\qquad\text{and}\qquad\widehat{\mathbb{M}}\coloneqq\Bigl(\int_{\R^3}|f_{\j'_{k'}}({x})|^2\frac{1}{\|{x-{p}_{\hat{\j}_k}}\|}dx\Bigr)_{k',k=1}^K\,.
	\end{equation}
	To show diagonal dominance of $\widehat{\mathbb{M}}$, \Cref{lem:diagonal-dominant-coulomb} requires the additional assumption
	\begin{equation}\label{appx-eq:cond-diag-dom}
		\max_{x\in B_{{\j}'_{k}}}\|{p}_{\hat{\j}_k}-x\|\leq \eta<\frac{y_*}{2(K-1)}\,.
	\end{equation}
    Due to the assumptions on $m$ (see \Cref{proof-eq:all_assumptions_ell}):
	\begin{equation}\label{appx-eq:proof-assumption2.1}
		\max_{x\in B_{{\j}'_{k}}}\|{p}_{\hat{\j}_k}-x\|\leq \sqrt{3}\ell s_3= \sqrt{3}(\ell\lceil\ell^{-2/3}\rceil+\ell\lceil\ell^{-11/12}\rceil+\ell)\leq\sqrt{3}(\sqrt[3]{\ell}+\sqrt[12]{\ell}+3\ell)\leq10\sqrt[12]{\ell}\leq\frac{5y_*}{16(K-1)}\,,
	\end{equation}
    which directly implies the above assumption with $\eta = \frac{5y_*}{16 (K-1)}$ and \Cref{lem:diagonal-dominant-coulomb}
	\begin{equation*}
		\|\widehat{\mathbb{M}}^{-1}\|\leq\frac{\eta y_*}{y_*-2\eta(K-1)}\leq\frac{y_*}{K-1}\,.
	\end{equation*} 
    This diagonal dominance of $\mathbb{M}$ allows us to solve the linear system $(\widehat{\omega}_{\j_k})_{k=1}^K = \widehat{\mathbb{M}}(\widehat{\lambda}_{k})_{k=1}^K$ and achieve a stable solution $(\widehat{\lambda}_{k})_{k=1}^K$ (see \cite{fawzilecture.2022numerical}). 
    
    Next, we prove the error made by solving the approximated linear system instead of the original one. For that, we need to find bounds on the following terms
	\begin{equation}\label{appx-eq:matric-inversion-error}
		\begin{aligned}
			\|\lambda-\widehat{\lambda}\|_\infty & \leq\|\mathbb{M}^{-1}\omega-\widehat{\mathbb{M}}^{-1}\widehat{\omega}\|_{\infty}\\
            & \leq\|\mathbb{M}^{-1}-\widehat{\mathbb{M}}^{-1}\|_{\infty\rightarrow\infty}\,\|\omega\|_{\infty}+\|\widehat{\mathbb{M}}^{-1}\|_{\infty\rightarrow\infty}\,\|\omega-\widehat{\omega}\|_{\infty}\\
            & \leq\|\mathbb{M}^{-1}\|_{\infty\rightarrow\infty}\,\|\mathbb{M}-\widehat{\mathbb{M}}\|_{\infty\rightarrow\infty}\,\|\widehat{\mathbb{M}}^{-1}\|_{\infty\rightarrow\infty}\,\|\omega\|_{\infty}+\|\widehat{\mathbb{M}}^{-1}\|_{\infty\rightarrow\infty}\,\varepsilon_\omega\,,
		\end{aligned}
	\end{equation}
    where $\|A\|_{p\rightarrow p}=\sup_{x\in\mathbb{R}^K}\frac{\|Ax\|_p}{\|x\|_p}$ denotes the matrix induced $p$-norms for $p\in\{1,...,\infty\}$, in particular $\|A\|_{2\rightarrow2}=\|A\|$, for matrices $A\in\mathbb{R}^{K\cross K}$. Note that only bounds on the infinity norm are required.

    Next, we repeat the diagonal dominance proof for $\mathbb{M}$, which follows similarly by (see \Cref{fig:multi-coulomb-regions} and \Cref{appx-eq:proof-assumption2.1}) 
    \begin{equation}\label{appx-eq:proof-assumption2.2}
		\max_{x\in B_{\j'_{k}}}\|y_{k}-x\|\leq \sqrt{3}(s_3+1+s_1)\ell\leq7\sqrt{3}\sqrt[12]{\ell}\leq14\sqrt[12]{\ell} \leq \frac{5y_*}{16(K-1)}=\eta<\frac{y_*}{K-1}\,,
	\end{equation}
    holds due to the assumptions on $\ell$ (see \Cref{proof-eq:all_assumptions_ell}) applied in the fourth inequality. Then, \Cref{lem:diagonal-dominant-coulomb} show diagonal dominance, in particular invertibility, and 
    \begin{equation*}
		\nn{\mathbb{M}^{-1}}\leq\frac{y_*}{(K-1)}\,.
	\end{equation*}
    Note that the inequalities $\|\mathbb{M}^{-1}x\|_\infty\leq\|\mathbb{M}^{-1}x\|_2\leq\|\mathbb{M}^{-1}\|\,\|x\|_2$ and $\|x\|_2\leq\sqrt{K}\|x\|_\infty$ directly prove
    \begin{equation*}
        \|\mathbb{M}^{-1}\|_{\infty\rightarrow\infty}=\sup_{x\in\mathbb{R}^K}\frac{\|\mathbb{M}^{-1}x\|_\infty}{\|x\|_\infty}\leq\sqrt{K}\sup_{x\in\mathbb{R}^K}\frac{\|\mathbb{M}^{-1}x\|_2}{\|x\|_2}=\sqrt{K}\|\mathbb{M}^{-1}\|\,.
    \end{equation*}
    Together with $\sqrt{K}\leq K-1$ for $k\geq2$, this shows directly
    \begin{equation*}
		\|\widehat{\mathbb{M}}^{-1}\|_{\infty\rightarrow\infty},\,\|\mathbb{M}^{-1}\|_{\infty\rightarrow\infty}\leq\frac{\sqrt{K}y_*}{K-1}\leq y_*\,.
	\end{equation*}
    Moreover,
    \begin{equation*}
        \|\omega\|_{\infty}\leq\max_{k'\in\{1:K\}}\sum_{k=1}^K\int_{\mathbb{R}^3}|f_{\mathbf{j}'_{k'}}(x)|^2\frac{|\lambda_k|}{\|x-y_k\|}dx\leq\frac{\Lambda^* K}{\ell(s_3-s_1)} \leq \frac{\Lambda^*K}{\sqrt[12]{\ell}}
    \end{equation*}
    with $s_3=s_1+\lceil\ell^{-11/12}\rceil+1$ (see \Cref{fig:multi-coulomb-regions}). Finally, we prove an upper bound on the difference $\|\mathbb{M}-\widehat{\mathbb{M}}\|_{\infty\rightarrow\infty}$ for all $k,k'\in\{1:K\}$
	\begin{equation}\label{appx-eq:difference-diag-dom-mat}
		\begin{aligned}
			\Bigl|\int_{\R^3}|f_{\j'_{k'}}(x)|^2\frac{1}{\|x-y_k\|}\d x-\int_{\R^3}|f_{\j'_{k'}}(x)|^2\frac{1}{\|x-{p}_{\hat{\j}_k}\|}\d x\Bigr| & \leq\int_{\R^3}|f_{\j'_{k'}}(x)|^2\frac{\bigl|\|x-{p}_{\hat{\j}_k}\|-\|x-y_k\|\bigr|}{\|x-{p}_{\hat{\j}_k}\|\|x-y_k\||}\d x \\
            & \leq\int_{\R^3}|f_{\j'_{k'}}(x)|^2\frac{\|{p}_{\hat{\j}_k}-y_k\|}{\sqrt[6]{\ell}} \d x\\
            & \leq\frac{\eta_y}{\sqrt[6]{\ell}}
		\end{aligned}
	\end{equation}
	because $\|x-{p}_{\hat{\j}_k}\|, \|x-y_k\|\geq \sqrt[12]{\ell}$ for all $x\in B_{\j'_k}$ (see \Cref{fig:multi-coulomb-regions}) and $\|{p}_{\hat{\j}_k}-y_k\|\leq \sqrt[3]{\ell}= \eta_y$ (see \Cref{appx-eq:input-precision-lgs}). Since $\|A\|_{\infty\rightarrow\infty}=\max_{k'\in\{1:K\}}\sum_{k=1}^K|A_{k'k}|$, this shows
	\begin{equation*}
		\|\mathbb{M}-\widehat{\mathbb{M}}\|_{\infty\rightarrow\infty}\leq K\frac{\eta_y}{\sqrt[6]{\ell}}\,.
	\end{equation*}
	This provides us with estimators $\widehat{\lambda}_{k}$ satisfying for all $k\in\{1:K\}$
	\begin{equation}\label{appx-eq:lambda-diag-dom}
        \begin{aligned}
            |\lambda_k-\widehat{\lambda}_k|&\leq \frac{Ky_*^2}{(K-1)^2}K\frac{\eta_y}{\sqrt[6]{\ell}}\frac{K\Lambda^*}{\sqrt[12]{\ell}}+y_*\varepsilon_\omega\\
            &\leq 2K\Lambda^*y_*^2\frac{\eta_y}{\ell^{3/12}}+y_*\varepsilon_\omega
        \end{aligned}
	\end{equation}
	with probability $(1-\delta)$.

	Fourth, we improve the charge numbers $\lambda$ and Coulomb centers iteratively by isolating each center and applying the single Coulomb result (\Cref{appx-thm:single-Coulomb}). For that, recall that the estimated centers ${p}_{\hat{\j}_{k}}$ satisfy 
    \begin{equation*}
		\|y_{k}-\widehat{y}_{k}\|\leq\eta_y\quad\text{and}\quad|\lambda_{k}-\widehat{\lambda}_{k}|\leq\frac{2K\Lambda^*y_*^2\eta_y}{\ell^{3/12}}+y_*\varepsilon_\omega\eqqcolon\eta_\lambda\,.
	\end{equation*}
    for all $k\in\{1:K\}$ and $\widehat{y}_k = p_{\hat{\mathbf{j}}_k}$ with $\hat{\mathbf{j}}_k$ such that $\widehat{y}_k\in B_{\hat{\mathbf{j}}_k}$. This notation is introduced because the estimator for $y_k$, which is $p_{\hat{\mathbf{j}}_k}$ in the first step, will be refined iteratively. In the following, we isolate the Coulomb center $k'\in\{1:K\}$ by estimating the following weighted local averages of single Coulomb values defined as
	\begin{equation*}
		\omega_{\j}^{(k')}\coloneqq\frac{1}{\ell^{-1/3}}\int_{\R^3}|f_\j({x})|^2\frac{\lambda_{k'}}{\nn{x-y_{k'}}}\d x=\frac{\omega_{\j}}{\ell^{-1/3}}-\frac{1}{\ell^{-1/3}}\sum_{k\neq k'}\int_{\R^3}|f_{\j}(x)|^2\frac{{\lambda}_{k}}{\nn{x-y_k}}\d x\,.
	\end{equation*}
    Here it is important that the points with index $\j$ are chosen in an annulus around the estimated Coulomb center (not to be confused with $\j'_{k'}$), i.e., $\j \in B_{\hat{\j}_k}\partial_{s_1+1}^{3}\coloneqq B_{\hat{\j}_k}\partial_{s_1+1}\cup B_{\hat{\j}_k}\partial_{s_1+2}\cup B_{\hat{\j}_k}\partial_{s_1+3}$. The above values are estimated by
	\begin{equation*}
		\begin{aligned}
			\widehat{\omega}_{\j}^{(k')}=\frac{\widehat{\omega}_{\j}}{\ell^{-1/3}}-\frac{1}{\ell^{-1/3}}\sum_{k\neq k'}\int_{\R^3}|f_{\j}(x)|^2\frac{\widehat{\lambda}_{k}}{\nn{x-\widehat{y}_{k}}}\d x\,,
		\end{aligned}
	\end{equation*}
	which approximate $\omega_{\j}^{(k')}$ well because
	\begin{equation*}
		\begin{aligned}
			|\omega_{\j}^{(k')}-\widehat{\omega}_{\j}^{(k')}| & \leq \frac{\varepsilon_\omega}{\ell^{-1/3}}+\frac{1}{\ell^{-1/3}}\sum_{k\neq k'}\int_{\R^3}|f_{\j}(x)|^2\Bigl|\frac{\lambda_{k}}{\nn{x-y_{k}}}-\frac{\widehat{\lambda}_{k}}{\nn{x-\widehat{y}_{k}}}\Bigr|\d x \\
            & \leq \frac{\varepsilon_\omega}{\ell^{-1/3}}+\frac{\eta_\lambda}{\ell^{-1/3}}\sum_{k\neq k'}\int_{\R^3}|f_{\j}(x)|^2\frac{1}{\nn{x-y_{k}}}\d x+\frac{\Lambda^*+\eta_\lambda}{\ell^{-1/3}}\sum_{k\neq k'}\int_{\R^3}|f_{\j}(x)|^2\Bigl|\frac{1}{\nn{x-y_{k}}}-\frac{1}{\nn{x-\widehat{y}_{k}}}\Bigr|\d x \\
            & \overset{\text{(1)}}{\leq} \frac{\varepsilon_\omega}{\ell^{-1/3}}+\frac{\eta_\lambda}{\ell^{-1/3}}\frac{K-1}{y_*-12\sqrt[3]{\ell}}+\frac{\Lambda^*+\eta_\lambda}{\ell^{-1/3}}\sum_{k\neq k'}\int_{\R^3}|f_{\j}(x)|^2\frac{\nn{\widehat{y}_{k}-y_{k}}}{\nn{x-y_{k}}\nn{x-\widehat{y}_{k}}}\d x\\
            & \overset{\text{(2)}}{\leq} \frac{\varepsilon_\omega}{\ell^{-1/3}}+\frac{\eta_\lambda}{\ell^{-1/3}}\frac{2K}{y_*}+\frac{\Lambda^*+\eta_\lambda}{\ell^{-1/3}}\sum_{k\neq k'}\int_{\R^3}|f_{\j}(x)|^2\frac{4\eta_y}{y_*^2}\d x\\
            & \overset{\text{(3)}}{\leq} \frac{\varepsilon_\omega}{\ell^{-1/3}}+\frac{\eta_\lambda}{\ell^{-1/3}}\frac{2K}{y_*}+\eta_y\frac{8\Lambda^*K}{\ell^{-1/3}y_*^2}
		\end{aligned}
	\end{equation*}
    where we used
    \begin{itemize}
        \item[\hypertarget{assum-3}{(1)}] The bounds $\|\widehat{y}_{k'}-x\|\leq\| \widehat{y}_{k'}-p_{\j}\|+\|p_{\j}-x\|\leq \ell\sqrt{3}(s_1+3+1)\leq \ell\sqrt{3}(\frac{\sqrt[3]{\ell}}{\ell}+5)\leq 12\sqrt[3]{\ell}\leq \frac{1}{2}y_*$ for all $x\in B_{\j}$ follow from the assumption $\ell\leq(\frac{y_*}{26})^3$ (see \Cref{proof-eq:all_assumptions_ell}), and together with the inverse triangle inequality
        \begin{equation*}
            \begin{aligned}
                \|x-y_k\|\geq \|y_k-\widehat{y}_{k'}\|-\|\widehat{y}_{k'}-x\| \geq y_* - 12\sqrt[3]{\ell} \geq \frac{1}{2}y_*\,.
            \end{aligned}
        \end{equation*}
        Moreover, we use again the inverse triangle inequality for $\bigl|\|x-y_k\|-\|x-\widehat{y}_k\|\bigr|\leq\|\widehat{y}_k-y_k\|$.
        \item[(2)] Additionally to the bounds in (1), $\|y_k-\widehat{y}_{k}\|\leq\|p_{\j}-x\|+\|y_k-\widehat{y}_{k}\|\leq \sqrt[3]{\ell}$ shows
        \begin{equation*}
            \|x-\widehat{y}_k\| \geq \|x-y_k\|-\|y_k-\widehat{y}_{k}\|\geq y_*-13\sqrt[3]{\ell} \geq \frac{1}{2}y_*\,.
        \end{equation*}
        \item[\hypertarget{assum-4}{(3)}] The bound $\eta_\lambda\leq\Lambda^*$ follows from the assumption $\ell\leq(\frac{1}{4Ky_*^2})^{12}$ and impose $\varepsilon_\omega\leq\frac{\Lambda^*}{2y_*}$ so that
        \begin{equation*}
            \eta_\lambda = \sqrt[12]{\ell}2K\Lambda^*y_*^2+y_*\varepsilon_\omega\leq\Lambda^*\,.
        \end{equation*}
    \end{itemize}
	Next, we insert the expression of $\eta_\lambda$
	\begin{equation*}
		\begin{aligned}
			|\omega_{\j}^{(k')}-\widehat{\omega}_{\j}^{(k')}| & \leq\frac{\varepsilon_\omega}{\ell^{-1/3}}+\biggl(\frac{\eta_y2K\Lambda^*y_*^2}{\ell^{3/12}}+y_*\varepsilon_\omega\biggr)\frac{2K}{\ell^{-1/3}y_*}+\eta_y\frac{8\Lambda^*K}{\ell^{-1/3}y_*^2}\\
            & \leq \frac{\varepsilon_\omega\Bigl(1+\frac{y_*2K}{y_*}\Bigr)}{\ell^{-1/3}}+\frac{\eta_y}{\ell^{-1/3}}\biggl(\frac{8\Lambda^*K}{y_*^2}+\frac{4K^2\Lambda^*y_*}{\ell^{3/12}}\biggr)\\
            & \leq \varepsilon_\omega \sqrt[3]{\ell}\Bigl(1+2K\Bigr)+\eta_y\biggl(\sqrt[3]{\ell}\frac{8\Lambda^*K}{y_*^2}+\sqrt[12]{\ell}4K^2\Lambda^*y_*\biggr)\\
            & \leq \varepsilon_\omega \sqrt[3]{\ell}3K+\eta_y\biggl(\sqrt[3]{\ell}\frac{8\Lambda^*K}{y_*^2}+\sqrt[12]{\ell}4K^2\Lambda^*y_*\biggr)\,.
		\end{aligned}
	\end{equation*}
    Then, we impose the following assumptions on $\ell$ (see \Cref{proof-eq:all_assumptions_ell}) for a variable $\widetilde{c}>1$:
    \begin{equation}\label{appx-eq:assumption-6}
		\begin{aligned}
			\ell & \leq\min\Bigl\{\Bigl(\frac{1}{\widetilde{c}6K}\Bigr)^{12},\,\Bigl(\frac{y_*^2}{\widetilde{c}32\Lambda^* K}\Bigr)^3, \Bigl(\frac{1}{\widetilde{c} 16K^2\Lambda^*y_*}\Bigr)^{12}\Bigr\}
		\end{aligned}
	\end{equation}
    which allows us to show
    \begin{equation}\label{appx-eq:error-local-averages}
		|\omega_{\j}^{(k')}-\widehat{\omega}_{\j}^{(k')}| \leq \varepsilon_\omega \ell^{3/12} \sqrt[12]{\ell}3K+\eta_y\biggl(\sqrt[3]{\ell}\frac{8\Lambda^*K}{y_*^2}+\sqrt[12]{\ell}4K^2\Lambda^*y_*\biggr)\leq\frac{\varepsilon_\omega \ell^{3/12}+\eta_y}{2\widetilde{c}}\,.
	\end{equation}
    Next, we consider two cases. First, $\varepsilon_\omega\ell^{3/12}\leq \eta_y$, which directly reduces the above equation to
    \begin{equation*}
		|\omega_{\j}^{(k')}-\widehat{\omega}_{\j}^{(k')}| \leq\frac{\eta_y}{\widetilde{c}}\,.
	\end{equation*}
	This allows us to exactly follow the proof of \Cref{appx-thm:single-Coulomb} starting from \Cref{appx-eq:multi-mode-start} but for the isolated, relative local averages $\widehat{\omega}_\j^{(k')}$:
    \begin{equation*}
		\omega_{\j}^{(k')}=\frac{\widetilde{\lambda}}{\nn{p_{\j}-y}}\qquad\text{or}\qquad\frac{\widetilde{\lambda}^2}{\omega_{\j}^2}=z_\j
	\end{equation*}
    for $\widetilde{\lambda}=\frac{\lambda}{\ell^{-1/3}}$. By the definitions 
	\begin{equation}\label{appx-ex:matrix-p-vec2}
		p_{\mathbf{i}\j}=\frac{1}{2}(p_{\mathbf{i}}-p_{\j})\qquad v_{\mathbf{i}\j}=(\|p_\mathbf{i}\|^2-\|p_\j\|^2)\qquad\eta_{\mathbf{i}\j} =(\frac{1}{\omega_\mathbf{i}^2}-\frac{1}{\omega_\j^2})\,.
	\end{equation}
	the above equations translates to (see \Cref{appx-eq:single-Coulomb-redefinition})
	\begin{equation*}
		v_{\mathbf{i}\j}=\eta_{\mathbf{i}\j}\widetilde{\lambda}^2+\cs{p_{\mathbf{i}\j},y}\,.
	\end{equation*}
	Next, we want to find eight appropriate lattice indices $\mathbf{i}$ and $\j$ such that the above linear system of equations can be solved (see  \Cref{appx-eq:matrix-LGS}). Then, as in \Cref{appx-eq:multi-mode-start}, we define
    \begin{equation*}
		\j_0^{(k')} \coloneqq \argmax_{\j\in B_{\hat{\j}_{k'}}\partial_{s_1+1}^{3}}\|\hat{\j}_{k'}-\j\|
	\end{equation*}
    for $\hat{\j}_y$ such that $\widehat{y}_{k'}\in B_{\hat{\j}_{k'}}$ and calculate the overlap with the canonical basis vectors in $e_1$, $e_2$, and $e_3$ (as in \Cref{eq-appx:overlap}). Then, we define the indices for the points $p_1,...,p_8$ by 
	\begin{equation*}
		\begin{aligned}
			1 & \mapsto\hat{\j}_{k'}-3e_3; & \qquad 2 & \mapsto\hat{\j}_{k'}-4e_3; & \qquad 3 & \mapsto\hat{\j}_{k'}-2e_3; & \qquad 4 & \mapsto\hat{\j}_{k'}-3e_3;\\
			5 & \mapsto\j_0^{(k')}+e_2; & \qquad 6 & \mapsto\j_0^{(k')}; & \qquad 7 & \mapsto\j_0^{(k')}+e_1; & \qquad 8 & \mapsto\j_0^{(k')}
		\end{aligned}
	\end{equation*}
	Next, we define $p_{ij}$, $v_{i,j}$, and $\eta_{i,j}$ for $i,j\in\{1:8\}$, (see \Cref{appx-ex:matrix-p-vec2}). By the same caluculation as in \Cref{appx-eq:lower-bound-single-C}, we have
	\begin{equation*}
		\eta_{34}-\eta_{12}=\frac{1}{2}\ell^2\,.
	\end{equation*}
    Next, we define $v_{ij}=(\|p_i\|^2-\|p_j\|^2)$ as in \Cref{appx-eq:single-Coulomb-redefinition} and note that these definitions directly leads to well-conditioned equations (see \Cref{appx-eq:solve-LGS-single-mode})
	\begin{equation*}
		\begin{aligned}
			|\widetilde{\lambda}| & = \sqrt{\Bigl|\frac{v_{34}-v_{12}}{\eta_{34}-\eta_{12}}\Bigr|} \\
			y_1       & =v_{34}-\eta_{34}\widetilde{\lambda}^2                                     \\
			y_2       & =v_{56}-\eta_{56}\widetilde{\lambda}^2                                     \\
			y_3       & =v_{78}-\eta_{78}\widetilde{\lambda}^2\,.
		\end{aligned}
	\end{equation*}
    This scheme provides the improved bound below on the location and on the scaled charge number $\widetilde{\lambda}$. However, the improvement on the scaled large number vanished when rescaling the charge number $\lambda=\ell^{-1/3}\widetilde{\lambda}$. For the centers of the Coulomb potential given by $y_{k'}=(y_1, y_2, y_3)$ we achive the following error
	\begin{equation*}
		\nn{y_{k'}-\widehat{y}_{k'}}\leq\frac{\eta_y}{\widetilde{c}}\coloneqq \eta_y^{(1)}\,.
	\end{equation*}
	Next, we apply again the diagonally dominant matrix inversion starting with precision $\eta_y^{(1)}$ in \Cref{appx-eq:input-precision-lgs} and finishing in \Cref{appx-eq:lambda-diag-dom}. In more detail, we chose $K$ points $j'_{k}$ satisfying $p_{\j'_{k}}\in B_{\hat{\j}_k}\partial_{s_3}$ with $\widehat{y}_{k}B_{\hat{\j}_k}$ and $s_3=s_1+\lceil\ell^{-11/12}\rceil+1$ for all $k\in\{1:K\}$ (see \Cref{fig:multi-coulomb-regions}). Moreover, we slightly adapt the notation (compare to \Cref{appx-eq:diag-dom-matrix}) to 
	\begin{equation*}
		\widehat{\mathbb{M}}\coloneqq\Bigl(\int_{\R^3}|f_{\j'_{k'}}({x})|^2\frac{1}{\|{x-\widehat{y}_{k}}\|}dx\Bigr)_{k',k=1}^K\,.
	\end{equation*}
	Similar to \Cref{appx-eq:cond-diag-dom}, we need to find the following bound
	\begin{equation*}
		\max_{x\in B_{{\j}'_{k}}}\|\widehat{y}_{k}-x\|\leq \eta<\frac{y_*}{2(K-1)}\,.
	\end{equation*}
    Due to the assumptions on $m$ (see \Cref{proof-eq:all_assumptions_ell}):
	\begin{equation}\label{appx-eq:proof-assumption3.1}
		\max_{x\in B_{{\j}'_{k}}}\|\widehat{y}_{k}-x\|\leq \sqrt{3}\ell (s_3+1)= \sqrt{3}(\ell\lceil\ell^{-2/3}\rceil+\ell\lceil\ell^{-11/12}\rceil+2\ell)\leq\sqrt{3}(\sqrt[3]{\ell}+\sqrt[12]{\ell}+4\ell)\leq12\sqrt[12]{\ell}\leq\frac{5y_*}{16(K-1)}\,,
	\end{equation}
    which directly implies the above assumption with $\eta = \frac{5y_*}{16 (K-1)}$ and \Cref{lem:diagonal-dominant-coulomb}
	\begin{equation*}
		\|\widehat{\mathbb{M}}^{-1}\|\leq\frac{\eta y_*}{y_*-2\eta(K-1)}\leq\frac{y_*}{K-1}\,.
	\end{equation*} 
    This diagonal dominance of $\mathbb{M}$ allows us to solve the linear system $(\widehat{\omega}_{\j_k})_{k=1}^K = \widehat{\mathbb{M}}(\widehat{\lambda}_{k})_{k=1}^K$ and achieve a stable solution $(\widehat{\lambda}_{k})_{k=1}^K$ (see \cite{fawzilecture.2022numerical}). To provide a explicit error bound, we need to find an upper bound on \Cref{appx-eq:matric-inversion-error}. The only terms, which have changed are $\|\widehat{\mathbb{M}}^{-1}\|$, which have changed above, and $\|M-\widehat{M}\|$, which is discussed in the following (compare to \Cref{appx-eq:difference-diag-dom-mat}):
	\begin{equation*}
		\Bigl|\int_{\R^3}|f_{\j'_{k'}}(x)|^2\frac{1}{\|x-y_k\|}\d x-\int_{\R^3}|f_{\j'_{k'}}(x)|^2\frac{1}{\|x-\widehat{y}_{k}\|}\d x\Bigr|\leq\int_{\R^3}|f_{\j'_{k'}}(x)|^2\frac{\|\widehat{y}_{k}-y_k\|}{\sqrt[6]{\ell}} \d x\leq\frac{\eta_y^{(1)}}{\sqrt[6]{\ell}}\,.
	\end{equation*}
	Then, following the same procedure as in \Cref{appx-eq:matric-inversion-error} until \Cref{appx-eq:lambda-diag-dom}, we obtain
	\begin{equation*}
        \begin{aligned}
            |\lambda_{k'}-\widehat{\lambda}_{k'}|\leq 2K\Lambda^*y_*^2\frac{\eta_y^{(1)}}{\ell^{3/12}}+y_*\varepsilon_\omega= 2K\Lambda^*y_*^2\frac{\eta_y}{\widetilde{c}\, \ell^{3/12}}+y_*\varepsilon_\omega\,.
        \end{aligned}
	\end{equation*}
	Next, we iterate the above case and procedure $S\in\N$ times until the case assumption $\varepsilon_\omega\ell^{3/12}\leq \frac{\eta_y}{\tilde{c}^S}$ is violated:
	\begin{equation*}
		\nn{y_{k'}-\widehat{y}_{k'}}\leq\max\Bigl\{\frac{\eta_y}{\widetilde{c}^S}, \ell^{3/12}\varepsilon_\omega\Bigr\}\eqqcolon \eta_y^{(S)}
	\end{equation*}
	with $\eta_y^{(0)}=\eta_y$ and
	\begin{equation*}
        \begin{aligned}
            |\lambda_{k'}-\widehat{\lambda}_{k'}|&\leq  2K\Lambda^*y_*^2\frac{\eta_y^{(S)}}{\ell^{3/12}}+y_*\varepsilon_\omega
		\end{aligned}
	\end{equation*}

    Then the case $\eta_y^{(S)}\leq\ell^{3/12}\varepsilon_\omega$ is achieved, apply the assumption to \Cref{appx-eq:error-local-averages} showing that
    \begin{equation}\label{appx-eq:case1-error}
		|\omega_{\j}^{(k')}-\widehat{\omega}_{\j}^{(k')}| \leq\frac{\varepsilon_\omega \ell^{3/12}+\eta_y^{(S)}}{2\widetilde{c}}\leq\frac{\ell^{3/12}\varepsilon_\omega}{\widetilde{c}}\,.
	\end{equation}
    By the same procedure as before, we achieve the better estimate
    \begin{equation*}
		\nn{y_{k'}-\widehat{y}_{k'}}\leq\frac{\ell^{3/12}\varepsilon_\omega}{\widetilde{c}}\,.
	\end{equation*}
    and by inserting the assumption 
    \begin{equation*}
        \begin{aligned}
            |\lambda_{k'}-\widehat{\lambda}_{k'}|\leq \Bigl(2K\Lambda^*y_*^2+y_*\Bigr)\varepsilon_\omega\,.
        \end{aligned}
	\end{equation*}
	Since, $\varepsilon_\omega$ is not fixed yet, we choose $\varepsilon_\omega= \min\Bigl\{\frac{1}{4K\Lambda^*y_*^2},\,\frac{1}{2y_*},1\Bigr\}\varepsilon$ and $\widetilde{c}>1$, which shows
	\begin{equation*}
		\nn{y_{k'}-\widehat{y}_{k'}}\leq\max\Bigl\{\frac{\eta_y}{\widetilde{c}^S}, \frac{\ell^{3/12}\varepsilon_\omega}{\widetilde{c}}\Bigr\}\leq\varepsilon\qquad\text{and}\qquad|\lambda_{k'}-\widehat{\lambda}_{k'}|\leq \varepsilon\,.
	\end{equation*}
	for all $k'\in\{1,...,K\}$using the assumption $\ell<1$. Next, we consider the number of iterations $S$ until the case $\eta_y^{(S)}> \ell^{3/12}\varepsilon_\omega$ is achieved. For that, we repeat the assumptions on $\ell$, which is summarized in \Cref{proof-eq:all_assumptions_ell}:
	\begin{equation}\label{appx-eq:ell-proof-assumption-7}
		\begin{aligned}
			\ell\leq \min\Bigl\{\Bigl(\frac{1}{4K\Lambda^*y_*^2}\Bigr)^{12}\!\!,\Bigl(\frac{1}{2y_*}\Bigr)^{12}\!\!,1\Bigr\}\,.
		\end{aligned}
	\end{equation}
	Then, using the definition of $\eta_y^{(S)}=\frac{\eta_y}{\widetilde{c}^S}=\frac{\sqrt[3]{\ell}}{\widetilde{c}^S}$, we have
	\begin{equation*}
		\widetilde{c}^S>\frac{1}{\varepsilon}\geq\frac{\sqrt[12]{\ell}}{\varepsilon_\omega}=\frac{\sqrt[3]{\ell}}{\ell^{3/12}\varepsilon_\omega}\,.
	\end{equation*}
	Therefore, the case $\eta_y^{(S)}\leq \ell^{3/12}\varepsilon_\omega$ with respect to $\widetilde{c}=\exp(1)$ is achieved for 
	\begin{equation*}		
		S = \Bigl\lceil\ln\bigl(\frac{1}{\varepsilon}\bigr)\Bigr\rceil\,.
	\end{equation*}
	For presentation, we summarize the assumptions on the grid size $\ell=\frac{L}{m}$ in the following:
	\begin{equation}\label{proof-eq:all_assumptions_ell}
		\begin{aligned}
			\ell & \overset{\eqref{appx-eq:ell-assum-1}}{\leq} \min\Bigl\{\frac{1}{64},\Bigl(\frac{y_*}{8}\Bigr)^3\!\!,\Bigl(\frac{2\Lambda_*y_*}{3K\Lambda^*}\Bigr)^{3}\!\!, \Bigl(\frac{1}{288y_*}\Bigr)^3, \frac{1}{2\sqrt{3}}\frac{\Lambda_*y_*}{3K\Lambda^*}\Bigr\}\\
			\ell & \overset{\eqref{appx-eq:proof-assumption2.1}}{\leq}\Bigl(\frac{1}{64K}\Bigr)^{12}\leq\Bigl(\frac{1}{32(K-1)}\Bigr)^{12}\\
            \ell & \overset{(\ref{appx-eq:proof-assumption2.2}; \ref{appx-eq:proof-assumption3.1})}{\leq}\Bigl(\frac{y_*}{90K}\Bigr)^{12}\leq\Bigl(\frac{y_*}{45(K-1)}\Bigr)^{12}\leq\Bigl(\frac{5y_*}{224(K-1)}\Bigr)^{12}\\
            \ell & \overset{\hyperlink{assum-3}{\#(1)}}{\leq}\Bigl(\frac{y_*}{26}\Bigr)^3\\
            \ell & \overset{\hyperlink{assum-4}{\#(3)}}{\leq} \Bigl(\frac{1}{4Ky_*^2}\Bigr)^{12}\\
			\ell & \overset{\eqref{appx-eq:assumption-6}}{\leq}\min\Bigl\{\Bigl(\frac{1}{e^16K}\Bigr)^{12},\,\Bigl(\frac{y_*^2}{e^132\Lambda^* K}\Bigr)^3, \Bigl(\frac{1}{e^1 16K^2\Lambda^*y_*}\Bigr)^{12}\Bigr\}\\
			\ell & \overset{\eqref{appx-eq:ell-proof-assumption-7}}{\leq} \min\Bigl\{\Bigl(\frac{1}{4K\Lambda^*y_*^2}\Bigr)^{12}\!\!,\Bigl(\frac{1}{2y_*}\Bigr)^{12}\!\!,1\Bigr\} \leq \min\Bigl\{\Bigl(\frac{1}{4K\Lambda^*y_*^2}\Bigr)^{12},\,\Bigl(\frac{1}{2y_*}\Bigr)^{12},1\Bigr\}
		\end{aligned}
	\end{equation}
	Together,
	\begin{equation}\label{appx-eq:upper-bound-ell}
		\begin{aligned}
			\ell & \leq \min\biggl\{\!\Bigl(\!\frac{\Lambda_*y_*}{2K\Lambda^*}\!\Bigr)^{3}\!\!,\frac{\Lambda_*y_*}{12K\Lambda^*}\!,\Bigl(\!\frac{1}{64K}\!\Bigr)^{12}\!\!,\Bigl(\!\frac{y_*}{90K}\!\Bigr)^{12}\!\!,\Bigl(\frac{y_*}{26}\Bigr)^3\!\!,\Bigl(\!\frac{1}{144K(y_*\!+\!1)^2}\!\Bigr)^{12}\!\!, \Bigl(\!\frac{y_*^2}{4K}\!\Bigr)^{12}\!\!, \Bigl(\!\frac{y_*^2}{87\Lambda^* K}\!\Bigr)^3\!\!, \Bigl(\!\frac{1}{44K^2(\Lambda^*\!+\!1)y_*}\!\Bigr)^{12}\!\biggr\},
		\end{aligned}
	\end{equation}
    which directly relates to the bound on $m$ given in the statement. Additionally, note that 
    \begin{equation}\label{appx-eq:varepsilon-multi-coulomb}
        \varepsilon_\omega= \min\Bigl\{\frac{1}{4K\Lambda^*y_*^2},\,\frac{1}{2y_*},1\Bigr\}\varepsilon
    \end{equation}
    induces the dependencies on the total evolution time, which finishes the proof.
\end{proof}
Next, we prove the following auxiliary result used in the above proof.:
\begin{proposition}\label{prop:multi-Coulomb-stability-max}
    Let $L$, $\Lambda_*$, $\Lambda^*>0$, unknown $K\in\N$ with limit $K\leq\widetilde{K}\in\N$ and $V$ a multi-Coulomb potential
	\begin{equation*}
		V:\R^3\to\R, \quad V(x)=\sum_{k=1}^{{K}}\frac{\lambda_k}{\|x-y_k\|}
	\end{equation*}
	for $\lambda_k\in[\Lambda_*,\Lambda^*]$, $y_k\in[0,L]^3$ satisfying $\|y_k-y_{k'}\|\geq y_*$ for all $k\neq k'\in\{1:{K}\}$. For $c>0$ and fixed grid size $\ell$ with
    \begin{equation*}\label{eq:grid-step-size}
        \ell\leq \min\Bigl\{\frac{1}{64}, \Bigl(\frac{y_*}{8}\Bigr)^3\!\!, \Lambda_*^3\Bigl(\max\Bigl\{\frac{2\widetilde{K}\Lambda^*}{y_*}-\frac{c}{2},0\Bigr\}\Bigr)^{-3}\!\!, \Bigl(\frac{c}{288K\Lambda^*}\Bigr)^3, \frac{\Lambda_*}{2\sqrt{3}} \Bigl(\frac{2\widetilde{K}\Lambda^*}{y_*}+c\Bigr)^{-1}\Bigr\}\,, 
	\end{equation*}
    we can prove the following equivalence: If there are $K$ indices $\j_k\in\{0:m-1\}$ for $k\in \{1:K\}^3$ such that the points $p_{\j_k}$ are separated by distances of at least $y_* - 2\ell s_1>0$ with $s_1 = \lceil \ell^{-2/3} \rceil$, and if the following holds:
	\begin{equation}\label{appx-eq:maximization-multi}
		\omega_{\j_k} - \omega_{\j'} \geq c \qquad \text{for all} \qquad \j' \in \partial_{s_1} B_{\j_k}\,,
	\end{equation}
	then the Coulomb potential consists of $K$ Coulomb centers $\{y_k\}_{k=1}^K$ satisfying
	\begin{equation*}
		\|y_k - p_{\j_k}\| \leq \sqrt[3]{\ell}\,.
	\end{equation*}
	Moreover, assuming that the Coulomb potential consists of $K$ Coulomb centers, there exists a set of exactly $K$ maximizers satisfying the condition in \Cref{appx-eq:maximization-multi}.
\end{proposition}

\begin{proof}
    Before starting with the proof, we note that $y_*-2\ell s_1>0$ because $2\ell s_1\leq 2(\sqrt[3]{\ell}+\ell)\leq 4\sqrt[3]{\ell}\leq\frac{y_*}{2}$ holds by the assumption $\ell\leq(\frac{y_*}{8})^3$. Then, we first prove the direction that if there are $K$ maxima $\{\omega_{\j_{k}}\}_{k=1}^K$ fulfilling \Cref{appx-eq:maximization-multi} and distanced by at least $y_*-2\ell s_1$, i.e.~$\min_{k,k'\in\{1:K\}}\|p_{\j_k}-p_{\j_{k'}}\|\geq y_*-2s_1\ell $, then $\|y_k - p_{\j_k}\| \leq \sqrt[3]{\ell}$ for all $k\in\{1:K\}$ and $s_1=\lceil\ell^{-2/3}\rceil$. We prove the statement by contradiction. Therefore, we assume that there is a set of K points $\j_k$ satisfying \Cref{appx-eq:maximization-multi} and that there exists a $k'\in\{1:K\}$ such that $\|{p_{\j_k}-y_{k'}}\|>\sqrt[3]{\ell}$ holds true for all $k\in\{1:K\}$. Then, we aim to show a contradiction. W.l.o.g.~we fix the index $\j_{k'}$ of the point $p_{\j_{k'}}$ closest to $y_{k'}$ and analyze the differences of the local averages related to $p_{\j_{k'}}$ and a point on the boundary $p_{\j'_{k'}}\in B_{\j_{k'}}\partial_{s_1}$ (see \Cref{fig:cases-uniqueness-max}):
    \begin{equation*}
		\begin{aligned}
			\omega_{\j_{k'}}-\omega_{\j'_{k'}} & =\int_{\R^3} \bigl(|f_{\j_{k'}}(x)|^2-|f_{\j'_{k'}}(x)|^2\bigr) V(x) \d x\\
            & =\int_{\R^3} \bigl(|\ell^{-1/2}\,g\!\left(\bigl\|\ell^{-1}x-\tfrac{1}{2}(1,\ldots,1)^T-\j_{k'}\bigr\|\right)|^2-|\ell^{-1/2}\,g\!\left(\bigl\|\ell^{-1}x-\tfrac{1}{2}(1,\ldots,1)^T-\j'_{k'}\bigr\|\right)|^2\bigr) V(x) \d x\\
            & =\int_{\R^3} \bigl(|\ell^{-1/2}\,g\!\left(\bigl\|\ell^{-1}(x-p_{\j_{k'}})\bigr\|\right)|^2-|\ell^{-1/2}\,g\!\left(\bigl\|\ell^{-1}(x-p_{\j'_{k'}})\bigr\|\right)|^2\bigr) V(x) \d x\\
            & =\int_{\R^3} \ell^{-1}|g(\|\ell^{-1}x\|)|^2 \bigl(V(x+p_{\j_{k'}})-V(x+p_{\j'_{k'}})\bigr) \d x
		\end{aligned}
	\end{equation*}
	\begin{figure}[ht!]
		\includegraphics[scale=0.8]{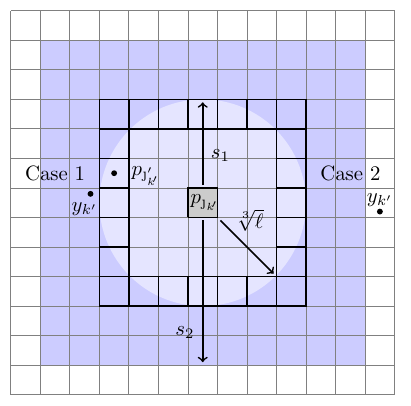}
		\caption{Schematic illustration of the regions explicitly defined in the proof of \Cref{prop:multi-Coulomb-stability-max}.}
		\label{fig:cases-uniqueness-max}
	\end{figure}
	\hypertarget{case-1}{\hspace{-0.8ex}First, we consider the case $y_{k'}\in\partial_{s_2}B_{\j_{k'}}\backslash\{x\,|\,\|x-p_{\j_{k'}}\|\leq \sqrt[3]{\ell}\}$ for $s_2=\lceil (1+\eta)\ell^{-2/3}\rceil$ and $\eta\in(0,1)$ (see \Cref{fig:cases-uniqueness-max}).} Next, we choose $p_{\j'_{k'}}=\argmin_{\j\in B_{\hat{\j}_{k'}}\partial_{s_1}}\|p_{\j}-y_{k'}\|$ where we used the short notation $\j\in B_{\hat{\j}_{k'}}\partial_{s_1}$ for $\j\in\{0:m-1\}^3$ such that $p_{\j}\in B_{\hat{\j}_{k'}}\partial_{s_1}$. Due to the geometry of the set $\partial_{s_1}B_{\j_{k'}}$,
	\begin{equation}\label{appy-eq:max-dist-annulus}
		\|p_{\j'_{k'}}-y_{k'}\|\leq\sqrt{3}\ell (s_2-s_1+\frac{1}{2})\leq\sqrt{3}(\eta\sqrt[3]{\ell} + \frac{3}{2}\ell)\,,
	\end{equation}
    where the value $1/2$ comes from the distance of the point $p_{\j'_{k'}}$ to the boundary of $\partial_{s_1}B_{\j_{k'}}$. Then, we write out $V$ and separate the case $k=k'$ and $k\neq k'$ for which we drop the negative terms to obtain the following bound:
	\begin{equation}\label{appx-eq:not-increasing1}
		\begin{aligned}
			\omega_{\j_{k'}}-\omega_{\j'_{k'}} & =\int_{\R^3} \ell^{-1}|g(\|\ell^{-1}x\|)|^2 \bigl(V(x+p_{\j_{k'}})-V(x+p_{\j'_{k'}})\bigr) \d x\\
            &\leq \sum_{k\neq k'}^K \int_{\R^3} \ell^{-1}|g(\|\ell^{-1}x\|)|^2\frac{\lambda_k}{\nn{x+p_{\j_{k'}} - y_{k}}} \d x\\
            & \qquad\qquad+ \lambda_{k'} \int_{\R^3} \ell^{-1}|g(\|\ell^{-1}x\|)|^2 \left(\frac{1}{\nn{x+p_{\j_{k'}} - y_{k'}}} - \frac{1}{\|x+p_{\j'_{k'}} - y_{k'}\|}\right) \d x \\
            & \leq \frac{\widetilde{K}\Lambda^*}{y_*-\sqrt{3}(\eta\sqrt[3]{\ell}+2\ell)} + \lambda_{k'}\frac{2}{\sqrt[3]{\ell}} - \lambda_{k'}\frac{1}{\sqrt{3}(\eta\sqrt[3]{\ell} + 2\ell)}\,.
		\end{aligned}
	\end{equation}
    In the last line, we use the following bound for the first term: for all $x\in B_{0}$ and $k\neq k'$, we have
	\begin{equation*}
		\begin{aligned}
			\nn{x+p_{\j_{k'}} - y_{k}}&\geq\nn{y_{k}-y_{k'}}-\nn{p_{\j_{k'}} - y_{k'}}-\nn{x}\\
            &\geq y_*-\sqrt{3}(\eta\sqrt[3]{\ell} + \frac{3}{2}\ell)-\sqrt{3}\frac{\ell}{2}\\
            &\geq y_*-\sqrt{3}(\eta\sqrt[3]{\ell}+2\ell)
		\end{aligned}
	\end{equation*}
	where we used \Cref{appy-eq:max-dist-annulus}. For the second term in \Cref{appx-eq:not-increasing1}, we apply the inverse triangle inequality and the assumptions $y_{k'}\in\partial_{s_2}B_{\j_{k'}}\backslash\{x\,|\,\|x-p_{\j_{k'}}\|\leq \sqrt[3]{\ell}\}$ as well as  $\sqrt[3]{\ell}\geq 3^{2/3}\ell \geq \sqrt{3}\ell$, which follows from $\ell\leq\frac{1}{3}$, to find
    \begin{equation}\label{appx-eq:prop-assum-ell-1.1}
        \nn{x+p_{\j_{k'}} - y_{k'}}\geq \nn{p_{\j_{k'}} - y_{k'}}-\nn{x}\geq\sqrt[3]{\ell}-\frac{\sqrt{3}}{2}\ell\geq \frac{1}{2}\sqrt[3]{\ell}
    \end{equation}
    for all $x\in B_{0}$. Also for term three of \Cref{appx-eq:not-increasing1}, we apply the triangle inequality to show
    \begin{equation*}
        \|x+p_{\j'_{k'}} - y_{k'}\|\leq\|p_{\j'_{k'}} - y_{k'}\|+\|x\|\leq\sqrt{3}(\eta\sqrt[3]{\ell} + 2\ell)\,.
    \end{equation*}
    Next, we fix $\eta=\frac{1}{3\sqrt{3}}-\frac{1}{6}>0$, which induces
    \begin{equation}\label{appx-eq:prop-assum-ell-1.2}
        \sqrt{3}(\eta\sqrt[3]{\ell} + 2\ell)=\frac{1}{3}\sqrt[3]{\ell}-\frac{\sqrt{3}}{6}\bigl(\sqrt[3]{\ell} - 12\ell\bigr)\leq\frac{1}{3}\sqrt[3]{\ell}
    \end{equation}
    because $\sqrt[3]{\ell} - 12\ell\geq0$ is induced by the assumption $\ell\leq\frac{1}{4^{3}}\leq\frac{1}{12^{3/2}}$. This simplifies \Cref{appx-eq:not-increasing1} as follows:
    \begin{equation}\label{appx-eq:condition-ellnew}
		\begin{aligned}
			\omega_{\j_{k'}}-\omega_{\j'_{k'}} \leq \frac{\widetilde{K}\Lambda^*}{y_*-\frac{1}{3}\sqrt[3]{\ell}} - \lambda_{k'}\frac{1}{\sqrt[3]{\ell}}\leq \frac{2\widetilde{K}\Lambda^*}{y_*} - \Lambda_*\frac{1}{\sqrt[3]{\ell}}\,, 
		\end{aligned}
	\end{equation}
	where we used the assumption $\ell\leq y_*^3\leq(\frac{3y_*}{2})^3$. For $\frac{c}{2}\geq \frac{2\widetilde{K}\Lambda^*}{y_*}$, the above bound directly reduces to  
	\begin{equation*}
		\begin{aligned}
			\omega_{\j_{k'}}-\omega_{\j'_{k'}} \leq \frac{2\widetilde{K}\Lambda^*}{y_*} - \Lambda_*\frac{1}{\sqrt[3]{\ell}}\leq\frac{c}{2}\,, 
		\end{aligned}
	\end{equation*}
	which contradicts the assumption. Otherwise, if $\frac{c}{2}< \frac{2\widetilde{K}\Lambda^*}{y_*}$, the assumption
	\begin{equation}\label{appx-eq:condition-ell1}
		\ell\leq\Lambda_*^3\Bigl(\frac{2\widetilde{K}\Lambda^*}{y_*}-\frac{c}{2}\Bigr)^{-3}\,,
	\end{equation}
	implies the following bound:
	\begin{equation*}
		\begin{aligned}
			\omega_{\j_{k'}}-\omega_{\j'_{k'}} \leq \frac{2\widetilde{K}\Lambda^*}{y_*} - \Lambda_*\frac{1}{\sqrt[3]{\ell}}\leq \frac{c}{2} 
		\end{aligned}
	\end{equation*}
	for any given $c>0$. This contracts the assumption that $|\omega_{\j_{k'}}-\omega_{\j'_{k'}}|\geq c$ for all $\j' \in \partial_{s_1} B_{\j_{k'}}$. Therefore, we have shown that the first case cannot be achieved.

	In the second case (compare to \hyperlink{case-1}{first case}), we assume that there is a $k'\in\{1:K\}$ such that $y_{k'}\in(\partial_{s_2}B_{\j_{k'}})^c$ for $s_2=\lceil (1+\eta)\ell^{-2/3}\rceil$ and $\eta=\frac{1}{3\sqrt{3}}-\frac{1}{6}>0$ (see \Cref{fig:cases-uniqueness-max}). For that, we repeat that a twice differentiable function $h:\R^2\supset U\rightarrow\R$ is called harmonic if $\Delta h = 0$ for an open subset $U$. Recall that the Coulomb potential is a harmonic function because of
    \begin{equation*}
        \frac{\partial^2}{\partial x_1^2}\frac{1}{\|x-p\|}=-\frac{\partial}{\partial x_1}\frac{x_1-p_1}{\|x-p\|^3}=\frac{1}{\|x-p\|^3}-3\frac{(x_1-p_1)^2}{\|x-p\|^5}
    \end{equation*}
    for $p\in\R^3\backslash U$. By adding all three derivatives, the above equation sums up to $0$. Next, we define
	\begin{equation*}
		U=\{x\in\partial_{s_1}B_{\j_{k'}}\,|\,\|x-z\|_\infty>\frac{\ell}{2} \text{ for all } z\in(\partial_{s_2}B_{\j_{k'}})^{c}\}\,,
	\end{equation*}
	i.e.~interior of the convex hull of the box-centers in $\partial_{s_1}B_{\j_{k'}}$. By linearity of the property of being harmonic, the function $V$ is harmonic on $U$. Next, we consider a local average as a function in $p$, i.e. 
    \begin{equation*}
        \omega_{p}: U\ni p\mapsto\int_{U}\ell^{-1}|g(\|\ell^{-1}(x-p)\|)|^2V(x)dx = V(p)\,,   
    \end{equation*}
    where the last equality is a direct application of Newton's shell theorem \cite[Thm.~9.7]{LiebLoss.2001}. Then, the weak maximum principle \cite[Thm.~1 on p.~327]{Evans.2010} directly states that the maximum is attained at the boundary of $U$. For that reason, there is a $p_*\in \{p\,|\, \|p-p_{\j_{k'}}\|_\infty=s_1\ell\}$ such that $\omega_{p_*}-\omega_{\j_{k'}}\geq 0$. Additionally, there is a $p_{\j'_{k'}}\in B_{\j_{k'}}\partial_{s_1}$ defined by our discrete lattice $\j'_{k'}\in\{1,...,m\}^3$ so that $\|p_{\j'_{k'}}-p_*\|\leq\frac{\ell}{2}$. Since. $y_k\notin \partial_{s_2}B_{\j_{k'}}$ for all $k\in\{1:K\}$, Newton's shell theorem \cite[Thm.~9.7]{LiebLoss.2001} can be applied and shows 
    \begin{equation}\label{appx-eq:proof-contradiction-far-y}
		\begin{aligned}
			\omega_{\j_{k'}}-\omega_{\j'_{k'}} & =\omega_{\j_{k'}}-\omega_{p_*}+\omega_{p_*}-\omega_{\j'_{k'}}\\
            &\overset{\text{(1)}}{\leq} V(p_{*})-V(p_{\j'_{k'}})\\
            &\overset{\text{(2)}}{\leq} \widetilde{K}\Lambda^*\frac{\|p_{\j'}-p_{*}\|}{\|y_{k'}-p_{*}\|\,\|y_{k'}-p_{\j'_{k'}}\|}\\
            &\overset{\text{(3)}}{\leq} \widetilde{K}\Lambda^*\frac{\ell}{\eta\ell^{1/3}(1+\eta)\ell^{1/3}}\\
            &\leq \frac{\widetilde{K}\Lambda^*}{\eta^2}\sqrt[3]{\ell}
		\end{aligned}
	\end{equation}
    where we used in (1) $\omega_{\j_{k'}}-\omega_{p_*}\leq 0$, in (2) that $y_{k'}$ is closest to $p_*$ under the Coulomb centers, and in (3) the definition of $p_{\j'_{k'}}$, $p_*$, and the assumption of the case $y_{k'}\in(\partial_{s_2}B_{\j_{k'}})^c$. For $\eta=\frac{1}{3\sqrt{3}}-\frac{1}{6}\leq \frac{1}{6}$, direct calculations show 
	\begin{equation*}
		\frac{1}{12}\leq\eta=\frac{1}{3\sqrt{3}}-\frac{1}{6}\leq \frac{1}{6}\,.
	\end{equation*}
	so that the assumption
	\begin{equation}\label{appx-eq:condition-ell2}
        \ell\leq\Bigl(\frac{c}{288\widetilde{K}\Lambda^*}\Bigr)^3
	\end{equation}
    implies
    \begin{equation*}
		|\omega_{\j_{k'}}-\omega_{\j'_{k'}}|\leq \frac{\widetilde{K}\Lambda^*}{\eta^2}\sqrt[3]{\ell}\leq 144 \widetilde{K}\Lambda^* \sqrt[3]{\ell}\leq\frac{c}{2}\,,
	\end{equation*}
	which contradicts our initial assumption and completes the first implication of the proof. The assumptions on $\ell$ are summarized in the following:
	\begin{equation}\label{appx-eq:assumptions-on-ell}
		\begin{aligned}
			\ell & \overset{\text{(\ref{appx-eq:prop-assum-ell-1.1}-\ref{appx-eq:prop-assum-ell-1.2})}}{\leq}\frac{1}{64}\,,\\
			\ell & \overset{\text{(\ref{appx-eq:condition-ellnew})}}{\leq}\Bigl(\frac{y_*}{8}\Bigr)^3\leq y_*^3\leq\Bigl(\frac{3y_*}{2}\Bigr)^3\,,\\
			\ell & \overset{\eqref{appx-eq:condition-ell1}}{\leq}\Lambda_*^3\Bigl(\max\Bigl\{\frac{2K\Lambda^*}{y_*}-\frac{c}{2},0\Bigr\}\Bigr)^{-3}\,,\\
			\ell & \overset{\eqref{appx-eq:condition-ell2}}{\leq}\Bigl(\frac{c}{288\widetilde{K}\Lambda^*}\Bigr)^3\,.
		\end{aligned}
	\end{equation}

	In the next step, we prove the opposite direction of the statement, i.e.~given a $K$-Coulomb potential, we show that there are exactly $K$ optimizers in terms of \Cref{appx-eq:maximization-multi}. First, we take the boxes $B_{\j_k}$ that satisfy $y_k\in B_{\j_k}$ for all $k\in\{1:K\}$. Then, we fix a $k'\in\{1:K\}$, choose $\j'_{k'}\in B_{\j_k}\partial_{s_1}$, and follow similar ideas as \Cref{appx-eq:not-increasing1}:
    \begin{equation*}
		\begin{aligned}
			\omega_{\j_{k'}}-\omega_{\j'_{k'}} & =\int_{\R^3} \ell^{-1}|g(\|\ell^{-1}x\|)|^2 \bigl(V(x+p_{\j_{k'}})-V(x+p_{\j'_{k'}})\bigr) \d x\\
            & \geq-\sum_{k\neq k'}^K\frac{\lambda_k}{\|p_{\j'_{k'}} - y_k\|} + \int_{\R^3} \ell^{-1}|g(\|\ell^{-1}x\|)|^2\frac{\lambda_{k'}}{\|x-p_{\j_{k'}}-y_{k'}\|} \d x - \frac{\lambda_{k'}}{\|p_{\j'_{k'}} - y_{k'}\|}\\
		\end{aligned}
	\end{equation*}
	In the last inequality, we just drop positive terms coming from the Coulomb centers $y_k$ for $k\neq k'$ and applied Netwons shell theorem \cite[Thm.~9.7]{LiebLoss.2001} to the terms not supported in the box $B_{\j_{k'}}$. Next, we use the following bounds:
    \begin{itemize}
        \item[(1)] By definition of $s_1$ and due to the assumption $y_*\geq 8\sqrt[3]{\ell}$: 
        \begin{equation*}
            \|p_{\j'_{k'}} - y_k\|\geq\|y_{k'} - y_k\|-\|p_{\j'_{k'}} - y_{k'}\|\geq y_*-\sqrt{3}(s_1+1)\ell\geq y_*-4\sqrt[3]{\ell}\geq\frac{y_*}{2}
        \end{equation*}
        \item[(2)] Since $y_{k'},\,x\in B_{0}$
        \begin{equation*}
            \|x-p_{\j_{k'}}-y_{k'}\|\leq\sqrt{3}\ell\,.
        \end{equation*}
        \item[(3)] By definition of $p_{\j'_{k'}}\in B_{\j_{k'}}\partial_{s_1}$ and $\sqrt[3]{\ell}\geq\ell$
        \begin{equation*}
            \|p_{\j'_{k'}} - y_{k'}\|\geq \sqrt[3]{\ell}-\frac{\ell}{2}\geq\frac{1}{2}\sqrt[3]{\ell}\,.
        \end{equation*}
    \end{itemize}
    Next, we apply the above bounds which demonstrate
    \begin{equation*}
		\begin{aligned}
			\omega_{\j_{k'}}-\omega_{\j'_{k'}} & \geq -\frac{2(K-1)\Lambda^*}{y_*}+\frac{\lambda_{k'}}{\sqrt{3}\ell}-\frac{2\lambda_{k'}}{\sqrt[3]{\ell}}\\
            & \overset{(1)}{\geq} -\frac{2K\Lambda^*}{y_*}+\frac{\lambda_{k'}}{2\sqrt{3}\ell}\\
            &\overset{(2)}{\geq} c
		\end{aligned}
	\end{equation*}
    In (1), we use that the assumption $\ell\leq\frac{1}{64}\leq\frac{1}{(4\sqrt{3})^{3/2}}$ implies $4\sqrt{3}\ell\leq\sqrt[3]{\ell}$ which is equivalent to $\frac{1}{2\sqrt{3}\ell}-\frac{2}{\sqrt[3]{\ell}}\geq0$ and in (2) we used the following assumption on $\ell$:
	\begin{equation}\label{appx-eq:final-assumption}
		\ell\leq \frac{\Lambda_*}{2\sqrt{3}} \Bigl(\frac{2K\Lambda^*}{y_*}+c\Bigr)^{-1}\,.
	\end{equation}
	This shows that there are at least $K$ points $p_{\j_k}$ that satisfy \Cref{appx-eq:maximization-multi}. Let $p_{\j_{K+1}}$ be another optimizer, which by assumption is $y_*-2\lceil\ell^{-1/3}\rceil$ away from any Coulomb center $y_k$. Then, \Cref{appx-eq:proof-contradiction-far-y} directly contradicts the construction of $p_{\j'_{k'}}$ (see \Cref{appx-eq:maximization-multi}), which proves that there are exactly $K$ optimizers and finishes the proof.
\end{proof}

\begin{remark}\label{appx-rmk:extension-multi-mode}
    Note that the above method can be extended to a rotationally invariant potential not satisfying Newton's shell theorem and only requires a certain decay rate of the individual potentials away from their centers, as well as stability of the maxima.
\end{remark}

One key ingredient of the above proofs is the diagonal dominance of the matrix $\mathbb{M}$, for which we establish the following criterion:

\begin{lemma}\label{lem:diagonal-dominant-coulomb}
	Let $\{y_k\}_{k=1}^K\subset\R^3$ be $K\geq2$ Coulomb centers separated by at least $y_*$, i.e.~$\|y_k-y_{k'}\|\geq y_*$ for all $k\neq k'$. Moreover, let $\{p_{\j'_{k'}}\}_{k'=1}^K\subset\R^3$ are $K$ points such that $\max_{x\in B_{\j'_{k'}}}\|y_k-x\|\leq \eta<\frac{y_*}{2(K-1)}$, then the matrix
	\begin{equation*}
		\mathbb{M}=\Bigl(\int_{\R^3} |f_{\j'_{k'}}(x)|^2 \frac{1}{\nn{x - y_{k}}} \d x\Bigr)_{k',k=1}^{K}
	\end{equation*}
	is strictly diagonally dominant with
	\begin{equation*}
		\|\mathbb{M}^{-1}\|\leq\frac{\eta y_*}{y_*-2\eta(K-1)}\,.
	\end{equation*}
\end{lemma}
\begin{proof}
    To prove the statement, we fix $k'=1$ w.l.o.g.~and apply the assumption $\max_{x\in B_{\j_{1}}}\|x-y_{1}\|\leq\eta<y_*$:
	\begin{equation*}
		\begin{aligned}
			\sum_{k=2}^{K} \int_{\R^3}| f_{\j'_1}(x)|^2 \frac{1}{\nn{x-y_k}} \d x <\frac{(K-1)}{y_*-\eta}\,.
		\end{aligned}
	\end{equation*}
	On the other side, we have
	\begin{equation*}
		\begin{aligned}
			\int_{\R^3} |f_{\j'_1}(x)|^2 \frac{1}{\nn{x-y_1}}\d x\geq\frac{1}{\eta}
		\end{aligned}
	\end{equation*}
	so that
	\begin{equation}\label{eq:proof-diag-dominant}
		\mathbb{M}_{11}-\sum_{k\neq 1}\mathbb{M}_{1,k}\geq \frac{1}{\eta}-\frac{(K-1)}{y_*-\eta}\,.
	\end{equation}
	Next, we lower bound the right-hand side by using the assumption $y_*\geq 2\eta$:
	\begin{equation*}
		\mathbb{M}_{11}-\sum_{k\neq 1}\mathbb{M}_{1,k}\geq \frac{1}{\eta}-\frac{(K-1)}{y_*-\eta}\geq \frac{1}{\eta}-\frac{2(K-1)}{y_*}=\frac{y_*-2\eta(K-1)}{\eta y_*}\,.
	\end{equation*}
    By the same reasoning, we can also sum over $k'$ achieving the same bound, which allows us to apply the Varah's bound (see \cite{Varah.1975} and \Cref{eq:diag-dom-inverse-norm}).
\end{proof}

\subsection{Statistics}
\begin{lemma}[Hoeffding's inequality \texorpdfstring{\cite{hoeffding1963probability}}{???}]\label{lem:hoeffy}
	Let $X_1,\dots, X_N$ be independent random variables such that $a\le X_j\le b$, $a\le b$, almost surely. Consider the sum $S_N\coloneqq\sum_{j=1}^N X_j$. Then
	\begin{align*}
		\mathbb{P}\Big(|S_N-\mathbb{E}[S_N]|\ge\epsilon\Big)\le 2\operatorname{exp}\left(-\frac{2\epsilon^2}{N(b-a)^2}\right)\,.
	\end{align*}
\end{lemma}

\section{LRBs and \texorpdfstring{$2^{\text{nd}}$}{2nd} Moment}\label{appx-subsec:finite-diff}
In this section we show a rigorous version of \Cref{eq:direct V}, i.e., how the local averages $\omega_\j$ can be expressed as derivatives at $t=0$ of the expectation values of our measurements up to an error, which we explicitly estimate. More precisely, we will prove the following theorem.
\begin{theorem}\label{th:final_fVf_D_estimate}
	Assume that $V$ is $4(d+1)$ times differentiable with bounded derivatives (for $\abs{\Jf} > 1$) or $V$ is relatively bounded by $-\Delta$ (for $\abs{\Jf} = 1$).
	We have
	\begin{align*}
		\cs{f^\alpha_\j, V f^\alpha_\j} & =  \sum_{1\leq \beta < \gamma \leq 3} 	\sigma_\alpha(\beta, \gamma)  \dDd{t}{\j}{\beta \gamma} -  2 \ell^{-2} \cs{f,(-\Delta) f}  + \abs{t} E_{t,\ell},
	\end{align*}
	where $f$ is the fixed profile function used in \cref{eq:f_j defn}, and
	\begin{align}
		\dDd{t}{\j}{\alpha \beta}
		                             & \coloneqq \frac{ \Dd_\j^{\alpha \beta} (t) -  \Dd_\j^{\alpha \beta} (0)  }{t}, \\
		\sigma_\alpha(\beta, \gamma) & :=
		\begin{cases}
			1  & \text{if } \alpha \in \{\beta,\gamma\}, \\
			-1 & \text{else},
		\end{cases} \label{eq:sigma constants}
	\end{align}
	and the error can be estimated by
	\begin{align*}
		\abs{E_{t,\ell} } \leq   C \begin{cases}
			                           \cs t^{3+d} \ell^{-\gd} & :  \abs{\Jf} > 1, \\ \ell^{-4} &: \abs{\Jf} = 1,
		                           \end{cases}
	\end{align*}
	with $\cs t := \sqrt{t^2 + 1}$.
\end{theorem}
We start by expressing local averages $\omega_\j$ as the first derivative at $t=0$ of the expectation values
\begin{align*}
	\Dd_\j^{\alpha\beta}(t) = \cs{ \psi^{\alpha\beta}_t, (W^{\alpha\beta}_\j)^* P^{\alpha\beta}_\j W^{\alpha\beta}_\j \psi^{\alpha\beta}_t}
\end{align*}
at $t=0$, up to constants depending on the fixed profile $f$. In the following we will use the notation
\begin{align*}
	\Aj{\j}{\alpha\beta} = a^*(f^\alpha_{\j})  a^*(f^\beta_{\j}).
\end{align*}
\begin{proposition}\label{th:first_order}
	For all $\alpha \in \{0,1,2\}$ and all $\j$, we have
	\begin{align}
		\label{eq:fVf}
		\cs{f^\alpha_\j, V f^\alpha_\j} & =  \frac{d}{dt}\bigg|_{t=0}  \sum_{0\leq \beta < \gamma \leq 2} 	\sigma_\alpha(\beta, \gamma)   \Dd_\j^{\beta \gamma}(t) - 2 \ell^{-2}\cs{ f, (-\Delta) f },
	\end{align}
	where the constants $\sigma_\alpha(\beta, \gamma) \in \{\pm 1\}$ were defined in Eq.~(\ref{eq:sigma constants}).
\end{proposition}
\begin{proof}
	By a direct calculation and due to the self-adjointness of the operators in the commutator, we find
	\begin{align}
		\label{eq:firstdervt}
		\frac{d}{dt}\bigg|_{t=0} \Dd_\j^{\alpha\beta}(t) = \frac{d}{dt}\bigg|_{t=0} \cs{ \psi^{\alpha\beta}_t, (W^{\alpha\beta}_\j)^* P^{\alpha\beta}_\j W^{\alpha\beta}_\j \psi^{\alpha\beta}_t}
		 & = \i \cs{\psi^{\alpha\beta}_0,  [H , (W^{\alpha\beta}_\j)^* P^{\alpha\beta}_\j W^{\alpha\beta}_\j ]  \psi^{\alpha\beta}_0}
	\end{align}
	Now, using the decomposition $\FF = \FF(L^2(B_\j^\alpha \cup B_\j^\beta)) \otimes \FF(L^2(\IR^d \setminus (B_\j^\alpha \cup B_\j^\beta)))$ (cf.~Eq.~(\ref{eq:fock decomposition})) we can write
	\begin{align*}
		\psi^{\alpha\beta}_0 = \phi_\j \otimes \phi_{\not= \j}, \qquad \phi_\j := 2^{-1/2} \left( (\Id + A^{\alpha\beta}_\j) \Omega \right), \quad \phi_{\not= \j} = 2^{-(\abs{\Jf} -1)/2}\prod_{\substack{\k \in \Jf, \\ \k \not= \j}} (\Id + A^{\alpha\beta}_\k) \Omega,
	\end{align*}
	where the vectors $\phi_\j$ and $\phi_{\not= \j}$ are normalized.
	Notice that $\dG(h) \phi_\j \otimes \phi_{\not= \j} = (\dG(h) \phi_j) \otimes \phi_{\not= \j} + \phi_\j \otimes  (\dG(h) \phi_{\not= \j})$, where we use that both tensor product factors lie in the domain of $\dG(h)$, which in this notation only acts on the Fock spaces of the subregions. Furthermore, recall that
	\begin{align*}
		\widetilde P^{\alpha\beta}_\j  := (W^{\alpha\beta}_\j)^* P^{\alpha\beta}_\j W^{\alpha\beta}_\j = \ketbra{(W^{\alpha\beta}_\j)^*  \Omega}{(W^{\alpha\beta}_\j)^* \Omega} \otimes \Id.
	\end{align*}
	Thus, we can compute right-hand side of the inner product in \eqref{eq:firstdervt} as
	\begin{align*}
		[H ,  \widetilde P^{\alpha\beta}_\j ]  \psi^{\alpha\beta}_0 & = \dG(h)  \left(\ketbra{(W^{\alpha\beta}_\j)^*  \Omega}{(W^{\alpha\beta}_\j)^* \Omega} \phi_\j \right) \otimes \phi_{\not= \j} -  \widetilde P^{\alpha\beta}_\j ((\dG(h) \phi_j) \otimes \phi_{\not= \j} + \phi_\j \otimes  (\dG(h) \phi_{\not= \j})) \\
		                                                            & = \cs{(W^{\alpha\beta}_\j)^* \Omega, \phi_\j} \left( ( \dG(h) (W^{\alpha\beta}_\j)^* \Omega) \otimes \phi_{\not=\j} +  ((W^{\alpha\beta}_\j)^* \Omega)\otimes (\dG(h) \phi_{\not= \j}) \right)                                                        \\
		                                                            & \qquad -  \cs{(W^{\alpha\beta}_\j)^* \Omega,  \dG(h) \phi_j} ((W^{\alpha\beta}_\j)^* \Omega) \otimes \phi_{\not= \j} -  \cs{(W^{\alpha\beta}_\j)^* \Omega, \phi_\j}  ((W^{\alpha\beta}_\j)^* \Omega)  \otimes  (\dG(h) \phi_{\not= \j}))              \\
		                                                            & = \cs{(W^{\alpha\beta}_\j)^* \Omega, \phi_\j} (\dG(h) (W^{\alpha\beta}_\j)^* \Omega) \otimes \phi_{\not=\j}  - \cs{(W^{\alpha\beta}_\j)^* \Omega,  \dG(h) \phi_\j} ((W^{\alpha\beta}_\j)^* \Omega) \otimes \phi_{\not= \j}
	\end{align*}
	Therefore,
	\begin{align}
		\i \cs{ \psi^{\alpha\beta}_0, [H ,  \widetilde P^{\alpha\beta}_\j ]  \psi^{\alpha\beta}_0 } & = \i \cs{(W^{\alpha\beta}_\j)^* \Omega, \phi_\j} \cs{ \phi_\j, \dG(h) (W^{\alpha\beta}_\j)^* \Omega} - \i\cs{(W^{\alpha\beta}_\j)^* \Omega,  \dG(h) \phi_\j} \cs{ \phi_\j,(W^{\alpha\beta}_\j)^* \Omega} \nonumber \\
		                                                                                            & = -  2 \Im \cs{(W^{\alpha\beta}_\j)^* \Omega, \phi_\j} \cs{ \phi_\j, \dG(h) (W^{\alpha\beta}_\j)^* \Omega}.
		\label{eq:psi-commutator-psi}
	\end{align}
	Now we use \eqref{eq:controlled-displacement} and the definition of $\phi_\j$, the commutation relations of $\dG(h)$ with creation operators \eqref{eq:af dG commutation relations}, the fact that $\dG(h) \Omega =0$ and $\cs{\Omega,a^*(f^\alpha_\j) a^*(f^\beta_\j) \Omega } =0$, and compute
	\begin{align*}
		 & \cs{(W^{\alpha\beta}_\j)^* \Omega, \phi_\j} \cs{ \phi_\j, \dG(h) (W^{\alpha\beta}_\j)^* \Omega}                                                                                                                                                     \\
		 & \qquad = \frac{1}{4} \cs{(\Id  - \i  a^*(f^\alpha_{\j}) a^*(f^\beta_{\j})) \Omega, (\Id + A^{\alpha\beta}_\j) \Omega} \cs{ (\Id + A^{\alpha\beta}_\j) \Omega, \dG(h) (\Id  - \i  a^*(f^\alpha_{\j}) a^*(f^\beta_{\j})) \Omega}                      \\
		 & \qquad = \frac{1}{4} \left( 1+\i \cs{a^*(f^\alpha_{\j}) a^*(f^\beta_{\j}) \Omega, a^*(f^\alpha_{\j}) a^*(f^\beta_{\j})  \Omega} \right) \cs{ a^*(f^\alpha_{\j}) a^*(f^\beta_{\j}) \Omega, (-\i)  \dG(h)   a^*(f^\alpha_{\j}) a^*(f^\beta_{\j}) \Omega} \\
		 & \qquad= - \frac{\i }{4}  \left( 1+ \i \norm{a^*(f^\alpha_\j) a^*(f^\beta_\j) \Omega}^2\right) \left(\cs{f^\alpha_\j, h f^\alpha_\j}  \norm{f^\beta_\j}^2  + \cs{f^\beta_\j, h f^\beta_\j} \norm{f^\alpha_\j}^2    \right).
	\end{align*}
	We plug this into \eqref{eq:firstdervt} and \eqref{eq:psi-commutator-psi}, using that $\norm{f^\beta_\j} =  \norm{f^\alpha_\j} = 1$. Furthermore, notice that $\cs{f^\alpha_\j,(-\Delta) f^\alpha_\j } = \ell^{-2} \cs{ f, (-\Delta) f }$ for all $\j$ and $\alpha$, since the $f^\alpha_\j$ are translated and $\ell$-scaled  versions of the original profile $f$,
	$f_\j(x) = \ell^{-d/2} f(\ell^{-1}x-\j)$.
	This entails
	\begin{align*}
		\frac{d}{dt}\bigg|_{t=0} \Dd_\j^{\alpha\beta}(t) =  \frac{1}{2} \left(\cs{f^\alpha_\j, h f^\alpha_\j}    + \cs{f^\beta_\j, h f^\beta_\j}    \right)
		=  \frac{1}{2} \left( \cs{f^\alpha_\j, V f^\alpha_\j}  + \cs{f^\beta_\j, V f^\beta_\j} \right) +   \ell^{-2}\cs{ f, (-\Delta) f } .
	\end{align*}
	for all $(\alpha,\beta) = (0,1),(0,2),(1,2)$. We can write these three equations as
	\begin{align*}
		\frac{d}{dt}\bigg|_{t=0} \begin{pmatrix}
			                         \Dd_\j^{01}(t) \\ 	\Dd_\j^{02}(t) \\ \Dd_\j^{12}(t)
		                         \end{pmatrix}  = \frac{1}{2} \begin{pmatrix}
			                                                      1 & 1 & 0 \\
			                                                      1 & 0 & 1 \\
			                                                      0 & 1 & 1
		                                                      \end{pmatrix} \begin{pmatrix}
			                                                                    \cs{f^0_\j, V f^0_\j} \\
			                                                                    \cs{f^1_\j, V f^1_\j} \\
			                                                                    \cs{f^2_\j, V f^2_\j}
		                                                                    \end{pmatrix} + \ell^{-2}\cs{ f, (-\Delta) f } \begin{pmatrix} 1 \\ 1 \\ 1 \end{pmatrix},
	\end{align*}
	whose solution is
	\begin{align*}
		\begin{pmatrix}
			\cs{f^0_\j, V f^0_\j} \\
			\cs{f^1_\j, V f^1_\j} \\
			\cs{f^2_\j, V f^2_\j}
		\end{pmatrix} = \frac{d}{dt}\bigg|_{t=0} \begin{pmatrix}
			                                         1  & 1  & -1 \\
			                                         1  & -1 & 1  \\
			                                         -1 & 1  & 1
		                                         \end{pmatrix} \begin{pmatrix}
			                                                       \Dd_\j^{01}(t) \\ 	\Dd_\j^{02}(t) \\ \Dd_\j^{12}(t)
		                                                       \end{pmatrix} - 2 \ell^{-2}\cs{ f, (-\Delta) f } \begin{pmatrix} 1 \\ 1 \\ 1 \end{pmatrix}.
	\end{align*}
	Reading it line by line yields the desired equations \eqref{eq:fVf}.
\end{proof}
In the next step, we have to bound the second commutator, however, at arbitrary small times $t > 0$.
Since the time evolution by a Schrödinger operator is not strictly local, the fermions first localized in the respective boxes can leak into neighboring boxes as soon as $t > 0$, which makes the computation significantly more challenging than in \Cref{th:first_order}. However, this leakage is small due to continuous Lieb-Robinson bounds. More precisely, we will use the following variant of the one-body bound from \cite{hinrichs2023lieb}.
\begin{theorem}
	\label{th:modified lrb}
	Let $h \coloneqq -  \Delta + V$ and assume that $V$ is $2(n+j-1)$ times differentiable with bounded derivatives. Let $\chi$ be compactly supported, $f$ be a Schwartz function, and $T_\k f \coloneqq f(\cdot - \k)$ be the unitary translation operator. Furthermore, let
	\begin{align*}
		U_\ell f(x) \coloneqq \ell^{-d/2} f(x / \ell)
	\end{align*}
	be the unitary dilation operator on $L^2(\R^d)$, and set $f_\k = U_\ell T_\k f$, corresponding to the notation before in Eq.~(\ref{eq:f_j defn}).
	Let $\chi_\ell(x) = \chi(x/\ell)$.
	Then, for all $n \in \N$, there exists a constant $C_n$ such that for all $\ell \in (0,1]$ and all $j =0,1,2$,
	\begin{align*}
		\norm{\chi_\ell e^{-\i t h} h^j f_\kb } \leq C_n \ell^{-2j-2n-1} \cs t^{1/2} \norm{f} \left( 1 \wedge \frac{\cs t}{\dist(\k ,\supp \chi)} \right)^n,
	\end{align*}
	where $h^0 := \Id$, $h^1 := h$, $h^2 := h h$ (applied $h$ twice) and $\cs{t} = \sqrt{t^2+1}$.
\end{theorem}
\begin{proof}
	First, notice that a direct calculation shows $U_\ell^* = U_{1/\ell}$ and therefore,
	\begin{align*}
		U_\ell^* \zeta  U_\ell = \zeta( \ell \cdot ), \qquad U_\ell^* (-\Delta) U_\ell = \ell^{-2} (-\Delta),
	\end{align*}
	for a function $\zeta \colon \IR^d \rightarrow \IR$ interpreted as a multiplication operator on $L^2(\IR^d)$. In particular, we get $U_\ell^* \chi_\ell U_\ell = \chi$ and
	\begin{align}
		\label{eq:effective time}
		h_\ell \coloneqq \ell^2 U_\ell^* h U_\ell = \ell^2 U_\ell^* (-\Delta + V) U_\ell = \ell^2( \ell^{-2}(-\Delta) + V(\ell \cdot)) =  - \Delta + \ell^2 V(\ell \cdot),
	\end{align}
	so that we can write,
	\begin{align}
		\label{eq:modified est for theorem}
		\norm{ \chi_\ell e^{-\i t h} h^j f_\kb}  = \norm{ \chi_\ell e^{-\i t h} h^j U_\ell T_\k f } = \norm{ U_\ell^* \chi_\ell U_\ell U_\ell^* e^{-\i  t h} h^j U_\ell T_\k f } = \ell^{-2j} \norm{   \chi  e^{-\i \ell^{-2} t  h_\ell} h_\ell^j  T_\k f }.
	\end{align}
	Now the claimed result will follow from a modified form of \cite[Proposition 3.4 (i)]{hinrichs2023lieb}. A first modification we need is that the bounds have to be formulated in terms of norms instead of inner products. More precisely, in \cite[Proposition 3.4 (i)]{hinrichs2023lieb} we can change the first bound to
	\begin{equation*}
		\norm{f e^{-\i t T}\varphi_x} \leq  \Cobt{0} \norm{f}_\infty \braket{t}^\delta G_{n,t}(R_f),
	\end{equation*}
	where $f$ on the left hand side denotes the multiplication operator with $f$, $\norm{\cdot}$ the $L^2$ norm and $\norm{\cdot}_\infty$ the supremum norm. The proof works in the same way. We only have to replace the inner products in the first two displayed equations by norms and instead of Cauchy-Schwarz we use the operator norm estimate for the multiplication operator with $f$, i.e., (3.25) in \cite[Proposition 3.4 (i)]{hinrichs2023lieb} becomes
	\begin{align*}
		\nn{ f e^{-\i t T} T^j \varphi_x} \leq \nn{f}_\infty  \nn{e^{-\i t T} (1-B) T^j \varphi_x} +  \nn{f}_\infty \nn{\chr_{\abs{\cdot - x} \geq R_f} e^{-\i t T} B} \nn{T^j \varphi_x}.
	\end{align*}
	We then apply \cite[Proposition 3.4 (i)]{hinrichs2023lieb} by using the following assignments of the variables therein:
	\begin{align*}
		f \rightarrow \chi, \qquad T \rightarrow h_\ell, \text{ i.e., } V \rightarrow \ell^2 V(\ell\cdot), \qquad t \rightarrow \ell^{-2} t, \qquad \varphi_x \rightarrow T_\k f,
	\end{align*}
	where the left-hand side denotes the notation in \cite[Proposition 3.4 (i)]{hinrichs2023lieb}. This yields the desired estimate for the last term in \eqref{eq:modified est for theorem} for $j=0$,
	\begin{align}
		\label{eq:apply lrb theorem}
		\norm{   \chi  e^{-\i \ell^{-2} t  h_\ell} h_\ell^j T_\k f } & \leq \Cobt{0}   (\ell^{-2}  \cs{t})^{1/2} \left(1 \wedge \frac{\ell^{-2}\cs{t}}{\dist(\k,\supp \chi)}\right)^n ,
	\end{align}
	where we set $\delta =1/2$, and used that $\nn{\chi}_\infty \leq 1$ and $\cs{\ell^{-2} t} \leq \ell^{-2} \cs t$ for $\ell \in (0,1]$. Notice that the proposition can actually applied with the new effective potential $\ell^2 V(\ell \cdot)$, as the potential in the original proof only enters in the form of infinity norms of itself and its derivatives. They can be uniformly bounded from above for all $\ell \in (0,1]$.

	Another necessary modification is the generalization of \eqref{eq:apply lrb theorem} to $j =1,2$. The only place where the terms $h_\ell^j$ enter in the proof is when  \cite[Lemma 3.3]{hinrichs2023lieb} is applied or when $\nn{T^j \varphi_x}$ is estimated. The latter can be estimated by a constant depending on the infinity norms of the derivatives of $V$ up to order $2j-1$. Furthermore, from the proof it is easy to see that (a) in \cite[Lemma 3.3]{hinrichs2023lieb}  generalizes $j =1,2$,
	\begin{equation*}
		\norm{\chiRx{R}{x} h_\ell^j \varphi_x} \leq C^{n,j}_{1} R^{-n} \qquad \text{for all}\ x\in\R^d, R>0,~\ell \in (0,1],
	\end{equation*}
	where the constant $C^{n,j}_{1}$ depends on the supremum norms of the derivatives of $V$ up to order $2(j-1)$. The same is true for the second bound in (b), which involves derivatives up to order $2(n+j-1)$,
	\begin{equation*}
		\norm{\chr_{[E,\infty)}(h_\ell) h_\ell^j \varphi_x} \leq  \frac{C^{n,j}_2}{E^n} \qquad \mbox{for all}\ x\in\R^d,\ E>0, ~\ell \in (0,1]. 
	\end{equation*}
	Thus, having established those bounds, the same estimate as in \eqref{eq:apply lrb theorem} with $j=1,2$ holds.
\end{proof}

\begin{proposition}\label{th:second order}
	\begin{itemize}
		\item[(a)]  In the case of one triple, $\abs{\Jf} = 1$ assume that $V$ is relatively bounded with respect to $-\Delta$. Then there exists a constant $C$ such that, for all $t \in \R$ and $\ell \in (0,1]$,
		      \begin{align*}
			      \norm{  	\frac{d^2}{dt^2} \cs{W^{\alpha\beta}_\j \psi^{\alpha\beta}_t,   P^{\alpha\beta}_\j W^{\alpha\beta}_\j \psi^{\alpha\beta}_t}  }  \leq  C  \ell^{-4}.
		      \end{align*}
		\item[(b)] For the general case $\abs{\Jf} > 1$ assume that $V$ is $4(d+1)$ times differentiable with bounded derivatives.
		      Then there exists a constant $C$ such that, for all $t \in \R$ and $\ell \in (0,1]$,
		      \begin{align*}
			      \norm{  	\frac{d^2}{dt^2} \cs{W^{\alpha\beta}_\j \psi^{\alpha\beta}_t,   P^{\alpha\beta}_\j W^{\alpha\beta}_\j \psi^{\alpha\beta}_t}  }  \leq  C  \cs t^{2d+3} \ell^{-\gd}\,,
		      \end{align*}
		      where
		      \begin{align}
			      \label{eq:gd constant}
			      \gd \coloneqq  4d +10.
		      \end{align}
	\end{itemize}

\end{proposition}
\begin{proof}
	Let $\chi_\ell$ be a smooth function which equals one in the two boxes $B^\alpha_\j$ and $B^\beta_\j$, which decays to zero in the neighboring boxes and which vanishes on all other boxes, see \Cref{fig:chi functions} for a depiction in 1D. Furthermore, assume there is a function $\chi$, independent of $\ell$ such that $\chi_\ell = \chi(\cdot / \ell)$. Similarly, let $\chit_\ell$ be a smooth function which equals one in the boxes  $B^\alpha_\j$, $B^\beta_\j$ and in all their direct neighboring boxes, which decays to zero in all of their neighbors (of the boxes where $\chit_\ell = 1$), and which vanishes otherwise.

	\begin{figure}
		\begin{tikzpicture}
			\begin{axis}[
					axis x line=middle, axis y line=middle,
					xmin=0, xmax=18, ymin=-0.1, ymax=1.2,
					xtick={1,3,5,7,9,11,13,15,17},
					xticklabels={$B_0$, $B_1$, $B_2$, $B_3$, $B_4$, $B_5$, $B_6$, $B_7$, $B_8$},
					ytick={0,1},
					xlabel={$x$}, ylabel={$\textcolor{blue}{\chi_\ell(x)}, \textcolor{red}{\chit_\ell(x)}$},
					samples=200, domain=0:18, yscale = 0.5
				]

				\addplot[blue, thick]
				{ (x < 8) * exp(-10*(8-x)^2) +
					(x >= 8 && x <= 12) * 1 +
					(x > 12) * exp(-10*(x-12)^2)
				};

				\addplot[red, thick, dashed]
				{ (x < 6) * exp(-10*(6-x)^2) +
					(x >= 6 && x <= 14) * 1 +
					(x > 14 && x <= 16) * exp(-10*(x-14)^2)
				};

				\foreach \x in {0,2,4,6,8,10,12,14,16} {\addplot[dashed] coordinates {(\x,0) (\x,1)};
					}
			\end{axis}
		\end{tikzpicture}
		\caption{\label{fig:chi functions}A suitable function $\chi_\ell$ (blue) and $\chit_\ell$ (red) in $d=1$ for $\j = 1$ and $(\alpha,\beta)=(1,2)$. The function $\chi_\ell$ is equal to one on $B_4$ and $B_5$, and decays to zero in $B_3$ and $B_6$. In all other boxes it equals zero. In particular, any derivative of $\chi_\ell$ is only non-zero in $B_3$ and $B_6$. The function $\chit_\ell$ also equals one in the neighboring boxes of the ones where $\chi_\ell$ is one, namely in $B_3$ and $B_6$.}
	\end{figure}

	We use the IMS localization formula (e.g., see \cite[Theorem 3.2]{cycon1987schrodinger}) and decompose, for $\chip_\ell \coloneqq \sqrt{  1- \chi_\ell^2 }$,
	\begin{align}
		\label{eq:dGh decomposition}
		\dG(h) = \dG(\chi_\ell h \chi_\ell)   +  \dG(\chip_\ell (-\Delta) \chip_\ell) + \dG((\nabla \chi_\ell)^2) + \dG((\nabla \chip_\ell)^2),
	\end{align}
	and similarly with $\chit_\ell$ and $\chitp_\ell$. Now notice that the operators $\chip_\ell (-\Delta) \chip_\ell$, $(\nabla \chi_\ell)^2$ and $(\nabla \chip_\ell)^2$ vanish on $L^2(B_\j^\alpha \cup B_\j^\beta)$. Hence, $(W^{\alpha\beta}_\j)^* P^{\alpha\beta}_\j W^{\alpha\beta}_\j$ commutes with all the operators of the right hand side of \Cref{eq:dGh decomposition} except the first one, which entails $[H,(W^{\alpha\beta}_\j)^* P^{\alpha\beta}_\j W^{\alpha\beta}_\j] = [\dG(\chi_\ell h \chi_\ell),(W^{\alpha\beta}_\j)^* P^{\alpha\beta}_\j W^{\alpha\beta}_\j]$. Similarly, $\chitp_\ell (-\Delta) \chitp_\ell$, $(\nabla \chit_\ell)^2$ and $(\nabla \chitp_\ell)^2$ vanish on $L^2(B)$, where $B$ denotes the union of the boxes $B_\j^\alpha$, $B_\j^\beta$ and all of their immediate neighbors. At the same time, the commutator $[\dG(\chi_\ell h \chi_\ell),(W^{\alpha\beta}_\j)^* P^{\alpha\beta}_\j W^{\alpha\beta}_\j]$ only acts on $\FF(L^2(B))$ non-trivially.
	Therefore,
	\begin{align*}
		E_{2,t} \coloneqq \frac{d^2}{dt^2} \cs{W^{\alpha\beta}_\j \psi^{\alpha\beta}_t,   P^{\alpha\beta}_\j W^{\alpha\beta}_\j \psi^{\alpha\beta}_t} & = - \cs{\psi^{\alpha\beta}_t, [H,[H,(W^{\alpha\beta}_\j)^* P^{\alpha\beta}_\j W^{\alpha\beta}_\j]] \psi^{\alpha\beta}_t}                                                        \\
		                                                                                                                                              & = -\cs{\psi^{\alpha\beta}_t, [\dG(\chit_\ell h \chit_\ell) ,[\dG(\chi_\ell h \chi_\ell) ,(W^{\alpha\beta}_\j)^* P^{\alpha\beta}_\j W^{\alpha\beta}_\j]] \psi^{\alpha\beta}_t} .
	\end{align*}

	We first prove (b), and use self-adjointness, $\nn{(W^{\alpha\beta}_\j)^* P^{\alpha\beta}_\j W^{\alpha\beta}_\j} = 1$, $\norm{\psi^{\alpha\beta}_t} =1$ and the Cauchy-Schwarz inequality to obtain
	\begin{align}
		\abs{E_{2,t}} & \leq   2 \abs{ \cs{\dG(\chit_\ell h \chit_\ell) \psi^{\alpha\beta}_t, [ \dG(\chi_\ell h \chi_\ell), (W^{\alpha\beta}_\j)^* P^{\alpha\beta}_\j W^{\alpha\beta}_\j] \psi^{\alpha\beta}_t   }} \nonumber                                               \\
		              & \leq 2 \left(\norm{ \dG(\chi_\ell h \chi_\ell) \dG(\chit_\ell h \chit_\ell) \psi^{\alpha\beta}_t } + \norm{ \dG(\chit_\ell h \chit_\ell) \psi^{\alpha\beta}_t} \norm{ \dG(\chi_\ell h \chi_\ell) \psi^{\alpha\beta}_t}  \right). \label{eq:E2t est}
	\end{align}
	Now, let us write $f^\alpha_{\j,t} = e^{-\i t h} f^\alpha_\j$ and  $A^{\alpha\beta}_{\j,t} = e^{-\i t H} A^{\alpha\beta}_{\j} e^{\i t H} = a^*(f^\alpha_{\j,t})  a^*(f^\beta_{\j,t})$, so that, analogously to \Cref{eq:psi0_defn}, we have
	\begin{align*}
		\psi^{\alpha\beta}_t = 2^{-\abs{\Jf} / 2}  \prod_{\j \in \Jf} (\Id + \Aj{\j,t}{\alpha\beta} ) \Omega.
	\end{align*}
	As $\dG(T) \Omega = 0$, we have $\dG(T) A \Omega = [\dG(T),A] \Omega$  for any self-adjoint operator $T$ and bounded operator $A$. Thus,
	\begin{align}
		\dG(\chi_\ell h \chi_\ell) \psi^{\alpha\beta}_t                              & = 2^{-\abs{\Jf} / 2}  \left[\dG(\chi_\ell h \chi_\ell),\prod_{\j \in \Jf} (\Id + \Aj{\j,t}{\alpha\beta} )\right] \Omega \nonumber
		\\
		                                                                             & = 2^{-\abs{\Jf} / 2}  \sum_{\k \in \Jf} [\dG(\chi_\ell h \chi_\ell),\Aj{\k,t}{\alpha\beta} ] \prod_{\substack{\j \in \Jf                                        \\ \j \not= \k}} (\Id + \Aj{\j,t}{\alpha\beta}) \Omega, \label{eq:dG psi} \\
		\dG(\chi_\ell h \chi_\ell) \dG(\chit_\ell h \chit_\ell) \psi^{\alpha\beta}_t & =  2^{-\abs{\Jf} / 2}  \sum_{\k \in \Jf} \bigg( [\dG(\chi_\ell h \chi_\ell),[\dG(\chit_\ell h \chit_\ell),\Aj{\k,t}{\alpha\beta} ]] \prod_{\substack{\j \in \Jf \\ \j \not= \k}}  (\Id + \Aj{\j,t}{\alpha\beta}) \Omega \nonumber \\
		                                                                             & \qquad +  [\dG(\chit_\ell h \chit_\ell),\Aj{\k,t}{\alpha\beta} ] \sum_{\substack{\q \in \Jf,                                                                    \\ \q \not= \k}} [\dG(\chi_\ell h \chi_\ell),\Aj{\q,t}{\alpha\beta} ] \prod_{\substack{\j \in \Jf \\ \j \not= \k,\q}}  (\Id + \Aj{\j,t}{\alpha\beta}) \Omega \bigg), \label{eq:dG dG psi}
	\end{align}
	and
	\begin{align}
		[\dG(\chi_\ell h \chi_\ell), \Aj{\k,t}{\alpha\beta} ]                                & = a^*(\chi_\ell h \chi_\ell f^\alpha_{\k,t}) a^*(f^\beta_{\k,t}) + a^*(f^\alpha_{\k,t}) a^*(\chi_\ell h \chi_\ell f^\beta_{\k,t}), \label{eq:dG A commutator}                               \\
		[\dG(\chi_\ell h \chi_\ell),[\dG(\chit_\ell h \chit_\ell), \Aj{\k,t}{\alpha\beta} ]] & = a^*(\chi_\ell h \chi_\ell \chit_\ell h \chit_\ell f^\alpha_{\k,t}) a^*(f^\beta_{\k,t}) + a^*(\chit_\ell h \chit_\ell f^\alpha_{\k,t}) a^*(\chi_\ell h \chi_\ell f^\beta_{\k,t}) \nonumber \\ &\qquad + a^*(\chi_\ell h \chi_\ell f^\alpha_{\k,t}) a^*(\chit_\ell h \chit_\ell f^\beta_{\k,t}) +  a^*(f^\alpha_{\k,t}) a^*(\chi_\ell h \chi_\ell \chit_\ell h \chit_\ell f^\beta_{\k,t}) . \label{eq:dG dG A commutator}
	\end{align}
	Then we write
	\begin{align}\label{eq:chi_h_chi_ft}
		\chi_\ell h \chi_\ell  = \chi_\ell^2 h + \chi_\ell [-\Delta, \chi_\ell]  = \chi_\ell^2 h  + \chi_\ell (-\Delta \chi_\ell)   +  \chi_\ell  (\nabla \chi_\ell)\cdot \nabla,
	\end{align}
	and claim  that there exist smooth functions $\xi_0, \xi_1$ independent of $\ell$ with $\supp \xi_j \subseteq \supp \chi$ such that
	\begin{align}
		\label{eq:xhx f est}
		\norm{ \chi_\ell h \chi_\ell f^\alpha_{\k,t} } \leq \sum_{j=0}^1 \ell^{-2(1-j)} \norm{ \xi_j(\cdot / \ell) e^{-\i t h} h^j f^\alpha_{\k} }.
	\end{align}
	This is obvious for the first two terms on the right hand side of \eqref{eq:chi_h_chi_ft}. In order to see \eqref{eq:xhx f est} also for the last operator in \eqref{eq:chi_h_chi_ft} applied to $f^\alpha_{\k,t}$, we use \Cref{th:L2 norm estimate lemma} and get
	\begin{align}
		\label{eq:lemma applied}
		\norm{ \chi_\ell  (\nabla \chi_\ell)\cdot \nabla f^\alpha_{\k,t}  } & \leq C_d \sum_{i=1}^d	\left(
		\sqrt{ \| \chi_\ell (\partial_i \chi_\ell) f^\alpha_{\k,t}\| \| \chi_\ell (\partial_i \chi_\ell)  \Delta f^\alpha_{\k,t}\| }
		+\|\nabla (\chi_\ell (\partial_i \chi_\ell) )\, f^\alpha_{\k,t}\|
		\right).
	\end{align}
	Notice that chain rule gives $\partial_i \chi_\ell = \ell^{-1} (\partial_i \chi)(\cdot / \ell)$. Hence, the second term in \eqref{eq:lemma applied} scales like $\ell^{-2}$, so it can be written in the form of the summand $j=0$ in \eqref{eq:xhx f est}. The first one can be estimated as follows, using $\sqrt{ab} \leq \frac{a}{2} + \frac{b}{2}$ and the fact that $V$ is bounded,
	\begin{align*}
		\sqrt{ \| \chi_\ell (\partial_i \chi_\ell) f^\alpha_{\k,t}\| \| \chi_\ell (\partial_i \chi_\ell)  \Delta f^\alpha_{\k,t}\| } & = \sqrt{ \ell^{-2} \| \chi_\ell (\partial_i \chi)(\cdot / \ell) f^\alpha_{\k,t}\| \| \chi_\ell (\partial_i \chi)(\cdot / \ell)  \Delta f^\alpha_{\k,t}\| }                                                                                                       \\
		                                                                                                                             & \leq \frac{1}{2}\left( \ell^{-2} \| \chi_\ell (\partial_i \chi)(\cdot / \ell) f^\alpha_{\k,t}\| + \| \chi_\ell (\partial_i \chi)(\cdot / \ell)  \Delta f^\alpha_{\k,t}\|  \right)                                                                                \\
		                                                                                                                             & \leq \frac{1}{2}\left( \ell^{-2} \| \chi_\ell (\partial_i \chi)(\cdot / \ell) f^\alpha_{\k,t}\| +  \nn{V}_\infty  \| \chi_\ell (\partial_i \chi)(\cdot / \ell)  f^\alpha_{\k,t}\|  + \| \chi_\ell (\partial_i \chi)(\cdot / \ell)  h f^\alpha_{\k,t}\|  \right).
	\end{align*}
	Again we see that the first and second term are of the form as the summand for $j= 0$ on the right-hand side of \eqref{eq:xhx f est}, whereas the last term is of the form as the summand for $j= 1$. In conclusion, we have shown that \eqref{eq:xhx f est} holds.


	With a similar argument, one can show that
	\begin{align}
		\label{eq:chitchitchichi est}
		\norm{ \chit_\ell h \chit_\ell \chi_\ell h \chi_\ell f^\alpha_{\k,t} } \leq \sum_{j=0}^2 \ell^{-2(2-j)} \norm{ \xi_j(\cdot / \ell) e^{-\i t h} h^j f^\alpha_{\k} },
	\end{align}
	where $\xi_j$, $j=0,1,2$, are again smooth functions not depending on $\ell$ with $\supp \xi_j \subseteq \supp \chit$.

	Now, combine \Cref{eq:dG psi} and (\ref{eq:dG A commutator}). Then we group the triples indexed by $\kb$ according to the integer supremum metric to the triple where $\chi$ is supported. To this end, let $d^\infty_\kb := \dist_\infty(\supp \chi,[0,3]^d + \kb)$ denote the distance between the support of $\chi$ and the box $[0,3]^d + \kb$ where the unscaled triple (i.e., for $\ell=1$) at position $\kb$ is contained in, and note that it has values in $\mathbb N_0$. We get
	\begin{align}
		\label{eq:dG psit est}
		\norm{ \dG(\chi_\ell h \chi_\ell) \psi^{\alpha\beta}_t  } & \leq  \frac{1}{2} \sum_{\k \in \Jf}  \left(\norm{\chi_\ell h \chi_\ell f^\alpha_{\k,t} }   + \norm{\chi_\ell h \chi_\ell f^\beta_{\k,t} }  \right)
		\leq  \frac{1}{2} \sum_{r=0}^\infty \sum_{\substack{\k \in \Jf,                                                                                                                                                \\ r=d^\infty_\kb}}  \left(\norm{\chi_\ell h \chi_\ell f^\alpha_{\k,t} }   + \norm{\chi_\ell h \chi_\ell f^\beta_{\k,t} }  \right).
	\end{align}
	The terms in the sum can be estimated with \eqref{eq:xhx f est} and the LRBs of \Cref{th:modified lrb}, which gives
	\begin{align*}
		\norm{\chi_\ell h \chi_\ell f^\alpha_{\k,t} }   + \norm{\chi_\ell h \chi_\ell f^\beta_{\k,t}}  \leq  C \ell^{-2n -3}  \cs t^{1/2}  \left(1 \wedge \frac{\cs t}{r}\right)^{n}
	\end{align*}
	for some constant $C$ (depending on $n,d$ but not on $\ell$), which will now change from line to line.
	The number of triple boxes such that $r=d^\infty_\kb$ behaves like $\cO((r+1)^{d-1})$. Thus, we get, setting $n = d +1$,
	\begin{align*}
		\norm{ \dG(\chi_\ell h \chi_\ell) \psi^{\alpha\beta}_t  } & \leq  C \ell^{-2(d+1)-3}  \cs t^{1/2+d+1}  \sum_{r=0}^\infty   \frac{(r+1)^{d-1}}{1 \wedge r^{d+1}}
		\leq  C \ell^{-2d -5}  \cs t^{d+3/2} ,
	\end{align*}
	where the series converges for our choice of $n$.

	Next, we use \Cref{eq:dG dG psi,eq:dG dG A commutator} and consider
	\begin{align}
		\norm{ \dG(\chi_\ell h \chi_\ell) \dG(\chit_\ell h \chit_\ell) \psi^{\alpha\beta}_t  } & \leq 2^{-1/2} \sum_{\k \in \Jf} \norm{	[\dG(\chi_\ell h \chi_\ell),[\dG(\chit_\ell h \chit_\ell), \Aj{\k,t}{\alpha\beta} ]]} \nonumber                                                                                                              \\
		                                                                                       & \qquad + 2^{-1} \sum_{\k,\q \in \Jf} \norm{ [\dG(\chit_\ell h \chit_\ell),\Aj{\k,t}{\alpha\beta} ]}  \norm{ [\dG(\chi_\ell h \chi_\ell),\Aj{\q,t}{\alpha\beta}] } \nonumber                                                                         \\
		                                                                                       & \leq 2^{-1/2} \cdot 2  \sum_{\k \in \Jf} \left( \norm{ \chi_\ell h \chi_\ell f^\alpha_{\k,t}} \norm{ \chit_\ell h \chit_\ell f^\beta_{\k,t}} +  \norm{ \chi_\ell h \chi_\ell \chit_\ell h \chit_\ell f^\beta_{\k,t}} \right) \label{eq:dGdG term12} \\
		                                                                                       & \qquad +  2^{-1} \sum_{\k,\q \in \Jf} \norm{ [\dG(\chit_\ell h \chit_\ell),\Aj{\k,t}{\alpha\beta} ]}  \norm{ [\dG(\chi_\ell h \chi_\ell),\Aj{\q,t}{\alpha\beta}] } \label{eq:dGdG term3}
	\end{align}
	The term \eqref{eq:dGdG term3} can be bounded with the same estimate as before by the square, i.e.,  $C (  \ell^{-2d -5}  \cs t^{d+3/2})^2$. The second term in \Cref{eq:dGdG term12} can be treated in the same way as the previous estimate using \Cref{eq:chitchitchichi est} and thus differs from the previous estimate only by a factor of $\ell^{-2}$ due to an additional operator $h$. This means that it can be estimated by $C   \ell^{-2d -7}  \cs t^{d+3/2}$.  The first term in \Cref{eq:dGdG term12} can be bounded from above by
	\begin{align*}
		\sum_{\k \in \Jf}  \norm{ \chi_\ell h \chi_\ell f^\alpha_{\k,t}} \norm{ \chit_\ell h \chit_\ell f^\beta_{\k,t}} & \leq   \frac{1}{2}\sum_{\k \in \Jf}  \left( \norm{ \chi_\ell h \chi_\ell f^\alpha_{\k,t}}^2 +  \norm{ \chit_\ell h \chit_\ell f^\beta_{\k,t}}^2 \right) \leq  C \ell^{-4n-6}  \cs t^{1+2n}  \sum_{r=0}^\infty   \frac{(r+1)^{d-1}}{1 \wedge r^{2n}},
	\end{align*}
	where the only difference to the first estimate are the squares inside the sum. The series converges if we set $n= \lceil (d+1)/2 \rceil$, i.e., $d+1 \leq 2n \leq d+2$. Hence, the whole term can be estimated by $C \cs t^{d+3} \ell^{-2d -10}$.

	Finally, by collecting all estimates, only tracking the worst $\ell$-dependence (using that $\ell \leq 1$), we arrive at
	\begin{align*}
		\norm{ \dG(\chi_\ell h \chi_\ell) \dG(\chit_\ell h \chit_\ell) \psi^{\alpha\beta}_t  }  \leq C \ell^{-4d -10} \cs{t}^{2d+3}, \\
		\norm{ \dG(\chi_\ell h \chi_\ell) \psi^{\alpha\beta}_t} \norm{ \dG(\chit_\ell h \chit_\ell) \psi^{\alpha\beta}_t}   \leq C \ell^{-4d -10} \cs{t}^{2d+3} .
	\end{align*}
	This together with \Cref{eq:E2t est} proves the result (b).

	Finally, let us prove result (a), where we assume $\abs{\Jf} =1$, but the only condition we impose on the potential $V$ is that it is relatively bounded with respect to the Laplacian.  Recall that we defined $P^{\alpha\beta}_\j =\ketbra{\Omega}{\Omega}_{\widetilde B} \otimes \Id_{\FF(L^2(\IR^d \setminus \widetilde B))}$, where $\widetilde B$ denotes the three boxes corresponding to $\j$ united with all its neighboring boxes. Then we can write
	\begin{align*}
		\dG(\chi_\ell h \chi_\ell) = \dG(\chi_\ell h \chi_\ell) \otimes \Id, \qquad (W^{\alpha\beta}_\j)^* P^{\alpha\beta}_\j W^{\alpha\beta}_\j = \ketbra{(W^{\alpha\beta}_\j)^* \Omega}{(W^{\alpha\beta}_\j)^* \Omega} \otimes \Id,
	\end{align*}
	where the first tensor factor acts on $\FF(L^2(\widetilde B))$ and the second one on $\FF(L^2(\IR^d \setminus \widetilde B))$.
	This allows us to obtain a direct expression for the double commutator appearing in $E_{2,t}$ and an alternative to its estimate \eqref{eq:E2t est}.
	Namely, we get
	\begin{align*}
		[\dG(\chi_\ell h \chi_\ell) ,(W^{\alpha\beta}_\j)^* P^{\alpha\beta}_\j W^{\alpha\beta}_\j] = \ketbra{\dG(\chi_\ell h \chi_\ell) (W^{\alpha\beta}_\j)^* \Omega}{(W^{\alpha\beta}_\j)^* \Omega} \otimes \Id - \ketbra{(W^{\alpha\beta}_\j)^* \Omega}{\dG(\chi_\ell h \chi_\ell) (W^{\alpha\beta}_\j)^* \Omega}    \otimes \Id .
	\end{align*}
	Thus,
	\begin{align}
		\label{eq:E2t est J1}
		\abs{E_{2,t}} & \leq 2 \nn{\dG(\chit_\ell h \chit_\ell) \psi^{\alpha\beta}_t } \nn{ [\dG(\chi_\ell h \chi_\ell) ,(W^{\alpha\beta}_\j)^* P^{\alpha\beta}_\j W^{\alpha\beta}_\j]  \psi^{\alpha\beta}_t} \leq 4 \nn{\dG(\chit_\ell h \chit_\ell) \psi^{\alpha\beta}_t }  \nn{\dG(\chi_\ell h \chi_\ell) (W^{\alpha\beta}_\j)^* \Omega}.
	\end{align}
	The second factor on the right hand side of \eqref{eq:E2t est J1} can be estimated with \eqref{eq:xhx f est} by
	\begin{align*}
		\nn{\dG(\chi_\ell h \chi_\ell) (W^{\alpha\beta}_\j)^* \Omega} = \nn{\dG(\chi_\ell h \chi_\ell) a^*(f^\alpha_{\j,t})  a^*(f^\beta_{\j,t})\Omega} \leq \nn{\chi_\ell h \chi_\ell f^\alpha_{\j,t}} + \nn{\chi_\ell h \chi_\ell f^\beta_{\j,t}}  \leq C \ell^{-2},
	\end{align*}
	where we bound $\xi_j(\cdot / \ell)$ by its supremum norm, use $\nn{ e^{-\i t h} } \leq 1$, and and the fact that $V$ is relatively bounded with respect to $-\Delta$,
	\begin{align*}
		\nn{h f^\alpha_{\j,t}} \leq C_1 \nn{(-\Delta+1) f^\alpha_{\j,t}} \leq C_2 \ell^{-2}, \qquad C_1,C_2 > 0.
	\end{align*}
	The first factor can be treated with the first inequality of \eqref{eq:dG psit est} and \eqref{eq:xhx f est} again, which yields the same bound $C \ell^{-2}$. This shows the statement (a).
\end{proof}
Now, the proof of the main result of this section follows by combining \Cref{th:first_order,th:second order}.
\begin{proof}[Proof of Theorem \ref{th:final_fVf_D_estimate}]
	Through a Taylor expansion with a second order remainder term, we can write
	\begin{align*}
		\dDd{t}{\j}{\alpha \beta} & \coloneqq \frac{ \Dd_\j^{\alpha \beta} (t) -  \Dd_\j^{\alpha \beta} (0)  }{t}  =    \frac{d}{dt}\bigg|_{t=0} \Dd_\j^{\alpha \beta} (t) + t \frac{d^2}{d t^2}\bigg|_{t = t_{\alpha \beta}}  \Dd_\j^{\alpha \beta} (t)
	\end{align*}
	for some numbers  $t_{\alpha \beta} \in \IR$, thus,
	\begin{align}
		\label{eq:first dev formula}
		\frac{d}{dt}\bigg|_{t=0} \Dd_\j^{\alpha \beta} (t) =  \dDd{t}{\j}{\alpha \beta}  - t \frac{d^2}{d t^2}\bigg|_{t = t_{\alpha \beta}}  \Dd_\j^{\alpha \beta} (t).
	\end{align}
	Using \Cref{eq:first dev formula} in \Cref{th:first_order} yields

	\begin{align*}
		\cs{f^\alpha_\j, V f^\alpha_\j}  =   \sum_{1\leq \beta < \gamma \leq 3} 	\sigma_\alpha(\beta, \gamma)  \left( \dDd{t}{\j}{\alpha \beta}  - t \frac{d^2}{d t^2}\bigg|_{t = t_{\alpha \beta}}  \Dd_\j^{\alpha \beta} (t) \right) -  2\ell^{-2}\cs{ f, (-\Delta) f }.
	\end{align*}
	Then we use \Cref{th:second order} to estimate the second derivatives.
\end{proof}
In the proof we were using the following elementary inequality to estimate single gradients by the Laplacian.
\begin{lemma}
	\label{th:L2 norm estimate lemma}
	There exists a universal constant $C$ such that
	for all $\eta \in W^{1,\infty}(\IR^d)$ (i.e., with an essentially bounded first derivative) and $f \in H^2(\IR^d)$,
	\begin{equation}\label{eq:localized-gradient}
		\|\eta \nabla f\|
		\leq
		C\Big(
		\sqrt{\|\eta f\| \|\eta \Delta f\|  }
		+\|(\nabla \eta) f\|
		\Big),
	\end{equation}
	where $\nn{\cdot}$ denotes the $L^2$-norm.

	In particular, there exists a constant $C_d$ depending only on the dimension $d$ such that	for all $\boldsymbol{\eta} = (\eta_1, \ldots, \eta_d)$ with $\eta_i \in W^{1,\infty}(\IR^d)$ and $f \in H^2(\IR^d)$,
	\begin{equation}\label{eq:localized-gradient2}
		\|(\boldsymbol\eta \cdot \nabla) f\|
		\;\le\;
		C_d \sum_{i=1}^d
		\Big(
		\sqrt{ \|\eta_i \, f\| \|\eta_i \,\Delta f\| }
		+\|(\nabla \eta_i)\, f\|
		\Big).
	\end{equation}

\end{lemma}
\begin{proof}
	Integration by parts yields
	\[
		\|\eta \nabla f\|^2
		= \int_{\mathbb{R}^d} \eta^2 |\nabla f|^2
		= -\Re \int_{\mathbb{R}^d} \eta^2 f \overline{\Delta f}
		- 2 \Re \int_{\mathbb{R}^d} \eta f \nabla \eta \cdot \overline{\nabla f}.
	\]
	Applying Cauchy--Schwarz,
	\[
		\|\eta \nabla f\|^2
		\le
		\|\eta f\|\|\eta \Delta f\|
		+ 2 \|\eta \nabla f\|\|(\nabla \eta) f\|.
	\]
	Using the inequality $2ab \le \tfrac12 a^2 + 2 b^2$ on the second term,
	we obtain
	\[
		\|\eta \nabla f\|^2
		\le
		\|\eta f\|\|\eta \Delta f\|
		+ \tfrac12 \|\eta \nabla f\|^2
		+ 2 \|(\nabla \eta) f\|^2.
	\]
	Rearranging terms yields
	\[
		\|\eta \nabla f\|^2
		\leq
		2\|\eta f\| \|\eta \Delta f\|  + 4 \|(\nabla \eta) f\|^2.
	\]
	Taking square roots gives \eqref{eq:localized-gradient}. The second bound \eqref{eq:localized-gradient2} follows from \eqref{eq:localized-gradient} using the triangle inequality.
\end{proof}

\section{Background}

\subsection*{Second Quantization Formalism and Rigorous Definition of the Hamiltonian}
\label{appx-sec:technical-details}

In this section we briefly recall the second quantization (Fock space) formalism in the continuum and provide a rigorous definition of our non-interacting fermionic many-body Hamiltonian in \Cref{eq:Hamiltonian physically}. For more details we refer the reader to \cite[Section 5.2.1.]{BR2} or \cite[Chapter 6]{Arai.2018}.

The Hilbert space for our model is the fermonic Fock space over $L^2(\IR^d)$,
\[ \FF \coloneqq \IC \oplus \bigoplus_{k=1}^\infty (L^2(\R^{d}))^{\wedge k}  =  \IC \oplus \bigoplus_{k=1}^\infty L^2_\sfa(\R^{d\cdot k}) \]
where $\wedge k$ denotes the $k$-times anti-symmetric tensor product and $L^2_\sfa(\R^{d\cdot k})$  the $L^2$-space where the anti-symmetrization in each summand occurs over the $k$ $d$-dimensional variables. Elements $\psi \in \FF$ can be written as sequences
\[
	\psi = (\psi^{(k)})_{k \in \IN_0 }, \qquad \psi^{(0)} \in \IC,~ \psi^{(k)} \in L^2_\sfa(\R^{d\cdot k}),~k\geq 1.
\]
The vacuum vector is defined as $\Omega = (1,0,0,0,\ldots)$.

Let $a(f)$, $f\in L^2(\R^d)$ be the fermionic annihilation operator
\[ a(f)(\psi^{(k)})_{k\in\IN_0} \coloneqq {\sqrt{k+1}}\left(\int \overline {f(x)} \psi^{(k+1)}(x,\cdots)\d x\right)_{k\in\IN_0}.  \]
This defines a bounded operator on $\FF$ with $\nn{a(f)}=\nn f$, $a(f) \Omega = 0$ and, denoting its adjoint as $a^*(f)\coloneqq a(f)^*$, we have the usual canonical anticommmutation relations (CAR)
\[ \{a(f),a(g)\}=\{a^*(f),a^*(g)\}=0, \quad \{a(f),a^*(g)\}=\braket{f,g} \qquad\mbox{for all}\ f\in  L^2(\R^d).\]
In particular, this implies $\nn{a^*(f)} = \nn f$ and
\begin{align}
	\label{eq:a products}
	\nn{ a^*(f_1) \cdots a^*(f_n) } = 1
\end{align}
for any orthonormal set $f_1, \ldots, f_n$. \Cref{eq:a products}. In particular, this implies that the initial states defined in \Cref{eq:psi0_defn} are actually normalized,
\begin{align}
	\label{eq:normalization}
	\norm{\prod_{j \in \Jf} (\Id +a^*(f^\alpha_{\j})  a^*(f^\beta_{\j})) \Omega}^2 = \norm{\sum_{\mathfrak{S} \subseteq \Jf} \prod_{\j \in \mathfrak{S}} a^*(f^\alpha_{\j})  a^*(f^\beta_{\j}) \Omega }^2 = \sum_{\mathfrak{S} \subseteq \Jf}  \norm{\prod_{\j \in \mathfrak{S}} a^*(f^\alpha_{\j})  a^*(f^\beta_{\j}) \Omega }^2 = \abs{\mathfrak{P}(\Jf)} = 2^{\abs{\Jf}},
\end{align}
where we expanded the product in the first equality: every factor is either $\Id$ ($\j \not\in \Sigma$) or $a^*(f^\alpha_{\j})  a^*(f^\beta_{\j})$ ($\j \in \Sigma$), and used the orthogonality of the states in the second equality.

Given a self-adjoint operator $A$ on $L^2(\R^d)$, we define its second quantization (lift to the Fock space) $\dG(A)$ on $\FF$ as
\[
	\dG(A) = \overline{ 0 \oplus \bigoplus_{k=1}^\infty \sum_{l=1}^{k}  \Id \otimes \cdots \otimes \underbrace{A}_{l-\text{th position}} \otimes \cdots \otimes \Id}
\]
where the bar denotes the self-adjoint closure.

Now, the non-interacting Hamiltonian \Cref{eq:Hamiltonian physically} can be rigorously defined an unbounded self-adjoint operator on $\FF$ by
\begin{align*}
	H = \dG(-\Delta + V).
\end{align*}

Heuristically, one can also write
\[
	a(f) = \int \overline{ f(x) } a_x \d x, \qquad a^*(f) = \int f(x) a^*_x \d x
\]
with the formal `pointwise operators' $a_x$, $a^*_x$ satisfying the following pointwise CAR:
\begin{align*}
	\{a_x, a_y\} = \{ a_x^*, a_y^*\} = 0, \qquad \{a_x,a^*_y\} = \delta(x-y).
\end{align*}
Using these relations one can formally compute the smeared-out CAR relations above.
Furthermore, one can easily derive the following commutation rules between $\dG(A)$ and $a(f)$,
\begin{align}
	\label{eq:af dG commutation relations}
	[a(f), \dG(A)]                          & = a(Af),                      \\
	e^{\i t \dG(A)} a^*(f) e^{-\i t \dG(A)} & = a^*(e^{i t A} f). \nonumber
\end{align}
They can be proven rigorously with the above definitions or can be seen heuristically with the pointwise CAR as well.

Next, note for the projector $P^{\alpha\beta}_\j$ defined in \Cref{eq:pvm},
\begin{align*}
	P^{\alpha\beta}_\j = \ketbra{\Omega}{\Omega}_{L^2(B_\j^\alpha \otimes B_\j^\beta)} \otimes \Id_{\text{Rest}},
\end{align*}
we use the decomposition of the fermionic Fock space (cf.~\cite[Theorem 6.21]{Arai.2018}),
\begin{align}
	\label{eq:fock decomposition}
	\FF\left( \bigoplus_{\j}  L^2(B_\j) \right) \cong \bigotimes_{\j} 	\FF\left( L^2(B_\j) \right),
\end{align}
and the projector to the vacuum only acts on the two tensor factors which are referenced by $\j,\alpha$ and $\j,\beta$.

Finally, we give a rigorous proof of the action of the discplacment operators on the vacuum,   \Cref{eq:controlled-displacement}.
\begin{lemma}
\label{th:controlled-displacement}
    Let $f,g \in L^2(\IR^d)$ be two normalized functions, which are perpendicular to each other, $\cs{f,g} = 0$. Then the operator $A = a^*(f) a^*(g) + a(g) a(f)$ is a self-adjoint bounded operator in $\FF$, so $e^{-\i \pi/4 A}$ is a well-defined operator as well. It acts on the vacuum as
    \begin{align*}
        e^{-\i \pi/4 A} \Omega =  2^{-1/2}\bigl(\Id  - \i  a^*(f) a^*(g)\bigr) \Omega.
    \end{align*}
\end{lemma}
\begin{proof}
    Consider $\psi_1 = (1 + a^*(f) a^*(g))\Omega$ and $\psi_2= (1 - a^*(f) a^*(g))\Omega$. From  the CAR, it is straightforward to see that $A \psi_1 = \psi_1$ and $A\psi_2 = - \psi_2$. Hence it follows with functional calculus
    	\begin{align*}
			e^{-\i \pi/4 A} \psi_1 =  e^{-\i \pi/4 } \psi_1  = \frac{1}{\sqrt 2}(1 - \i) \psi_1, \qquad    e^{-\i \pi/4 A} \psi_2 = e^{i \pi/4} \psi_2 = \frac{1}{\sqrt 2}(1 + \i) \psi_2
		\end{align*}
		and thus,
		\[
			e^{-\i \pi/4 A} \Omega = \frac{1}{2}e^{-\i \pi/4 A} (\psi_1 + \psi_2)  =  \frac{1}{2 \sqrt 2} ( (1-\i) \psi_1 + (1+\i) \psi_2) = \frac{1}{\sqrt 2} ( 1 -  \i a^*(f) a^*(g)) \Omega. \qedhere
		\]
\end{proof}

\section{Numerics for Coulomb Post-Processing}\label{appx-sec:numerics}

In this section, we present exemplary runs of the single and multi-Coulomb post-processing algorithms. Before explaining the setup, we recall that the corresponding code is available on GitHub and Zenodo under the following links: \href{https://github.com/MitTimM/learning_coulomb_potentials_in_continuous_space}{GitRepo} and \cite{DataZenodo.2026}.

For simplicity, we choose $\lambda \in (1,2)$ at random and $y \in [0,1]^3$ at random.

We test the post-processing algorithm on data generated by numerically computing the local averages $\omega_{\jmath}$ using Gauss–Legendre quadrature with between 24 and 40 nodes, depending on the desired numerical accuracy and computational performance. In addition, an adaptive quadrature scheme is implemented for computing the local averages. Gaussian noise is then added pointwise to simulate noisy data.

\begin{remark}
	The code presented in \href{https://github.com/MitTimM/learning_coulomb_potentials_in_continuous_space}{GitRepo} and \cite{DataZenodo.2026} is a direct realization of the algorithm analyzed in this paper. It is clear that the bounds achieved in the paper cannot be used for many applications in practice, which is why manual control in the code is allowed. Hence, better estimators for these bounds would be desirable. Moreover, a further speedup or adaptive error bounds are interesting technical directions for future work.
\end{remark}

\end{document}